\documentclass{lmcs}
\usepackage[utf8]{inputenc}
\pdfoutput=1

\usepackage{lastpage}
\lmcsdoi{18}{1}{13}
\lmcsheading{}{\pageref{LastPage}}{}{}%
{Mar.~30,~2020}{Jan.~19,~2022}{}


\usepackage{amsmath}
\usepackage{amsfonts}
\usepackage{amsthm}
\usepackage{enumitem}
\usepackage{tikz-cd}
\usepackage{stmaryrd}
\usepackage{bussproofs}
\usepackage{mathtools}
\usepackage{thm-restate}
\usepackage{tabularx}
\usepackage{booktabs}

\usepackage{color}

\theoremstyle{plain}
\newtheorem{nonexample}[thm]{Non-Example}
\newtheorem{counter}[thm]{Counterexample}

\newcommand{\define}[1]{\textbf{#1}}
\newcommand{\reason}[1]{\;\;\left\{\; #1 \;\right\}}
\newcommand{\bb}[1]{\mathbb{#1}}
\newcommand{\given}{\, | \,}

\newcommand{\anyterm}{\phi}
\newcommand{\anytermtwo}{\phi'}
\newcommand{\anytermthree}{\phi''}

\newcommand{\theoryeq}[1]{=_\mathbb{#1}}
\newcommand{\specialopS}{s}
\newcommand{\specialopT}{t}

\newcommand{\signatureT}{\Sigma^\bb{T}}
\newcommand{\eqT}{E^\bb{T}}
\newcommand{\essuniqfunctionA}{f}
\newcommand{\essuniqfunctionB}{f'}
\newcommand{\essuniqfunctionX}{h}
\newcommand{\essuniqfunctionY}{h'}
\mathchardef\mathhyphen="2D 

\usepackage{xspace}
\newcommand{\catset}{\mathbf{Set}}

\newcommand{\refsec}{Section}
\newcommand{\refsecs}{Sections}
\newcommand{\refthm}{Theorem}
\newcommand{\refthms}{Theorems}
\newcommand{\reflem}{Lemma}
\newcommand{\reflems}{Lemmas}
\newcommand{\refex}{Example}
\newcommand{\refexs}{Examples}
\newcommand{\refcounter}{Counterexample}

\newcommand{\refnon}{Non-Example}

\newcommand{\refcor}{Corollary}
\newcommand{\refeqn}{Equation}
\newcommand{\refeqns}{Equations}
\newcommand{\refdef}{Definition}

\newcommand{\refpro}{property}

\newcommand{\reftab}{Table}

\newcommand{\refprop}{Proposition}

\DeclareMathOperator{\var}{var}
\DeclareMathOperator{\ran}{ran}
\DeclareMathOperator{\supp}{supp}

\newcommand{\plotkinspecial}{\Xi}
\newcommand{\treemonad}{B}

\newcommand{\no}{\ensuremath{\times}\xspace}
\newcommand{\yes}{\ensuremath{\checkmark}\xspace}
\newcommand{\maybe}{\ensuremath{?}\xspace}

\begin{document}

\title[No-Go Theorems]{No-Go Theorems for Distributive Laws}
\titlecomment{{\lsuper*} This is an extended version of~\cite{ZwartMarsden2019}.}

\author[M.~Zwart]{Maaike Zwart}
\author[D.~Marsden]{Dan Marsden}
\address{Department of Computer Science\\ University of Oxford}
\email{maaike.annebeth@gmail.com, daniel.marsden@cs.ox.ac.uk}


\begin{abstract} 
Monads are commonplace in computer science, and can be composed using Beck's distributive laws. Unfortunately, finding distributive laws can be extremely difficult and error-prone. The literature contains some general principles for constructing distributive laws. However, until now there have been no such techniques for establishing when no distributive law exists.

We present three families of theorems for showing when there can be no distributive law between two monads. The first widely generalizes a counterexample attributed to Plotkin. It covers all the previous known no-go results for specific pairs of monads, and includes many new results. The second and third families are entirely novel, encompassing various new practical situations. For example, they negatively resolve the open question of whether the list monad distributes over itself, reveal a previously unobserved error in the literature, and confirm a conjecture made by Beck himself in his first paper on distributive laws. In addition, we establish conditions under which there can be at most one possible distributive law between two monads, proving various known distributive laws to be unique.
\end{abstract}

\maketitle

\section{Introduction}
Monads have become a key tool in computer science. They are, amongst other things, used to provide semantics for computational effects such as state, exceptions, and I/O~\cite{Moggi1991}. They are also used to structure functional programs~\cite{Wadler1995,PeytonJones2001}, and even appear explicitly in the standard library of the Haskell programming language~\cite{HaskellReport}. As such, it is important to fully understand and characterise their behaviour.

Monads are a categorical concept. A monad on a category~$\mathcal{C}$ is a triple $\langle T, \eta, \mu\rangle$ consisting of an endofunctor~$T:\mathcal{C} \rightarrow \mathcal{C}$ and two natural transformations~$\eta: 1 \Rightarrow T$ and~$\mu: T \circ T \Rightarrow T$ satisfying axioms described in \refdef~\ref{def:monad}. Given two monads with underlying functors~$S$ and~$T$, it is natural to ask if~$T \circ S$ always carries the structure of a monad. This would, for example, provide a way to combine simple monads together to model more complex computational effects.

Unfortunately, composing the functor parts of two monads does not, in general, result in a new monad. Beck has shown that the existence of a~\emph{distributive law} provides sufficient (but not necessary) conditions for such a composition to form a monad~\cite{Beck1969}. A distributive law between monads $S$ and $T$ is a natural transformation of type:
\[
S \circ T \Rightarrow T \circ S,
\]
satisfying four equations described in \refdef~\ref{def:distlaw} below. This important idea has since been generalised to notions of distributive laws for combining monads with comonads, monads with pointed endofunctors, endofunctors with endofunctors and various other combinations, see for example the work by Lenisa \emph{et al.}~\cite{LenisaPowerWatanabe2000}.

General-purpose techniques have been developed to construct distributive laws~\cite{Bonsangue2013, Jacobs1994, Dahlqvist2017b, ManesMulry2007,ManesMulry2008}. These methods are highly valuable, for, in the words of Bonsangue \emph{et al.}: \emph{``It can be rather difficult to prove the defining axioms of a distributive law.''}~\cite{Bonsangue2013}. In fact, it can be so difficult that on occasion a distributive law has been published which later turned out to be incorrect; Klin and Salamanca have made an overview of such cases involving the powerset monad~\cite{KlinSalamanca2018}.

The literature has tended to focus on positive results, either demonstrating specific distributive laws, or developing general-purpose techniques for constructing them. By comparison, there is a relative paucity of negative results, showing when no distributive law can exist. The most well-known result of this type appears in the paper \emph{Distributing probability over non-determinism}~\cite{VaraccaWinskel2006}, where it is shown that there is no distributive law combining the powerset and probability distribution monads, via a proof credited to Plotkin. This result was strengthened by Dahlqvist and Neves to show that the composite functor carries no monad structure at all~\cite{DahlqvistNeves2018}. Recently, the same proof technique was used by Klin and Salamanca to show that composing the covariant powerset functor  with itself yields an endofunctor that does not carry any monad structure~\cite{KlinSalamanca2018}, correcting an earlier error in the literature~\cite{ManesMulry2007}. To the best our knowledge, these are currently the only published impossibility results.

In this paper we present several theorems for proving the absence of distributive laws for large classes of monads on the category $\catset$ of sets and functions. These theorems significantly extend the current understanding of distributive laws.
Our results can roughly be divided into three classes:

\begin{itemize}
\item Firstly, we generalize Plotkin's method to general-purpose theorems covering all the previously
  published no-go results about distributive laws, and yielding new results as well.
\item Secondly, we develop a completely new family of no-go theorems, emphasising unitality axioms rather than the idempotence central to Plotkin's argument.
\item Finally, we prove two further no-go theorems, motivated by a question in Beck's original paper~\cite{Beck1969}. Here we introduce further novel techniques, based on the ability to judiciously ``make variables disappear'' in terms.
\end{itemize}

\noindent
As one application of our new methods, we show that the list monad cannot distribute over itself,
resolving an open question~\cite{ManesMulry2007, ManesMulry2008} and previous error~\cite{KingWadler1993} in the literature.
Another open question we resolve is a conjecture made by Beck in 1969~\cite{Beck1969}, stating that the Abelian group monad does not distribute over the list monad. In addition to answering open questions, the no-go theorems produced by these methods reveal yet another faulty distributive law in the literature, involving the list and exception monads.

Apart from our negative results, we also prove a uniqueness result. Under certain conditions, we prove that there is only one possible distributive law facilitating monad compositions, namely the distributive law resembling the times over plus distributivity. This proves that several known distributive laws, such as the distributive law for the multiset monad over itself, are unique.

Monads have deep connections with universal algebra. We fully embrace this algebraic viewpoint on monads, basing our proofs on an explicit algebraic equivalent of distributive laws formulated by Pir\'og and Staton~\cite{PirogStaton2017}, which was inspired by the work of Cheng~\cite{Cheng2011}. Formulating our results in algebraic form is a key contribution of our work, simplifying and clarifying the essentials of our proofs, which can be obscured by more direct calculations.

In physics, theorems proving the impossibility of certain things are called \emph{no-go theorems}, because they clearly identify theoretical directions that cannot succeed. We follow this example, and hope that by sharing our results, we prevent others from wasting time on forlorn searches for distributive laws that do not exist.

\subsection{Contributions}
We briefly outline our contribution. By taking an algebraic perspective, we demonstrate the non-existence of distributive laws for large classes of monads:
  \begin{itemize}
  \item In \refsec~\ref{section:Plotkin}  we widely generalize the essentials of a counterexample due to Plotkin~\cite{VaraccaWinskel2006}:
    \begin{itemize}
    \item We establish in \refthm~\ref{thm:plotkin-1} purely algebraic conditions under which a no-go result holds. This theorem recovers all the known negative results we are aware of, and many useful new results. The key ingredients for this theorem are binary terms that are idempotent and commutative.
    \item In~\refthm~\ref{thm:plotkin-2} we generalize \refthm~\ref{thm:plotkin-1} further, showing there is nothing essential about binary terms.
    \item \refthm~\ref{thm:plotkin-3} eliminates commutativity assumptions, yielding more useful insights.
    \item Finally, in \refsec~\ref{sec:plotkin-idempunit} we replace one of the idempotent terms from \refthm~\ref{thm:plotkin-1} with a unital one, still following the same general proof technique, yielding \refthm~\ref{thm:NoGo-Videm-Punit}.
    \end{itemize}
\item In \refsec~\ref{section:beyondPlotkin} we present three entirely new no-go theorems:
     \begin{itemize}
     \item Theorem~\ref{thm:nogoConstants} states conditions under which algebraic theories with more than one constant do not combine well with other theories. The usefulness of this theorem is demonstrated by \refex~\ref{ex:error-in-the-literature}, which identifies a previously unnoticed error in a theorem in the literature.
     \item Theorem~\ref{thm:nogoTreeList} is a general no-go theorem for monads that do not satisfy the so-called \emph{abides} equation~\cite{Bird1988}. One of its applications is that it negatively answers the open question of whether the list monad distributes over itself.
     \item Theorem~\ref{thm:nogoTreeIdem} is a general no-go theorem focussing on the combination of idempotence and units. From this theorem it follows that there is no distributive law for the powerset monad over the multiset monad: $P \circ M \Rightarrow M \circ P$.
     \item Section~\ref{section:beyondPlotkin} also contains two characterization results, giving insights into how a distributive law has to behave, if it exists at all. Proposition~\ref{propnary} identifies when the unit of one algebraic theory, presenting a monad, acts as an annihilator when another monad is distributed over it. For example, this characterizes the behaviour of empty lists and sets in certain composite monads. Theorem~\ref{thm:times-over-plus} gives conditions under which at most one distributive law can exist, and what form it must take. For example, the well-known distributive law for the multiset monad over itself is unique.
     \end{itemize}
   \item In \refsec~\ref{sec:plus-over-times} a second new family of no-go theorems is presented,
     based on another new approach to yielding no-go results within our algebraic framework:
     \begin{itemize}
     \item Theorem~\ref{theoremwithconst} confirms a negative conjecture of Beck from the original paper on distributive laws~\cite{Beck1969}.
     \item Theorem~\ref{theoremwithvar} provides further new results, showing there is no distributive law $bL \circ P \Rightarrow P \circ bL$ between the powerset and bounded distributive lattice monads.
     \end{itemize}
   \item In \refsec~\ref{ch:boom} we provide a detailed analysis of the availability of distributive laws for some natural families of monads.
     This includes many examples of applications of our theorems, as well as results from elsewhere in the literature to provide a detailed picture.
     \begin{itemize}
     \item We examine distributive laws between monads in the so-called~Boom hierarchy,
       a small family of monads representing data structures, which are studied in the functional programming community.
     \item To gain more examples of existing/non-existing distributive laws, we consider an extension of the Boom hierarchy previously studied by Uustalu~\cite{UUSTALU2016}.
     \item Finally, we consider possible iterations of compositions in the Boom hierarchy.
     \end{itemize}
   \item In \refsec~\ref{sec:conclusion}, we provide a summary of all axioms used for the various theorems in this paper, together with an overview of which theorem uses which axioms.
\end{itemize}

\noindent
Throughout this paper we restrict our attention to monads on the category $\catset$ of sets and functions, as this is already an incredibly rich setting. Our results are general-purpose in the sense that they are phrased in terms of abstract properties of the algebraic theories corresponding to both monads.
\begin{rem}
   To compose monads~$\langle S,\eta,\mu \rangle$ and~$\langle T,\eta,\mu\rangle$ one may ask if all we really want is some arbitrary monad structure on the functor~$TS$, rather than the specific structure given by a distributive law? Generally, a monad structure arising from a distributive law is vastly preferable to an arbitrary one, as it has many desirable properties. For example, there are canonical monad morphisms~$S \Rightarrow T \circ S$ and~$T \Rightarrow T \circ S$ inducing functors between both the Eilenberg-Moore and Kleisli categories of the corresponding monads. Furthermore, $T$ lifts to a monad on the Eilenberg-Moore category of~$S$, and~$S$ lifts to a monad on the Kleisli category of~$T$. More succinctly, a distributive law ensures that there is a strong relationship between the parts and the whole.
\end{rem}

\subsection{Additional Material}
This paper is an extended version of our LiCS2019 conference paper~\cite{ZwartMarsden2019}. Besides providing detailed proofs and more discussion of the original results, the following additions have been made:
\begin{itemize}
\item Section~\ref{sec:plotkin-idempunit} extends our Plotkin style counterexamples with a new theorem, bringing more monads within the scope of our techniques.
\item \refsec~\ref{sec:plus-over-times} introduces a third family of theorems, not present in the original conference paper. These theorems preclude the possibility of distributive laws between additional naturally occurring monads.
\item \refsec~\ref{ch:boom} provides detailed analysis of when distributive laws are available for many common monads. These concrete applications further illustrate the use of our main theorems, and situate our results in the broader understanding of distributive laws in the community.
\item \refsec~\ref{sec:conclusion} includes an overview of the axioms used in this paper, and which theorem uses which axioms.
\end{itemize}

\section{Preliminaries}\label{sec:prelims}

\subsection{Monads and Distributive Laws}
We introduce monads, distributive laws, and various examples that will recur in later sections, primarily to fix notation. The material is standard, and may be skipped by the expert reader.
\begin{defi}[Monad]%
  \label{def:monad}
For any category $\mathcal{C}$, a \define{monad} $\langle T,\eta, \mu \rangle$ on $\mathcal{C}$ consists of an endofunctor $T: \mathcal{C} \rightarrow \mathcal{C}$, and natural transformations~$\eta: 1 \Rightarrow T$ and~$\mu: T \circ T \Rightarrow T$ referred to as the~\define{unit} and~\define{multiplication}, satisfying the following axioms:
\begin{align}
  \mu \cdot T\eta & = id \tag{unit 1}\\
  \mu \cdot \eta T & = id \tag{unit 2}\\
  \mu \cdot T\mu & = \mu \cdot \mu T\tag{associativity}
\end{align}
Or, as commuting diagrams:
 
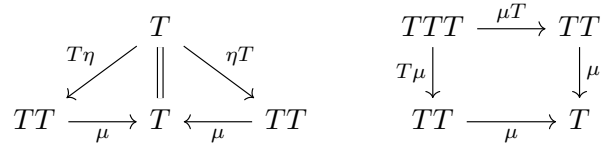
\begin{center}
  \begin{tikzcd}
   & T \arrow[rd, "\eta T"]\arrow[ld, "T\eta"'] \arrow[d, equal] & & TTT \arrow[r, "\mu T"]\arrow[d, "T \mu"'] & TT \arrow[d, "\mu"] \\
   TT \arrow[r, "\mu"'] & T & TT \arrow[l, "\mu"] & TT \arrow[r, "\mu"'] & T
  \end{tikzcd}
 \end{center}
\end{defi}
We will restrict ourselves to monads on the category $\catset$ of sets and functions. In addition, if there is a finitary version and a full version of a monad, we mean the finitary one unless otherwise specified. We list a few examples of monads on $\catset$, which we will use throughout this paper.
\begin{exa}%
  \label{ex:exception-monad}
  For any set $E$, the~\define{exception monad} $(-+E)$ is given by:
  \begin{itemize}
  \item $(-+E)$ maps a set $X$ to the disjoint union $X+E$.
  \item $\eta^E_X$ is the left inclusion morphism.
  \item $\mu^E_X$ is the identity on $X$, and collapses the two copies of $E$ down to a single copy. That is, $\mu: (X + E) + E \Rightarrow X + E$.
  \end{itemize}
  When $E$ is a singleton set, this monad is also known as the~\define{maybe monad}, written as~${(-)}_\bot$.
\end{exa}
\begin{exa}%
  \label{ex:list-monad}
  The~\define{list monad}~$L$ is given by:
  \begin{itemize}
  \item $L(X)$ is the set of all finite lists of elements of~$X$.
  \item $\eta^L_X(x)$ is the singleton list $[x]$.
  \item $\mu^L_X$ concatenates a list of lists.
  \end{itemize}
  This monad is also known as the \define{free monoid monad}, in acknowledgement of its connection to the algebraic theory of monoids, see \refex~\ref{ex:monoids} below.
\end{exa}
\begin{exa}%
  \label{ex:multiset-monad}
  The~\define{multiset monad}~$M$ is given by:
  \begin{itemize}
  \item $M(X)$ is the set of all finite multisets\footnote{by which we mean: multisets in which only finitely many elements have a non-zero multiplicity. In other words: multisets with finite support. We assume the multiplicities are in the natural numbers. Multisets are also known as `bags'.} of elements of~$X$.
  \item $\eta^M_X(x)$ is the singleton multiset $\Lbag x \Rbag$.
  \item $\mu^M_X$ takes a union of multisets, adding multiplicities.
  \end{itemize}
  We can generalise the notion of multiset to take multiplicities in the integers rather than the natural numbers. This results in the~\define{Abelian group monad}. This monad is again named after its algebraic presentation, see \refprop~\ref{prop:presentation}.
\end{exa}
\begin{exa}%
  \label{ex:powerset-monad}
  The~\define{finite powerset monad}~$P$ is given by:
  \begin{itemize}
  \item $P(X)$ is the set of all finite subsets of~$X$.
  \item $\eta^P_X(x)$ is the singleton set~$\{x\}$.
  \item $\mu^P_X$ takes a union of sets.
  \end{itemize}
\end{exa}
\begin{exa}[Binary Tree Monad]%
  \label{ex:tree-monad}
  The~\define{binary tree monad}~$\treemonad$ is given by:
  \begin{itemize}
  \item $\treemonad(X)$ is the set of all binary trees with leaves labelled by elements from $X$.
  \item $\eta^{\treemonad}_X(x)$ is the tree consisting of a single leaf labelled with $x$.
  \item $\mu^{\treemonad}_X$ flattens a tree of trees.
  \end{itemize}
\end{exa}
\begin{exa}%
  \label{ex:distribution-monad}
  The~\define{probability distribution monad} $D$ is given by:
  \begin{itemize}
  \item $D(X)$ is the set of all finitely supported probability distributions over~$X$.
  \item $\eta^{D}_X(x)$ is the point distribution at~$x$.
  \item $\mu^{D}(e)(x)$ is the weighted average~$\sum_{d \in \supp(e)} e(d)d(x)$.
  \end{itemize}
\end{exa}
\begin{exa}%
  \label{ex:reader-monad}
  For any set of states~$R$, the~\define{reader monad}~${(-)}^R$ is given by:
  \begin{itemize}
  \item $X^R$ is the set of functions from~$R$ to~$X$.
  \item $\eta^{R}_X(x)$ is constantly~$x$.
  \item $\mu^R_X(f)(r) = f(r)(r)$.
  \end{itemize}
\end{exa}

\noindent
Given a pair of monads, sufficient conditions for the composite functor to form a monad are given by Beck's distributive laws~\cite{Beck1969}:
\begin{defi}[Distributive Law]\label{def:distlaw}
Given monads $S$ and $T$, a \define{distributive law} for monad composition $T\circ S$ is a natural transformation $\lambda: S \circ T \Rightarrow T \circ S$ satisfying the following axioms:
\begin{align}
  \lambda \cdot \eta^S T & = T\eta^S \tag{unit1}\\
  \lambda \cdot S\eta^T  & = \eta^T S \tag{unit2} \\
  \lambda \cdot \mu^S T & = T\mu^S \cdot \lambda S \cdot S\lambda \tag{multiplication1} \\
  \lambda \cdot S\mu^T  & = \mu^T S \cdot T\lambda \cdot \lambda T \tag{multiplication2}
\end{align}
Or, as commuting diagrams:
\begin{center}
\begin{tikzcd}\label{monad_squared_eq}
& T \arrow[dl, "\eta^S T"']\arrow[dr, "T \eta^S"] & & SST \arrow[d, "\mu^S T"'] \arrow[r, "S \lambda"] & STS \arrow[r, "\lambda S"] & TSS \arrow[d, "T \mu^S"] \\
ST \arrow[rr, "\lambda"]& & TS & ST \arrow[rr,"\lambda"] & & TS \\
& S \arrow[dl, "S\eta^T"'] \arrow[dr, "\eta^T S"] & & STT \arrow[d, "S \mu^T"'] \arrow[r, "\lambda T"] & TST \arrow[r, "T\lambda"] & TTS \arrow[d, "\mu^T S"] \\
 ST \arrow[rr, "\lambda"] & & TS & ST \arrow[rr,"\lambda"] & & TS \\
\end{tikzcd}
\end{center}
\end{defi}
\begin{rem}
For a pair of monads~$S, T$ the expression~``$S$ distributes over~$T$'' is often used. This phrasing is somewhat ambiguous and prone to errors. We will therefore explicitly state the type of the natural transformation, for example~``there is a distributive law of type~$S \circ T \Rightarrow T \circ S$''.
\end{rem}
\begin{thm}[Beck~\cite{Beck1969}]
  Let~$\mathcal{C}$ be a category, and~$\langle S, \eta^S, \mu^S \rangle$ and~$\langle T, \eta^T, \mu^T \rangle$ two monads on~$\mathcal{C}$. If~$\lambda: S \circ T \Rightarrow T \circ S$ is a distributive law, then~$T \circ S$ carries a monad structure with unit~$\eta^T\eta^S$ and multiplication~$\mu^T\mu^S \cdot T\lambda S$.
\end{thm}
\begin{exa}[Ring Monad~\cite{Beck1969}]%
  \label{ex:ring-monad}
  The motivating example of a distributive law involves the list monad and the Abelian group monad, and has type $\lambda: L \circ A \Rightarrow  A \circ L$. It captures exactly the distributivity of multiplication over addition:
  \begin{equation}%
  \label{timesoverplus}
  \lambda \left( a \cdot (b + c) \right) = (a \cdot b) + (a \cdot c)
  \end{equation}
  The term `distributive law' is derived from this example, and many other distributive laws exploit similar algebraic properties.
  However, as we will see in~\refsec~\ref{section:beyondPlotkin}, caution is needed: the validity of an equation such as~\eqref{timesoverplus} does not automatically imply the existence of a distributive law.
\end{exa}
\begin{exa}[Multiset Monad]\label{ex:multisetdistlaw}
  The multiset monad distributes over itself in a manner analogous to distributing multiplication over addition.
\end{exa}
\subsection{Algebraic Theories and Composite Theories}
We now outline the connections between algebras, monads, and distributive laws that we require in later sections.
\begin{defi}[Algebraic Theory]
  An~\define{algebraic signature} is a set of operation symbols~$\Sigma$, each with an associated natural number referred to as its~\define{arity}. The set of~\define{$\Sigma$-terms} over a set~$X$ contains~$X$ as~\define{variables} and is inductively closed under forming terms~$\sigma(t_1,\dots,t_n)$ for an~$n$-ary operation symbol~$\sigma$ and terms~$t_1,\dots,t_n$.

  An~\define{algebraic theory}~$\bb{T}$ consists of a signature~$\signatureT$, and a set~$\eqT$ of pairs of $\Sigma$-terms referred to as~\define{equations} or~\define{axioms}. We will often write a pair~$(s,t) \in \eqT$ as~$s \theoryeq{T} t$ or simply $s = t$ when convenient. For a subset $Y \subseteq X$ and a term $t$ we write $Y \vdash_\bb{T} t$ or $Y \vdash t$ to indicate that the variables appearing in $t$ are contained in the~\define{variable context} $Y$. The precise set of variables appearing in $t$ will be denoted $\var(t)$, and $\# \var(t)$ denotes the cardinality of this set.
\end{defi}
The following is well-known~\cite{Linton1966,Lawvere1963,Manes1976}:
\begin{prop}\label{prop:presentation}
  Given a theory~$(\signatureT, \eqT)$, the~\define{free model monad} over that signature maps~$X$ to the set of $\signatureT$-terms over~$X$, quotiented by provable equality in equational logic from the axioms~$\eqT$. The unit maps a variable to its corresponding equivalence class, and the multiplication flattens a term-of-terms to a term in the obvious way. If a monad is isomorphic to a free model monad, it is said to be~\define{presented} by the corresponding theory.
\end{prop}
\begin{exa}[Monoids]%
  \label{ex:monoids}
  The algebraic \define{theory of monoids} has a signature containing a constant and a binary operation, satisfying left and right unitality and associativity.
  The \define{theory of commutative monoids} extends this theory with the commutativity equation.
  The \define{theory of join semilattices} further extends the theory of commutative monoids with an additional idempotence axiom.

  The corresponding free model monads are the list, multiset and finite powerset monads respectively.
\end{exa}
In the paper \emph{Notions of computation determine monads}, Plotkin and Power show that many monads describing computational effects have natural algebraic presentations~\cite{PlotkinPower2002}.
\begin{exa}%
  \label{ex:reader-monad-presentation}
  An algebraic presentation of the reader monad of~\refex~\ref{ex:reader-monad}, with state space~$\{0,1\}$, has a signature containing a single binary operation. Intuitively, $x * y$ is a process that proceeds as~$x$ if the state is~$0$, and~$y$ otherwise. This operation should satisfy:
  \begin{align*}
    x * x = x \quad\text{and}\quad
    (w * x) * (y * z) = w * z
  \end{align*}
  These axioms generalize naturally to larger state spaces.
\end{exa}
\begin{exa}%
  \label{ex:distribution-monad-presentation}
The distribution monad \refex~\ref{ex:distribution-monad} can be presented by a family of binary operations~$+^p$, for~$p \in (0,1)$, satisfying the following axioms~\cite{Jacobs2010,Stone1949}:
    \begin{align*}
      x +^p x & = x\\
      x +^p y &= y +^{1-p} x \\
      x +^p (y +^r z) &= (x +^{\frac{p}{p + (1-p)r}} y) +^{p +(1-p)r} z
    \end{align*}
\end{exa}

In addition to an algebraic equivalent of monads, we require an algebraic version of distributive laws. This key notion is provided in the form of \emph{composite theories}, which were introduced by Pir\'og and Staton~\cite{PirogStaton2017}.
\begin{defi}[Composite Theory]\label{def:composite}
Let $\bb{U}$ be an algebraic theory that contains two theories $\bb{S}$ and $\bb{T}$.
\begin{enumerate}
\item A term in $\bb{U}$ is \define{separated} if it is of the form $t[s_x/x]$, where $X \vdash_\bb{T} t$ is a $\bb{T}$-term and $s_x$ is a family of $\bb{S}$-terms indexed by $x \in X$.
\item Two separated terms $r_1 = t[s_x/x]$ and~$r_2 = t'[s'_{x'}/x']$ in $\bb{U}$ are~\define{equal modulo~$(\bb{T},\bb{S})$} if there are functions
$\essuniqfunctionX: X \rightarrow Y, \essuniqfunctionY: X' \rightarrow Y$ and terms $\bar{s}_y$, such that:
    \begin{enumerate}
    \item $t[\essuniqfunctionX(x)/x] \theoryeq{T} t'[\essuniqfunctionY(x')/x']$
    \item $\forall x \in X : s_x \theoryeq{S} \bar{s}_{\essuniqfunctionX(x)}$
    \item $\forall x' \in X' : s'_{x'} \theoryeq{S} \bar{s}_{\essuniqfunctionY(x')}$
    \end{enumerate}
\end{enumerate}
The theory $\bb{U}$ is said to be a~\define{composite} of $\bb{T}$ after $\bb{S}$, if every term $u$ in~$\bb{U}$ is equal to a separated term, and moreover this term is~\define{essentially unique} in the sense that if~$v,v'$ are separated and~$v \theoryeq{U} u \theoryeq{U} v'$ then~$v$ and~$v'$ are equal modulo~$(\bb{T},\bb{S})$. Note that this is an oriented notion, a composite of~$\bb{T}$ after~$\bb{S}$ is not equivalent to a composite of~$\bb{S}$ after~$\bb{T}$.
\end{defi}

Note that a decomposition of a separated term $r$ in $\bb{U}$ into its $\bb{T}$ and $\bb{S}$ components, $r = t[s_x/x]$, is not necessarily unique: the term $t(s, s)$ can be written both as $t(x,y)[s/x, s/y]$ and $t(x,x)[s/x]$. We prove that all possible decompositions of a term are necessarily equal modulo~$(\bb{T},\bb{S})$, and therefore the possible choice in decomposition in the definition above does not matter.

\begin{lem}
If two separated terms $t[s_x/x]$ and $t'[s'_{x'}/x']$ are syntactically equal to each other, $t[s_x/x] = t'[s'_{x'}/x']$, then they are equal modulo~$(\bb{T},\bb{S})$. Moreover, all equalities in conditions $(a)$, $(b)$, and $(c)$ above are syntactic equalities in this case.
\end{lem}

\begin{proof}
Suppose that $t[s_x/x] = t'[s'_{x'}/x']$, where `$=$' denotes syntactic equality. We consider the syntactic trees of these terms, which must be equal. We conclude that the terms $t$ and $t'$ have the same syntactic tree, except for possibly their variables. That means that there are variable substitutions $\essuniqfunctionX: X \rightarrow Y, \essuniqfunctionY: X' \rightarrow Y$ such that $t[\essuniqfunctionX(x)/x] = t'[\essuniqfunctionY(x')/x']$, which proves condition (a).

Furthermore, again looking at the syntactic tree $t[s_x/x]$ and $t'[s'_{x'}/x']$, we see that the $s_x$ and $s_{x'}$ that appear in the same place in the tree must be syntactically equal. Hence we can choose the terms $\bar{s}_y$ to be those terms that already appear in the tree. That is, $\bar{s}_{\essuniqfunctionX(x)} = s_x$ and $\bar{s}_{\essuniqfunctionY(x')} = s_{x'}$. This proves conditions (b) and (c).
\end{proof}

In this paper, we shall use an equivalent definition of equality modulo~$(\bb{T},\bb{S})$.
\begin{thm}\label{thm:essentialuniqueness}
  Let $\bb{S}$ and $\bb{T}$ be two algebraic theories, and let $\bb{U}$ be an algebraic theory that contains both $\bb{S}$ and $\bb{T}$. For terms $X \vdash_\bb{T} t$, $X'\vdash_{\bb{T}} t'$ and families of $\bb{S}$-terms $\{s_x \;|\; x \in X\}$ and $\{s'_{x'} \;|\; x' \in X'\}$, the following are equivalent:
  \begin{enumerate}
  \item\label{it:es} The terms~$t[s_x/x]$ and~$t'[s'_{x'}/x']$ are equal modulo~$(\bb{T},\bb{S})$.
  \item\label{it:strong-es} There are functions $\essuniqfunctionA: X \rightarrow Z$, $\essuniqfunctionB: X' \rightarrow Z$ satisfying:
    \begin{enumerate}
    \item\label{ax:essuniq1}$t[\essuniqfunctionA(x)/x] \theoryeq{T} t'[\essuniqfunctionB(x')/x']$.
    \item\label{ax:essuniq2}$\essuniqfunctionA(x_1)\; = \essuniqfunctionA(x_2) \;\,\Leftrightarrow\; s_{x_1}\theoryeq{S} s_{x_2}$.
    \item\label{ax:essuniq3}$\essuniqfunctionB(x'_1) = \essuniqfunctionB(x'_2) \;\Leftrightarrow\; s'_{x'_1}\theoryeq{S} s'_{x'_2}$.
    \item\label{ax:essuniq4}$\;\,\essuniqfunctionA(x)\; = \essuniqfunctionB(x') \;\,\Leftrightarrow\;\; s_{x}\, \theoryeq{S} s'_{x'}$.
    \end{enumerate}
  \end{enumerate}
\end{thm}
\begin{proof}
  Showing that condition~\ref{it:strong-es} implies condition~\ref{it:es} is straightforward. Taking~$Y$ to be the union of the ranges of $\essuniqfunctionA$ and $\essuniqfunctionB$, requirements~\ref{ax:essuniq2}-\ref{ax:essuniq4} ensure that we can choose~$\bar{s}_y$ such that:
  \[
  \bar{s}_y = \begin{cases}
                s_{x}, & \mbox{if } \essuniqfunctionA(x) = y \\
                s_{x'}, & \mbox{if } \essuniqfunctionB(x') = y.
              \end{cases}
  \]

  To show that condition~\ref{it:es} implies condition~\ref{it:strong-es}, notice that by transitivity of equality, $\essuniqfunctionX$ and $\essuniqfunctionY$ from \refdef~\ref{def:composite} already have the properties:
      \begin{enumerate}
        \item\label{ax:essuniq1prime}$t[\essuniqfunctionX(x)/x] \theoryeq{T} t'[\essuniqfunctionY(x')/x']$.
        \item\label{ax:essuniq2prime}$\essuniqfunctionX(x_1)\; = \essuniqfunctionX(x_2) \;\Rightarrow\; s_{x_1} \theoryeq{S} s_{x_2}$.
        \item\label{ax:essuniq3prime}$\essuniqfunctionY(x'_1) = \essuniqfunctionY(x'_2) \Rightarrow\; s'_{x'_1} \theoryeq{S} s'_{x'_2}$.
        \item\label{ax:essuniq4-prime}$\;\essuniqfunctionX(x)\;\, = \essuniqfunctionY(x') \,\Rightarrow\;\, s_{x}\, \theoryeq{S} s'_{x'}$.
      \end{enumerate}
  So all we need to show are the reverse implications of the latter three points. To this end, pick a function $g: Y \rightarrow Z$ such that:
  \begin{itemize}
  \item If $s_{x_1} \theoryeq{S} s_{x_2}$, then $g(\essuniqfunctionX(x_1)) \; = g(\essuniqfunctionX(x_2))$
  \item If $s'_{x'_1} \theoryeq{S} s'_{x'_2}$, then $g(\essuniqfunctionY(x'_1)) = g(\essuniqfunctionY(x'_2))$
  \item If $\;s_{x}\, \theoryeq{S} s'_{x'}$, then $\;\,g(\essuniqfunctionX(x))\,\, = g(\essuniqfunctionY(x'))$
  \end{itemize}
  The function $g$ effectively takes a quotient of $Y$, defined by the three conditions above. Transitivity of equality ensures that the compositions $\essuniqfunctionA = g \circ \essuniqfunctionX$ and $\essuniqfunctionB = g \circ \essuniqfunctionY$ preserve properties~\ref{ax:essuniq1prime}-\ref{ax:essuniq4-prime}. By definition, they also satisfy the reverse implications, and so they satisfy condition~\ref{it:strong-es}.
\end{proof}

We may assume that a composite theory is never inconsistent:
\begin{prop}\label{prop:consistency}
If $\bb{U}$ is a composite theory of theories $\bb{T}$ after $\bb{S}$, and both $\bb{S}$ and $\bb{T}$ are consistent, then $\bb{U}$ is consistent.
\end{prop}
\begin{proof}
  Suppose for contradiction that $\bb{U}$ is an inconsistent theory. Then for each pair of variables $x,y$ we have $x \theoryeq{U} y$. Since both $x$ and $y$ are separated terms, essential uniqueness gives us two substitutions $f: \{x\} \rightarrow Z$ and $g:\{y\} \rightarrow Z$ such that:
  \begin{align*}
  x[f] \theoryeq{T} y[g] & \\
  f(x) = g(y) & \;\;\Leftrightarrow\;\;  x \theoryeq{S} y
  \end{align*}
  For the first equation to be satisfied without violating the consistency of $\bb{T}$, we must have that $f(x) = g(y)$. This implies, however, that $x \theoryeq{S} y$, which contradicts the consistency of $\bb{S}$. We hence conclude that if $\bb{U}$ is a composite theory of $\bb{T}$ after $\bb{S}$, $\bb{U}$ must be consistent.
\end{proof}

The following theorem shows that composite theories are indeed the algebraic equivalent of distributive laws.
\begin{thm}[Pir\'og \& Staton~\cite{PirogStaton2017}]%
  \label{thm:distlaw-vs-compositetheory}
  Let~$S$ and~$T$ be $\catset$-monads presented by theories~$\bb{S}$ and~$\bb{T}$. There is a distributive law of type~$S\circ T \Rightarrow T \circ S$ if and only if there is a composite theory of~$\bb{T}$ after~$\bb{S}$.
\end{thm}
We will frequently exploit~\refthm~\ref{thm:distlaw-vs-compositetheory} by showing that no composite theory exists, and therefore no distributive law.

For positive results it will be useful to know the actual action of the distributive law promised by~\refthm~\ref{thm:distlaw-vs-compositetheory}, if we know the composite theory. Conversely, when a distributive law is known to exist, it is useful to have an algebraic presentation for the resulting composite monad. The following theorems provide us with just that. Similar observations have been made by Lack~\cite[Proposition 4.7]{Lack2004}, written up more explicitly by Zanasi in his thesis~\cite[Proposition 2.27]{Zanasi2015}. Their results hold for symmetric monoidal theories.

\begin{prop}\label{thm:distlaw-from-compth}
 Let $\bb{S}$ and $\bb{T}$ be algebraic theories presenting monads $S$ and $T$, and let $\bb{U}$ be a composite theory of $\bb{T}$ after $\bb{S}$. Then the free model monad $U$ is isomorphic to the composition of $T\circ S$ via the distributive law mapping the equivalence class of representative $s[t_x/x]$ to the suitable equivalence class of a separated term in $\bb{U}$ equal to $s[t_x/x]$.
\end{prop}
The proof of this proposition is straightforward but tedious. We give a sketch of the main ideas below, but for full details we refer to~\cite{Zwart2020}.
\begin{proof}[Proof sketch]
We first establish that the functor $U$ is isomorphic to the functor $T \circ S$ by finding an explicit natural isomorphism $\phi: U \Rightarrow T \circ S$. We conclude that the functor $T \circ S$ has a monad structure, which is given by the monad structure of the free model monad $U$ of theory $\bb{U}$. To prove that this structure comes from a distributive law, we prove the following three statements. Together, these statements are equivalent to having a distributive law~\cite{Beck1969}:
\begin{itemize}
\item $\langle TS, \eta^T\eta^S, \mu^{TS}\rangle$ is a monad, with multiplication $\mu^{TS} = \phi \cdot \mu^U \cdot (\phi^{-1}) TS(\phi^{-1})$.
\item The natural transformations $\eta^T S$ and $T\eta^S$ are monad maps.
\item The middle unitary law holds: $\mu^{TS}\cdot T\eta^S\eta^{T}S = id_{TS}$.
\end{itemize}

\noindent
The distributive law is then given by: $\lambda = \mu^{TS} \cdot \eta^T ST \eta^S$
\end{proof}

Knowing how to construct a distributive law from a composite theory, we will now do the reverse: constructing a composite theory from a known distributive law. To do this, we first define a set of equations derived from the distributive law.

\begin{defi}
Let $S$ and $T$ be the free model monads of algebraic theories $\bb{S}$ and $\bb{T}$. If there is a distributive law $\lambda: S \circ T \Rightarrow T \circ S$, then we define the set of $\lambda$-equations \define{$E^\lambda$} as follows:
Let $s[t_x/x]$ be a representative of an element in $STX$, and $t[s_y/y]$ a representative of an element in $TSX$. Then $s[t_x/x] = t[s_y/y] \in E^\lambda$ iff $\lambda$ maps the equivalence class of $s[t_x/x]$ to the equivalence class of $t[s_y/y]$.
\end{defi}

\begin{prop}\label{thm:theoryfromdistlaw}
Let $S$ and $T$ be the free model monads of algebraic theories $\bb{S}$ and $\bb{T}$. If there is a distributive law $\lambda: S \circ T \Rightarrow T \circ S$, then the following theory is a composite of $\bb{T}$ after $\bb{S}$, and the monad $\langle TS, \eta^T\eta^S, \mu^T\mu^S \cdot T\lambda S\rangle$ is the free model monad of this theory:
\begin{align*}
\Sigma^{\bb{TS}^\lambda} & = \Sigma^\bb{S} \uplus \Sigma^\bb{T} \\
E^{\bb{TS}^\lambda} & = E^S \cup E^T \cup E^\lambda
\end{align*}
We call the composite theory hence constructed $\bb{TS}^\lambda$.
\end{prop}

Again the proof is straightforward but tedious. We give a proof sketch below, but for full details we refer to~\cite{Zwart2020}.
\begin{proof}[Proof sketch]
 From Pir\'og and Staton's theorem~\cite[Theorem 5]{PirogStaton2017}, we already know that $\lambda$ induces a monad that is the free model monad of a composite of $\bb{T}$ after $\bb{S}$. What is left to show is that $\bb{TS}^\lambda$ is indeed this composite theory. We prove this by establishing an isomorphism between the Eilenberg-Moore category of the monad TS and the category of algebras for $\bb{TS}^\lambda$.
\end{proof}

Propositions~\ref{thm:distlaw-from-compth} and~\ref{thm:theoryfromdistlaw} together give us a concrete presentation for composite theories:
\begin{cor}\label{cor:presentationforcomposite}
 Let $\bb{S}$ and $\bb{T}$ be algebraic theories presenting monads $S$ and $T$, and let $\bb{TS}$ be a composite theory of $\bb{T}$ after $\bb{S}$. Then the following gives a presentation of $\bb{TS}$:
\begin{align*}
\Sigma^{\bb{TS}} & = \Sigma^\bb{S} \uplus \Sigma^\bb{T} \\
E^{\bb{TS}^\lambda} & = E^\bb{S} \cup E^\bb{T} \cup E^\bb{\lambda}
\end{align*}
where $E^\bb{\lambda}$ consists of all provable equations in $\bb{TS}$ of form $s[t_x/x] = t[s_y/y]$.
\end{cor}

\section{General Plotkin Theorems}\label{section:Plotkin} 

In this section we develop algebraic generalizations of the counterexample attributed to Gordon Plotkin~\cite{VaraccaWinskel2006}, which showed that there is no distributive law of type $D \circ P \Rightarrow P \circ D$,
where $D$ is the distribution monad of \refex~\ref{ex:distribution-monad} and~$P$ is the finite powerset monad of \refex~\ref{ex:powerset-monad}. We present a slight rephrasing of this counterexample, augmented with commentary indicating the main proof ideas that will be used in the later generalizations. Throughout this section we adopt the notational conventions of~\cite{VaraccaWinskel2006,PirogStaton2017} whenever possible to ease comparison with those papers.
\begin{counter}[Probability does not distribute over non-determinism]\label{counter:Plotkin-original}
  Assume, for contradiction, that there is a distributive law of type~$\lambda : D \circ P \Rightarrow P \circ D$. Fix the set~$X = \{ a, b, c, d \}$, and consider the element~$\plotkinspecial \in DP(X)$ defined by:
  \[
    \plotkinspecial = \{ a, b \} +^{\frac{1}{2}} \{c, d \}.
  \]
  We define three functions $f_1, f_2, f_3 : X \rightarrow X$:
  \begin{alignat*}{3}
    & f_1(a) = a &&\quad f_2(a) = a &&\quad f_3(a) = a\\
    & f_1(b) = b &&\quad f_2(b) = b &&\quad f_3(b) = a\\
    & f_1(c) = a &&\quad f_2(c) = b &&\quad f_3(c) = c\\
    & f_1(d) = b &&\quad f_2(d) = a &&\quad f_3(d) = c
  \end{alignat*}
  The plan of the proof is to analyze how~$\plotkinspecial$ travels around the naturality square for~$\lambda$, for each of the three functions. The element~$\plotkinspecial$ and the three functions have been carefully chosen so that the distributive law unit axioms can be applied during the proof.
  \begin{equation}%
    \label{eq:plotkin-key}
    \begin{gathered}
      \begin{tikzpicture}[scale=0.5, node distance=2cm, ->]
        \node (tl) {$DP(X)$};
        \node[below of=tl, node distance=1.5cm] (bl) {$DP(X)$};
        \node[right of=tl] (tr) {$PD(X)$};
        \node[right of=bl] (br) {$PD(X)$};
        \draw (tl) to node[above]{$\lambda_X$} (tr);
        \draw (bl) to node[below]{$\lambda_X$} (br);
        \draw (tl) to node[left]{$DP(f_i)$} (bl);
        \draw (tr) to node[right]{$PD(f_i)$} (br);
      \end{tikzpicture}
    \end{gathered}
  \end{equation}
  We proceed as follows:
  \begin{itemize}
  \item Trace~$\plotkinspecial$ around the naturality square~\eqref{eq:plotkin-key} for both~$f_1$ and~$f_2$. We note that as~$\{-,-\}$ is commutative, and~$+^{\frac{1}{2}}$ is idempotent:
    \[
      DP(f_1)(\plotkinspecial) = \eta^D_{PX}\{ a, b \} = DP(f_2)(\plotkinspecial).
    \]
    Commutativity, and particularly idempotence, will be important ideas for our subsequent generalizations. For both~$f_1$ and~$f_2$ we can apply the first distributive law unit axiom to conclude that:
    \[
      \lambda_X \circ DP(f_1)(\plotkinspecial) = \{ \eta^D_X(a), \eta^D_X(b) \} = \lambda_X \circ DP(f_2)(\plotkinspecial).
    \]
    Now a careful consideration of the actions~$PD(f_1)$ and~$PD(f_2)$ allows us to deduce that~$\lambda_X(\plotkinspecial)$ must be a subset of:
    \[
      \{ \eta^D_X(a), \eta^D_X(b), \eta^D_X(c), \eta^D_X(d) \}.
    \]
    This part is less straightforward to generalize. In principle it involves inverse images of equivalence classes of terms in one algebraic theory, with variables labelled by equivalence classes of terms in a second algebraic theory. This motivates our move to an explicitly algebraic formulation. We can see this whole step as establishing an upper bound on the set of variables appearing in~$\lambda_X(\plotkinspecial)$.
  \item We then trace~$\plotkinspecial$ around the naturality square~\eqref{eq:plotkin-key} for~$f_3$. In this case, we exploit the idempotence of the operation~$\{-, -\}$ to conclude:
    \begin{equation}%
      \label{eq:plotkin-intermediate}
      DP(f_3)(\plotkinspecial) = \{a\} +^{\frac{1}{2}} \{ c \}.
    \end{equation}
    This indicates idempotence is actually an important aspect of both monads for this argument to work. Equation~\eqref{eq:plotkin-intermediate} allows us to apply the second unit axiom for~$\lambda$ to conclude:
    \[
      \lambda_X \circ DP(f_3)(\plotkinspecial) = \{ a +^{\frac{1}{2}}  c \}.
    \]
    By considering the action of~$PD(f_3)$ as before, we conclude that~$\lambda_X(\plotkinspecial)$ must contain an element mapped onto~$a +^{\frac{1}{2}} c$ by~$PD(f_3)$, placing a lower bound on the set of variables that appear in~$\lambda_X(\plotkinspecial)$.
  \item The lower and upper bounds established in the previous two steps contradict each other, and so no distributive law of type~$D \circ P \Rightarrow P \circ D$ can exist.
  \end{itemize}
\end{counter}

\noindent
In summary, the argument requires two components:
\begin{enumerate}[label={[\arabic*]}]
\item Some operations satisfying certain algebraic equational properties such as idempotence and commutativity.
\item Some slightly more mysterious properties of our monads, making the ``inverse image'' parts of the argument work correctly.
\end{enumerate}
\begin{rem}
  The original counterexample is actually shown for what is known as the free real cone or finite valuation monad, as this requires slightly weaker assumptions. We state it here for the distribution monad simply to avoid the distraction of introducing yet another monad. This is essentially a cosmetic decision, and our later results are equally applicable to the original counterexample for the free real cone monad.
\end{rem}

Our aim in this section is to extract general methods for showing no-go results for constructing distributive laws, derived from the essential steps in this counterexample.
In order to do this, we isolate sufficient conditions on algebraic theories inducing two monads, such that there can be no distributive law between them. Earlier generalizations of this counterexample have appeared in~\cite{KlinSalamanca2018, DahlqvistNeves2018}. All the existing approaches involve direct calculations with the distributive law axioms, leading to somewhat opaque conditions. They also remain limited to the case where one of the two monads is the powerset monad, restricting their scope of application.

We introduce terminology for some special sets of terms in an algebraic theory. The theorems are stated in terms of these special sets, which in some cases can restrict the scope for which certain ``global'' conditions need to apply, broadening the range of applicability.
\begin{defi}[Universal Terms]
  For an algebraic theory, we say that a set of terms~$T$ is:
  \begin{itemize}
  \item \define{Universal} if every term is provably equal to a term in $T$.
  \item \define{Stable} if $T$ is closed under substitution of variables for variables.
  \end{itemize}
\end{defi}
\begin{exa}
  Some examples of universal and stable sets:
  \begin{enumerate}[label={[\arabic*]}]
  \item For any theory, the set of all terms is a stable universal set.
  \item For the theory of real vector spaces, every term is equal to a term in which scaling by the zero element does not appear. Terms that do not contain the scale by zero operation are clearly also stable under variable renaming. Therefore the terms not containing the scale by zero operation are a stable universal set.
  \item In the theory of groups, every term is equal to a term in which no subterm and its inverse are ``adjacent''. This set is therefore universal. It is not stable, as variable renaming may introduce a subterm adjacent to its inverse.
  \end{enumerate}
\end{exa}
\begin{rem}
  On first reading, it is probably easiest to take the universal stable sets required in subsequent theorems to be the set of all terms in a theory. This is by far the most common case.
\end{rem}

Throughout this section, the variable labels for any algebraic theory will range over the natural numbers.
We will also write~$n$ for the set~$\{1,\dots,n\}$, so for example, $2 \vdash t$
means~$t$ is a term containing at most two variables.

We proceed in four steps. Theorem~\ref{thm:plotkin-1} is an algebraic generalization of Plotkin's counterexample, capturing the algebraic properties required of both theories in order for a proof of this type to work. In \refthm~\ref{thm:plotkin-2} we generalize further, removing the restriction to binary terms that was sufficient for the original application. This generalization complicates the proof slightly, and so to clarify the methods involved we present two separate theorems. Finally, \refthms~\ref{thm:plotkin-3} and~\ref{thm:NoGo-Videm-Punit} address the need for commutativity-like axioms and idempotency axioms respectively, and bring further combinations of monads into scope.

\begin{thm}%
  \label{thm:plotkin-1}
  Let~$\mathbb{P}$ and~$\mathbb{V}$ be two algebraic theories, $T_{\mathbb{P}}$ a stable universal set of~$\mathbb{P}$-terms, and~$T_{\mathbb{V}}$ a stable universal set of $\mathbb{V}$-terms. If there are terms:
  \[
    2 \vdash_{\mathbb{P}} p \qquad\text{ and }\qquad 2 \vdash_{\mathbb{V}} v
  \]
  such that:
  \begin{enumerate}[label=(P\arabic*)]
  \item\label{ax:pcomm} $p$ is commutative:
    \[
      2 \vdash p(1,2) \theoryeq{P} p(2,1)
    \]
  \item\label{ax:pidem} $p$ is idempotent:
    \[
      1 \vdash p(1,1) \theoryeq{P} 1
    \]
  \item\label{ax:pprolif} For all~$p' \in T_{\mathbb{P}}$:
    \[
      \Gamma \vdash p(1,2) \theoryeq{P} p' \quad\Rightarrow\quad 2 \vdash p'
    \]
  \end{enumerate}
  \begin{enumerate}[label=(V\arabic*)]
  \item\label{ax:videm} $v$ is idempotent:
    \[
      1 \vdash v(1,1) \theoryeq{V} 1
    \]
  \item\label{ax:vvar} For all~$v' \in T_{\mathbb{V}}$, and any variable~$x$:
    \[
      \Gamma \vdash x \theoryeq{V} v' \quad\Rightarrow\quad \{x\} \vdash v'
    \]
  \item\label{ax:vbinary} For all~$v' \in T_{\mathbb{V}}$:
    \[
      \Gamma \vdash v(1,2) \theoryeq{V} v' \quad\Rightarrow\quad \neg (\{ 1 \} \vdash v'\vee \{ 2 \} \vdash v')
    \]
  \end{enumerate}
  Then there is no composite theory of~$\mathbb{P}$ after~$\mathbb{V}$.
\end{thm}
\begin{rem}
Properties~\ref{ax:pprolif},~\ref{ax:vvar}, and~\ref{ax:vbinary} are constraints on the variables appearing in certain terms, which are needed for the ``inverse image'' part of \refcounter~\ref{counter:Plotkin-original}. Property~\ref{ax:pprolif} states that any term equal to the special binary term $p$ can have at most two free variables. Property~\ref{ax:vvar} states that any term equal to a variable can only contain that variable, and property~\ref{ax:vbinary} states that any term equal to the special binary term $v$ must have at least two free variables. Notice that the upper/lower bound principle from the original argument is reflected in these conditions.
\end{rem}
\begin{proof}
  Assume by way of a contradiction that a composite theory~$\mathbb{U}$ of~$\mathbb{P}$ after~$\mathbb{V}$ exists.
  Consider the term $v(p(1,2),p(3,4))$. Then as~$\mathbb{U}$ is composite, there exist~$X \vdash p'$ and~$\Gamma \vdash v'_x$ for each $x \in X$ such that:
  \begin{equation}%
    \label{eq:plotkin-key-step}
    \Gamma \vdash v(p(1,2),p(3,4)) \theoryeq{U} p'[v'_x/x].
  \end{equation}
  Without loss of generality, we may assume that $p' \in T_{\mathbb{P}}$ and $v'_x \in T_{\mathbb{S}}$ by universality.
  Define partial function~$f_1$ as follows:
  \begin{align*}
    f_1(1) &= f_1(3) = 1\\
    f_1(2) &= f_1(4) = 2
  \end{align*}
  Then, using this substitution of variables and assumption~\ref{ax:videm}:
  \begin{align*}
    & \Gamma \vdash v(p(1,2),p(3,4)) \theoryeq{U} p'[v'_x/x] \\
    \Rightarrow \;& \reason{\text{substitution}} \\
     & \Gamma \vdash v(p(1,2),p(3,4))[f_1] \theoryeq{U} p'[v'_x[f_1]/x] \\
    \Rightarrow \; & \reason{\text{applying the substitution on the left hand side}} \\
     & \Gamma \vdash v(p(1,2),p(1,2)) \theoryeq{U} p'[v'_x[f_1]/x] \\
    \Rightarrow \; & \reason{\text{assumption~\ref{ax:videm}: } v \text{ is idempotent}} \\
     & \Gamma \vdash p(1,2) \theoryeq{U} p'[v'_x[f_1]/x].
  \end{align*}
  We now have two separated terms that are equal to each other, so we can apply essential uniqueness. This gives us that there are functions $g_1: 2 \rightarrow Z$ and $g_2: X \rightarrow Z$ such that:
  \begin{align}
    p[g_1] \theoryeq{P} p'[g_2] & \tag{Thm~\ref{thm:essentialuniqueness},~\ref{ax:essuniq1}}\\
    g_1(1) \neq g_1(2) & \tag{Thm~\ref{thm:essentialuniqueness},~\ref{ax:essuniq2}} \\
    g_1(1) = g_2(x) & \Leftrightarrow 1 \theoryeq{V} v'_x[f_1] \tag{Thm~\ref{thm:essentialuniqueness},~\ref{ax:essuniq4}} \\
    g_1(2) = g_2(x) & \Leftrightarrow 2 \theoryeq{V} v'_x[f_1] \tag{Thm~\ref{thm:essentialuniqueness},~\ref{ax:essuniq4}}.
  \end{align}
  As~$T_\mathbb{P}$ is stable, any variable renaming of~$p'$ is also in~$T_{\mathbb{P}}$. And so, by assumption~\ref{ax:pprolif} we must have for all $x \in \var(p')$: $g_2(x) = g_1(1)$ or $g_2(x) = g_1(2)$, which means for each~$v'_x$:
  \[
    \Gamma \vdash 1 \theoryeq{V} v'_x[f_1] \quad\vee\quad \Gamma \vdash 2 \theoryeq{V} v'_x[f_1].
  \]
  Then, using assumption~\ref{ax:vvar} and the preimage of $f_1$, we conclude that for all~$v'_x$:
  \begin{equation}%
    \label{eq:condition1}
    \{1,3\} \vdash v'_x \quad\vee\quad \{2,4\} \vdash v'_x.
  \end{equation}

  We use the same strategy again, but with a different substitution. Define a second partial function~$f_2$ as follows:
  \begin{align*}
    f_2(1) &= f_2(4) = 1\\
    f_2(2) &= f_2(3) = 2
  \end{align*}
  Using this substitution and assumptions~\ref{ax:pcomm} and~\ref{ax:videm}:
  \begin{align*}
    & \Gamma \vdash v(p(1,2),p(3,4)) \theoryeq{U} p'[v'_x/x] \\
    \Rightarrow \; & \reason{\text{substitution}} \\
     & \Gamma \vdash v(p(1,2),p(3,4))[f_2] \theoryeq{U} p'[v'_x[f_2]/x] \\
    \Rightarrow \; & \reason{\text{applying the substitution on the left hand side}} \\
     & \Gamma \vdash v(p(1,2),p(2,1)) \theoryeq{U} p'[v'_x[f_2]/x] \\
    \Rightarrow \; & \reason{\text{assumption~\ref{ax:pcomm}: } p \text{ is commutative}} \\
    & \Gamma \vdash v(p(1,2),p(1,2)) \theoryeq{U} p'[v'_x[f_2]/x] \\
    \Rightarrow \; & \reason{\text{assumption~\ref{ax:videm}: } v \text{ is idempotent}} \\
    & \Gamma \vdash p(1,2) \theoryeq{U} p'[v'_x[f_2]/x].
  \end{align*}
  Again we have two separated terms that are equal, so we apply essential uniqueness. This gives us that there are functions $g_3: 2 \rightarrow Z$ and $g_4: X \rightarrow Z$ such that:
  \begin{align}
    p[g_3] \theoryeq{P} p'[g_4] & \tag{Thm~\ref{thm:essentialuniqueness},~\ref{ax:essuniq1}}\\
    g_3(1) \neq g_3(2) & \tag{Thm~\ref{thm:essentialuniqueness},~\ref{ax:essuniq2}} \\
    g_3(1) = g_4(x) & \Leftrightarrow 1 \theoryeq{V} v'_x[f_2] \tag{Thm~\ref{thm:essentialuniqueness},~\ref{ax:essuniq4}} \\
    g_3(2) = g_4(x) & \Leftrightarrow 2 \theoryeq{V} v'_x[f_2] \tag{Thm~\ref{thm:essentialuniqueness},~\ref{ax:essuniq4}}.
  \end{align}
  As~$T_\mathbb{P}$ is stable, any variable renaming of~$p'$ is also in~$T_{\mathbb{P}}$. Therefore, by assumption~\ref{ax:pprolif} we must have for all $x \in \var(p')$: $g_4(x) = g_3(1)$ or $g_4(x) = g_3(2)$, which means for each~$v'_x$:
  \[
    \Gamma \vdash 1 \theoryeq{V} v'_x[f_2] \quad\vee\quad \Gamma \vdash 2 \theoryeq{V} v'_x[f_2].
  \]
  So, using assumption~\ref{ax:vvar} and the preimage of $f_2$, we know that for all~$v'_x$:
  \begin{equation}%
    \label{eq:condition2}
    \{1,4\} \vdash v'_x \qquad\vee\qquad \{2,3\} \vdash v'_x.
  \end{equation}

  We combine the two conclusions~\eqref{eq:condition1} and~\eqref{eq:condition2}, yielding that for all~$v'_x$:
  \begin{align*}
  & (\{1,3\} \vdash v'_x \;\vee\; \{2,4\} \vdash v'_x) \;\wedge\; (\{1,4\} \vdash v'_x \;\vee\; \{2,3\} \vdash v'_x) \\
  \Rightarrow \;& \reason{\text{distributing } \wedge \text{ over } \vee} \\
  & (\{1,3\} \vdash v'_x \;\wedge\; \{1,4\} \vdash v'_x) \;\vee\; (\{1,3\} \vdash v'_x \;\wedge\; \{2,3\} \vdash v'_x) \;\vee \\
  & (\{2,4\} \vdash v'_x \;\wedge\; \{1,4\} \vdash v'_x) \;\vee\; (\{2,4\} \vdash v'_x \;\wedge\; \{2,3\} \vdash v'_x) \\
   \Rightarrow \; & \reason{\text{this is only possible if}} \\
   & \{1\} \vdash v'_x \;\vee\; \{3\} \vdash v'_x \;\vee\; \{4\} \vdash v'_x \;\vee\; \{2\} \vdash v'_x.
  \end{align*}
  In other words:
  \begin{equation}%
    \label{eq:condition3}
    \bigvee_{n \in 4} \{n\} \vdash v'_x,
  \end{equation}
  that is, each $v'_x$ can have at most one variable.

  To get a contradiction, we will now find a $v'_x$ that must have at least two variables. This is where we need assumption~\ref{ax:vbinary}. We make one more substitution. Define a third partial function~$f_3$ as follows:
  \begin{align*}
    f_3(1) &= f_3(2) = 1\\
    f_3(3) &= f_3(4) = 2
  \end{align*}
  Using this final substitution and~\ref{ax:pidem}:
  \begin{align*}
    & \Gamma \vdash v(p(1,2),p(3,4)) \theoryeq{U} p'[v'_x/x] \\
    \Rightarrow \; & \reason{\text{substitution}} \\
    & \Gamma \vdash v(p(1,2),p(3,4))[f_3] \theoryeq{U} p'[v'_x[f_3]/x] \\
    \Rightarrow \; & \reason{\text{applying the substitution on the left hand side}} \\
    & \Gamma \vdash v(p(1,1),p(2,2)) \theoryeq{U} p'[v'_x[f_3]/x] \\
    \Rightarrow \; & \reason{\text{assumption~\ref{ax:pidem}: } p \text{ is idempotent}} \\
    & \Gamma \vdash v(1,2) \theoryeq{U} p'[v'_x[f_3]/x] \\
    \Rightarrow \; & \reason{\text{making it more obvious that } v \text{ is a separated term}} \\
    & \Gamma \vdash 5[v(1,2)/5] \theoryeq{U} p'[v'_x[f_3]/x].
  \end{align*}
  Again, we arrive at an equality between two separated terms, allowing us to apply essential uniqueness. This gives us that there are functions $g_5: \{5\} \rightarrow Z$ and $g_6: X \rightarrow Z$ such that:
  \begin{align}
    5[g_5] \theoryeq{P} p'[g_6] & \tag{Thm~\ref{thm:essentialuniqueness},~\ref{ax:essuniq1}}\\
    g_5(5) = g_6(x) & \Leftrightarrow v(1,2) \theoryeq{V} v'_x[f_3] \tag{Thm~\ref{thm:essentialuniqueness},~\ref{ax:essuniq4}}.
  \end{align}
  As $\mathbb{P}$ is consistent, the variable $5[g_5]$ must appear in $p'[g_6]$. If it did not, we could define a substitution $h$ mapping $5[g_5]$ to any other variable $y$, and then conclude $y \theoryeq{P} 5[g_5][h] \theoryeq{P} p'[g_6][h] \theoryeq{P} p'[g_6] \theoryeq{P} 5[g_5]$, which proves all variables are equal to each other in $\bb{P}$, which means $\bb{P}$ is inconsistent. So by consistency of $\bb{P}$, the variable $5[g_5]$ must appear in $p'[g_6]$. Hence, there is an $x_0$ such that $g_5(5) = g_6(x_0)$. And so:
  \[
    \Gamma \vdash v(1,2) \theoryeq{V} v'_{x_0}[f_3].
  \]
  As~$T_{\mathbb{S}}$ is stable, this $v'_{x_0}[f_3]$ is an element of~$T_{\mathbb{V}}$, and so by~\ref{ax:vbinary}, $v'_{x_0}[f_3]$ must have at least two variables. Which means $v'_{x_0}$ must have at least two variables. This contradicts \refeqn~\eqref{eq:condition3}, which claims that each $v'_x$ can have at most one variable. Therefore the assumed composite theory cannot exist.
\end{proof}

The following corollary reflects our real interest in monads:
\begin{cor}%
  \label{cor:plotkin-1}
  If monads~$P$ and~$V$ have presentations~$\mathbb{P}$ and~$\mathbb{V}$ such that the conditions of \refthm~\ref{thm:plotkin-1} can be satisfied, then there is no distributive law of type~$V \circ P \Rightarrow P \circ V$.
\end{cor}
The subsequent theorems have similar corollaries, which we will not state explicitly.
\begin{exa}[Powerset and Distribution Monad]
  Consider the terms $1 \vee 2$ and $1 +^{\frac{1}{2}} 2$ in the theories representing the powerset and distribution monads of~\refexs~\ref{ex:monoids} and~\ref{ex:distribution-monad-presentation}.
   Since both of these terms are binary, commutative, and idempotent, and the remaining axioms are satisfied, \refthm~\ref{thm:plotkin-1} captures the known results that there are no distributive laws of type~$D \circ P \Rightarrow P \circ D$~\cite{VaraccaWinskel2006}, $P \circ P \Rightarrow P \circ P$~\cite{KlinSalamanca2018}, or $P \circ D \Rightarrow D \circ P$~\cite[stated without proof]{Varacca2003}.
  In addition, \refthm~\ref{thm:plotkin-1} yields the new result that there is no distributive law of type~$D \circ D \Rightarrow D \circ D$, completing the picture for these monads.
\end{exa}
\begin{exa}[Powerset and Distribution Monad Again]
  We can also consider the distribution monad to be presented by binary operations~$+^p$ with~$p$ in the~\emph{closed} interval~$[0,1]$, and in fact this is the more common formulation. In this case, \refthm~\ref{thm:plotkin-1} can still be directly applied, without having to move to the more parsimonious presentation. We simply note that the terms not involving the operations~$+^1$ and~$+^0$ form a stable universal set satisfying the required axioms. The results discussed in the previous example can then be recovered using the conventional presentation of the distribution monad.
\end{exa}
\begin{nonexample}[Reader Monad]
It is well known that the reader monad distributes over itself. Looking at the presentation of the reader monad given in \refex~\ref{ex:reader-monad-presentation}, we see that although it has idempotent terms, there is no commutative term and hence \refthm~\ref{thm:plotkin-1} does not apply.
\end{nonexample}
A natural question to ask with regard to \refthm~\ref{thm:plotkin-1} is whether the choice of~\emph{binary} terms for both~$p$ and~$v$ is necessary. We thank Prakash Panangaden for posing this question during an informal presentation of an earlier version of this work~\cite{Panangaden2018}. The answer is that we can generalize to terms with any arities strictly greater than one.  Before we prove this more general statement, we introduce a lemma that is central to establishing the upper bound part of the argument.
\begin{lem}%
  \label{lem:filter}
  Let~$n,m$ be strictly positive natural numbers, and~$\sigma$ a fixed-point free permutation of~$\{1,\dots,m\}$. For distinct variables~$a^j_i$, $1 \leq i \leq m$, $1 \leq j \leq n$, the sets:
  \begin{align*}
    &\{ a_{i_1}^1,a_{i_1}^2, a_{i_1}^3, \ldots, a_{i_1}^n \}\\
    &\{ a_{i_2}^1,a_{\sigma (i_2)}^2, a_{i_2}^3, \ldots, a_{i_2}^n \}\\
    &\qquad \quad \;\; \vdots\\
    &\{ a_{i_n}^1,a_{i_n}^2, a_{i_n}^3, \ldots, a_{\sigma (i_n)}^n \}
  \end{align*}
  have at most one common element. Here, each $i_k$ is an element of $\{i \mid 1 \leq i \leq m\}$, not necessarily unique.
\end{lem}
\begin{proof}
We proceed by induction on~$n$. The base case $n = 1$ is trivially true. For $n = n' + 1$, we consider the first two rows of our table of sets. There are two cases.
  \begin{enumerate}
    \item If~$i_1 = \sigma(i_2)$, then the first two rows can only agree at their second element, because each $a_i^j$ is distinct, and from the assumption that~$\sigma$ is fixed-point free we know that $i_1$ is different from $i_2$ if $\sigma(i_2) = i_1$. The claim follows directly from this observation.
    \item If~$i_1 \neq \sigma(i_2)$ then the first two rows disagree in the second column. Therefore the elements common to all the sets cannot appear in the second column.
      We then remove both row and column 2, and invoke the induction hypothesis for~$n = n'$.
      \qedhere
  \end{enumerate}
\end{proof}

\noindent
We then get a more general variant of~\refthm~\ref{thm:plotkin-1}.
\begin{thm}%
  \label{thm:plotkin-2}
  Let~$\mathbb{P}$ and~$\mathbb{V}$ be two algebraic theories, $T_{\mathbb{P}}$ a stable universal set of~$\mathbb{P}$-terms, and~$T_{\mathbb{V}}$ a stable universal set of $\mathbb{V}$-terms. If there are terms:
  \[
    m \vdash_{\mathbb{P}} p \qquad\text{ and }\qquad n \vdash_{\mathbb{V}} v
  \]
  such that:
  \begin{enumerate}[label=(P\arabic*)]
    \setcounter{enumi}{3}
  \item\label{ax:pcomm-2} $p$ is stable under a fixed-point free permutation~$\sigma$:
    \[
      m \vdash p \theoryeq{P} p[\sigma]
    \]
  \item\label{ax:pidem-2} $p$ is idempotent:
    \[
      1 \vdash p[1/i] \theoryeq{P} 1
    \]
  \item\label{ax:pprolif-2} For all~$p' \in T_{\mathbb{P}}$:
    \[
      \Gamma \vdash p \theoryeq{P} p' \quad\Rightarrow\quad m \vdash p'
    \]
  \end{enumerate}
  \begin{enumerate}[label=(V\arabic*)]
    \setcounter{enumi}{3}
  \item\label{ax:videm-2} $v$ is idempotent:
    \[
      1 \vdash v[1/i] \theoryeq{V} 1
    \]
  \item\label{ax:vvar-2} For all~$v' \in T_{\mathbb{V}}$, and any variable~$i$:
    \[
      \Gamma \vdash i \theoryeq{V} v' \quad\Rightarrow\quad \{ i \} \vdash v'
    \]
  \item\label{ax:vbinary-2} For all~$v' \in T_{\mathbb{V}}$:
    \[
      \Gamma \vdash v \theoryeq{V} v' \quad\Rightarrow\quad \neg \left(\bigvee_{i \in \Gamma} \{ i \} \vdash v' \right)
    \]
  \end{enumerate}
  Then there is no composite theory of~$\mathbb{P}$ after~$\mathbb{V}$.
\end{thm}
\begin{rem}
  The required properties are generalizations of the binary conditions in \refthm~\ref{thm:plotkin-1}. Most are straightforward, but axiom~\ref{ax:pcomm-2}, the analogue of binary commutativity, is perhaps slightly surprising. Here we only require stability under a single fixed-point free permutation.
\end{rem}
\begin{proof}
  Assume by way of a contradiction that a composite theory~$\mathbb{U}$ of~$\mathbb{P}$ after~$\mathbb{V}$ exists. Let
  \[
    a_i^j, 1 \leq i \leq m, 1 \leq j \leq n
  \]
  denote distinct variables. We consider the term
  \[
  v(p(a_1^1, \ldots, a_m^1), \ldots, p(a_1^n, \ldots, a_m^n)).
  \]
  Then as~$\mathbb{U}$ is composite, the separation axiom of composite theories tells us that there exist~$X \vdash p'$ and~$\Gamma \vdash v'_x$ for each $x \in X$ such that:
  \[
    \Gamma \vdash v(p(a_1^1, \ldots, a_m^1), \ldots, p(a_1^n, \ldots, a_m^n)) \theoryeq{U} p'[v'_x/x].
  \]
  Without loss of generality, we may assume~$p' \in T_{\mathbb{P}}$ and~$v'_x \in T_{\mathbb{V}}$ by universality.

  We use the same approach as in the proof of \refthm~\ref{thm:plotkin-1}, using substitutions to bound the variables that can appear in each of the $v'_x$.
  Define substitution~$f_1$ as follows:
  \[
    f_1(a_i^j) = a_i^1.
  \]
  We then have:
  \[
    \Gamma \vdash v(p(a_1^1, \ldots, a_m^1), \ldots, p(a_1^1, \ldots, a_m^1)) \theoryeq{U} p'[v'_x[f_1]/x].
  \]
  By assumption~\ref{ax:videm-2}, idempotence of $v$:
  \[
    \Gamma\vdash p(a_1^1, \ldots, a_m^1) \theoryeq{U} p'[v'_x[f_1]/x].
  \]
  As~$T_\mathbb{P}$ is stable, any variable renaming of~$p'$ is also in~$T_{\mathbb{P}}$. Therefore, essential uniqueness tells us that there are substitutions $g_1$ and $g_2$ such that:
  \begin{align}
    p[g_1] & \theoryeq{P} p'[g_2]  \tag{Thm~\ref{thm:essentialuniqueness},~\ref{ax:essuniq1}}\\
    g_1(a_i^1) & \neq g_1(a_j^1) (i \neq j) \tag{Thm~\ref{thm:essentialuniqueness},~\ref{ax:essuniq2}} \\
    g_1(a_i^1) & = g_2(x)  \Leftrightarrow a_i^1 \theoryeq{V} v'_x[f_1] \tag{Thm~\ref{thm:essentialuniqueness},~\ref{ax:essuniq4}},
  \end{align}
  and~\ref{ax:pprolif-2} gives us that:
  \[
    \forall x \; \exists i:\; \Gamma \vdash a_i^1 \theoryeq{V} v'_x[f_1].
  \]
  Then by assumption~\ref{ax:vvar-2}:
  \[
  \forall x \; \exists i:\; \{a_i^1\} \vdash v'_x[f_1].
  \]
  So our first approximation of the variables appearing in $v'_x$ is:
  \begin{equation}\label{eq:filter-1}
    \forall x \; \exists i: \; \{ a_i^1, \ldots, a_i^n \} \vdash v'_x.
  \end{equation}
  Now we define a family of substitutions for~$2 \leq k \leq n$ as follows:
  \[
    f_k(a_i^j) =
    \begin{cases}
      a_{\sigma(i)}^k \text{ if } j = k\\
      a_i^k \text{ otherwise }.
    \end{cases}
  \]
  If we follow a similar argument as before, using essential uniqueness,~\ref{ax:pprolif-2},~\ref{ax:vvar-2}, and also exploiting assumption~\ref{ax:pcomm-2}, we conclude that:
  \[
  \forall x,k \; \exists i_k: \{a^k_{i_k}\} \vdash v'_x[f_k].
  \]
  And so:
  \begin{equation}%
    \label{eq:filter-2}
    \forall x,k \; \exists i_k: \; \{ a_{\sigma^{-1} (i_k)}^j \mid j = k \} \cup \{ a_{i_k}^j \mid j \neq k \} \vdash v'_x.
  \end{equation}
  Then we note that by \reflem~\ref{lem:filter}, conditions~\eqref{eq:filter-1} and~\eqref{eq:filter-2}:
  \begin{equation}%
    \label{eq:filtered}
    \forall x \; \exists i,j: \; \{ a_i^j \} \vdash v'_x.
  \end{equation}
  This provides an upper bound on the number of variables appearing in the $v'_x$, just like the first two substitutions did in the proof of \refthm~\ref{thm:plotkin-1}. To finish the argument, we define another substitution:
  \[
    f_{n+1}(a_i^j) = a_1^j
  \]
  Applying this substitution:
  \[
    \Gamma \vdash v(p(a_1^1, \ldots, a_1^1), \ldots, p(a_1^n, \ldots, a_1^n)) \theoryeq{U} p'[v'_x[f_{n+1}] / x ].
  \]
  Using assumption~\ref{ax:pidem-2}:
  \[
    \Gamma \vdash v(a_1^1, \ldots, a_1^n) \theoryeq{U} p'[v'_x[f_{n+1}] / x ].
  \]
  By essential uniqueness and consistency:
  \[
    \exists x_0: \; \Gamma \vdash v(a_1^1, \ldots, a_1^n) \theoryeq{V} v'_{x_0}[f_{n+1}].
  \]
  As~$T_{\mathbb{V}}$ is stable, this $v'_x[f_{n+1}]$ is an element of~$T_{\mathbb{V}}$. And so, by assumption~\ref{ax:vbinary-2}, $v'_{x_0}$ must contain at least two variables, but this contradicts conclusion~\eqref{eq:filtered}, and so the assumed composite theory cannot exist.
\end{proof}

It is clear that the simpler \refthm~\ref{thm:plotkin-1} is a special case of \refthm~\ref{thm:plotkin-2}. Besides providing greater generality, the main point of \refthm~\ref{thm:plotkin-2} is that it clearly demonstrates that there is nothing special about binary terms. This further clarifies our understanding of what abstract properties make the original counterexample of Plotkin work. By moving to such a high level of abstraction it is also easier to see that our other methods, described in \refsecs~\ref{section:beyondPlotkin} and~\ref{sec:plus-over-times}, are not simply a further generalization of Plotkin's counterexample, as they make fundamentally different assumptions of the underlying algebraic theories.

\subsection{Concerning Commutativity}

In \refthm~\ref{thm:plotkin-1} we require the special term~$p$ to be commutative in order to establish that no composite theory exists. In \refthm~\ref{thm:plotkin-2} this commutativity was generalized to stability under the action of a fixed-point free permutation. This raises the question of whether commutativity-like axioms are essential to this type of proof. In fact, this is not the case, and a similar no-go theorem can be established under modified assumptions that make no use of commutativity.

\begin{thm}%
  \label{thm:plotkin-3}
  Let~$\mathbb{P}$ and~$\mathbb{V}$ be two algebraic theories, $T_{\mathbb{P}}$ a stable universal set of~$\mathbb{P}$-terms, and~$T_{\mathbb{V}}$ a stable universal set of $\mathbb{V}$-terms. If there are terms:
  \[
    2 \vdash_{\mathbb{P}} p \qquad\text{ and }\qquad 2 \vdash_{\mathbb{V}} v
  \]
  such that axioms~\ref{ax:pidem},~\ref{ax:pprolif},~\ref{ax:videm},~\ref{ax:vvar} and~\ref{ax:vbinary} hold, and:
  \begin{enumerate}[label=(P\arabic*)]
    \setcounter{enumi}{6}
  \item\label{ax:pvar-3} For all~$p' \in T_{\mathbb{P}}$, and any variable~$x$:
    \[
      \Gamma \vdash x \theoryeq{P} p' \quad\Rightarrow\quad \{x\} \vdash p'
    \]
  \item\label{ax:pbinary-3} For all~$p' \in T_{\mathbb{P}}$:
    \[
      \Gamma \vdash p(1,2) \theoryeq{P} p' \quad\Rightarrow\quad \neg (\emptyset \vdash p')
    \]
  \end{enumerate}
  Then there is no composite theory of~$\mathbb{P}$ after~$\mathbb{V}$.
\end{thm}

\begin{rem}
  Axiom~\ref{ax:pvar-3} is the same as axiom~\ref{ax:vvar}, but now we require it to hold for $\bb{P}$ as well as $\bb{V}$. Axiom~\ref{ax:pbinary-3} states that the term $p(1,2)$ is not equal to a constant in $\bb{P}$.
\end{rem}

\begin{proof}
  Assume by way of a contradiction that a composite theory~$\mathbb{U}$ of~$\mathbb{P}$ after~$\mathbb{V}$ exists. Then as~$\mathbb{U}$ is composite, there exist~$X \vdash p'$ and~$\Gamma \vdash v'_x$ for each variable $x \in \var(p')$ such that:
  \begin{equation}\label{eq:separation}
    \Gamma \vdash v(p(1,2),p(3,4)) \theoryeq{U} p'[v'_x/x].
  \end{equation}
  Without loss of generality, we may assume~$p' \in T_{\mathbb{P}}$ and all~$v'_x \in T_{\mathbb{V}}$ by universality.

   As in previous proofs, we use various substitutions to learn more about the terms $p'$ and all $v'_x$. The first substitution reduces the left hand side of Equation~\eqref{eq:separation} to a term involving just $p(1,2)$. This allows us to use essential uniqueness to get information about $v'_x$. Define substitution $f_1$ as follows:
  \begin{align*}
    f_1(1) &= f_1(3) = 1\\
    f_1(2) &= f_1(4) = 2
  \end{align*}
  Then from \refeqn~\eqref{eq:separation} we conclude:
  \begin{align*}
  & v(p(1,2),p(3,4)) \theoryeq{U} p'[v'_x/x] \\
  \Rightarrow\; & \reason{\text{Substitution}} \\
  & v(p(1,2)[f_1],p(3,4)[f_1]) \theoryeq{U} p'[v'_x[f_1]/x] \\
  \Rightarrow\; & \reason{\text{applying the substitution on the left hand side}} \\
  & v(p(1,2),p(1,2)) \theoryeq{U} p'[v'_x[f_1]/x] \\
  \Rightarrow\; & \reason{\text{Assumption }\ref{ax:videm}: v \text{ is idempotent }} \\
  & p(1,2) \theoryeq{U} p'[v'_x[f_1]/x].
  \end{align*}
  We now have two separated terms that are equal to each other, so by essential uniqueness, there are functions $g_1: 2 \rightarrow Z$ and $g_2: X \rightarrow Z$ such that:
  \begin{equation}\label{eq:p=p'}
    p[g_1] \theoryeq{P} p'[g_2],
  \end{equation}
  and, for all $x \in \var(p')$:
  \begin{align*}
    g_1(1) = g_2(x) & \Leftrightarrow 1 \theoryeq{V} v'_x[f_1] \\
    g_1(2) = g_2(x) & \Leftrightarrow 2 \theoryeq{V} v'_x[f_1].
  \end{align*}
  As~$T_\mathbb{P}$ is stable, any variable renaming of~$p'$ is also in~$T_{\mathbb{P}}$. And so, by assumption~\ref{ax:pprolif}, we must have $g_2(x) = g_1(1)$ or $g_2(x) = g_1(2)$, which means:
  \[
    \Gamma \vdash 1 \theoryeq{V} v'_x[f_1] \quad\vee\quad \Gamma \vdash 2 \theoryeq{V} v'_x[f_1].
  \]
  Then using assumption~\ref{ax:vvar}, for all~$v'_x$:
  \begin{equation}%
    \label{eq:condition1-plotkin3}
    \{1,3\} \vdash v'_x \quad\vee\quad \{2,4\} \vdash v'_x.
  \end{equation}
We can now split the set of variables $\var(p')$ into two disjoint subsets:
\begin{align}\label{index1}
  X_1 & = \{x \;|\; \{1,3\} \vdash v'_x \} \\
  X_2 & = \{x \;|\; \{2,4\} \vdash v'_x \}.
\end{align}
The aim of the rest of the proof will be to show that $\var(p') = \emptyset$, which will give a contradiction with assumption~\ref{ax:pbinary-3}. For this, we will need some more substitutions.
Define a second partial function $f_2$ as follows:
\begin{align*}
    f_2(1) &= f_2(2) = 1\\
    f_2(3) &= f_2(4) = 2
\end{align*}
Applying this substitution to \refeqn~\eqref{eq:separation}:
\begin{align*}
& v(p(1,2),p(3,4)) \theoryeq{U} p'[v'_x/x] \\
\Rightarrow\; & \reason{\text{Substitution}} \\
& v(p(1,2)[f_2],p(3,4)[f_2]) \theoryeq{U} p'[v'_x[f_2]/x] \\
\Rightarrow\; & \reason{\text{applying the substitution on the left hand side}} \\
& v(p(1,1),p(2,2)) \theoryeq{U} p'[v'_x[f_2]/x] \\
\Rightarrow\; & \reason{\text{Assumption }\ref{ax:pidem}: p \text{ is idempotent }} \\
& v(1,2) \theoryeq{U} p'[v'_x[f_2]/x],
\end{align*}
which is the same as:
\[
5[v(1,2)/5] \theoryeq{U} p'[v'_x[f_2] / x].
\]
Again we have two separated terms that are equal, allowing us to use essential uniqueness. We conclude that there are functions $g_3: \{5\} \rightarrow Z$ and $g_4: X \rightarrow Z$ such that:
  \[
    5[g_3] \theoryeq{P} p'[g_4],
  \]
  and, whenever $g_3(5) = g_4(x)$, $v(1,2) \theoryeq{V} v'_x[f_2]$.
  By assumption~\ref{ax:pvar-3}:
  \[
   \{g_3(5)\} \vdash p'[g_4].
  \]
  And so for all $x \in \var(p')$:
  \begin{equation}\label{eq:condition2-plotkin3}
    v'_x[f_2] \theoryeq{V} v(1,2).
  \end{equation}
That is, under substitution $f_2$, all the $v'_x$ are equal to $v(1,2)$. The next step is designed to get rid of substitution $f_2$, and hence fully understand each $v'_x$. Since we already know from \refeqn~\eqref{eq:condition1-plotkin3} that the variables appearing in $v'_x$ are different depending on whether $x \in X_1$ or $x \in X_2$, we need to treat those cases separately.
Starting with $x \in X_1$, consider the partial functions $f_3$:
\begin{align*}
    f_3(1) & = 1  & f_3(3) & = 3\\
    f_3(2) & = 3  & f_3(4) & = 4
\end{align*}
This substitution will act as an `inverse' for $f_2$, as will become clear in the following argument.
Combining Equations~\eqref{eq:condition1-plotkin3} and~\eqref{eq:condition2-plotkin3}:
\begin{align*}
 & v'_x[f_2][f_3] \\
 =\; & \reason{\text{writing out the substitutions}}  \\
 & v'_x[1/1, 1/2, 2/3, 2/4][1/1, 3/2, 3/3, 4/4] \\
 =\; & \reason{\text{ignoring the identity substitutions}} \\
 & v'_x[1/2, 2/3, 2/4][3/2] \\
 =\; & \reason{x \in X_1 \text{ so }\{1,3\} \vdash v'_x} \\
 & v'_x[2/3][3/2] \\
 =\; & \reason{\text{the second substitution is the inverse of the first}} \\
 & v'_x.
\end{align*}
Since $v'_x[f_2] \theoryeq{V} v(1,2)$ by \refeqn~\eqref{eq:condition2-plotkin3}, we conclude for all $x \in X_1$:
\begin{align}
& v'_x[f_2] \theoryeq{V} v(1,2) \nonumber\\
\Rightarrow\; & \reason{\text{substitution}} \nonumber\\
 & v'_x[f_2][f_3] \theoryeq{V} v(1,2)[f_3] \nonumber\\
\Rightarrow\; & \reason{\text{by the above: } v'_x[f_2][f_3] = v'_x } \nonumber\\
 & v'_x \theoryeq{V} v(1,2)[f_3] \nonumber\\
\Rightarrow\; & \reason{\text{applying } f_3} \nonumber \\
& v'_x \theoryeq{V} v(1,3).\label{eq:vi=v}
\end{align}
So for $x \in X_1$, we know that $v'_x \theoryeq{V} v(1,3)$.
A similar line of reasoning can be followed for $x \in X_2$. To negate substitution $f_2$ for these $v'_x$, we need the following substitution:
\begin{align*}
    f_4(1) & = 2  & f_4(3) & = 3\\
    f_4(2) & = 4  & f_4(4) & = 4
\end{align*}
The rest of the argument is the same, leading to the conclusion that for all $x \in X_2$, $v'_x \theoryeq{V} v(2,4)$.
We now have all the information about the $v'_x$ that we need. We start working towards a contradiction with yet another substitution. This time, it is a substitution of \emph{terms} for variables:
\begin{align*}
  f_5(1) & = p(1,2) & f_5(3) & = p(1,1) \\
  f_5(2) & = p(1,2) & f_5(4) & = p(2,2)
\end{align*}
These terms are chosen in such a way that the left hand side of Equation~\eqref{eq:separation} reduces to $p(1,2)$ after substitution with $f_5$, so that we can use essential uniqueness again. On the right hand side, the substitution creates $\mathbb{V}$-terms with $\mathbb{P}$-terms in them. Separating these terms into $\mathbb{P}$-terms built out of $\mathbb{V}$-terms will yield new information, and eventually the contradiction.
So, starting from Equation~\eqref{eq:separation}:
\begin{align}
  & v(p(1,2),p(3,4)) \theoryeq{U} p'[v'_x/x] \nonumber \\
  \Rightarrow\; & \reason{\text{Substitution}} \nonumber \\
  & v(p(1,2)[f_5],p(3,4)[f_5]) \theoryeq{U} p'[v'_x[f_5]/x] \nonumber\\
  \Rightarrow\; & \reason{\text{applying the substitution on the left hand side}} \nonumber\\
   & v(p(p(1,2),p(1,2)),p(p(1,1),p(2,2))) \theoryeq{U} p'[v'_x[f_5]/x] \nonumber \\
  \Rightarrow\; & \reason{p \text{ is idempotent}} \nonumber\\
  & v(p(1,2),p(1,2)) \theoryeq{U} p'[v'_x[f_5]/x] \nonumber\\
  \Rightarrow\; & \reason{v \text{ is idempotent}} \nonumber\\
  & p(1,2) \theoryeq{U} p'[v'_x[f_5]/x].\label{eq:condition3-plotkin3}
\end{align}
Before applying essential uniqueness to Equation~\eqref{eq:condition3-plotkin3}, we need to separate the right hand side of this equation. Remember from Equation~\eqref{eq:vi=v} that $v'_x \theoryeq{V} v(1,3)$ or $v'_x \theoryeq{V} v(2,4)$ depending on whether $x \in X_1$ or $x \in X_2$. And so, for $x \in X_1$:
\begin{align*}
  & v'_x \theoryeq{V} v(1,3) \\
 \Rightarrow\; & \reason{\text{substitution}} \\
  & v'_x[f_5] \theoryeq{U} v(1,3)[f_5] \\
 \Rightarrow\; & \reason{\text{applying the substitution on the right hand side}}\\
  & v'_x[f_5] \theoryeq{U} v(p(1,2),p(1,1)) \\
 \Rightarrow\; & \reason{\text{Substitution: } 1/3, 1/4 } \\
 & v'_x[f_5] \theoryeq{U} v(p(1,2),p(3,4)[1/3, 1/4]) \\
 \Rightarrow\; & \reason{\text{Equation~\eqref{eq:separation}}} \\
 & v'_x[f_5] \theoryeq{U} p'[v'_y[1/3, 1/4]/y],
\end{align*}
where we know from Equation~\eqref{eq:vi=v} that for $y \in X_1$, $v'_y \theoryeq{V} v(1,3)$ and for $y \in X_2$, $v'_y \theoryeq{V} v(2,4)$.

Similarly, for $x \in X_2$:
\begin{align*}
  & v'_x \theoryeq{V} v(2,4) \\
  \Rightarrow\; & \reason{\text{substitution}} \\
  & v'_x[f_5] \theoryeq{U} v(2,4)[f_5] \\
 \Rightarrow\; & \reason{\text{applying the substitution on the right hand side}}\\
 & v'_x[f_5] \theoryeq{U} v(p(1,2),p(2,2)) \\
 \Rightarrow\; & \reason{\text{Substitution: } 2/3, 2/4} \\
 & v'_x[f_5] \theoryeq{U} v(p(1,2),p(3,4)[2/3, 2/4]) \\
 \Rightarrow\; & \reason{\text{Equation~\eqref{eq:separation}}} \\
 & v'_x[f_5] \theoryeq{U} p'[v'_y[2/3, 2/4]/y],
\end{align*}
where we know from Equation~\eqref{eq:vi=v} that for $y \in X_1$, $v'_y \theoryeq{V} v(1,3)$ and for $y \in X_2$, $v'_y \theoryeq{V} v(2,4)$.
And so, continuing from Equation~\eqref{eq:condition3-plotkin3}:
\begin{align*}
& p(1,2) \theoryeq{U} p'[v'_x[f_5]/x] \\
\Rightarrow \; & p(1,2) \theoryeq{U} p'[p'[v'_y[1/3, 1/4]/y] / x \in X_1, \\
& \;\;\;\;\;\;\;\;\;\;\;\;\;\;\;\;\;\;\;\;\, p'[v'_y[2/3, 2/4]/y] / x \in X_2],
\end{align*}
where $v'_y \theoryeq{V} v(1,3)$ for $y \in X_1$ and $v'_y \theoryeq{V} v(2,4)$ for $y \in X_2$.
We can now apply essential uniqueness, and use property~\ref{ax:pprolif} to conclude that both:
\begin{align}
1 \theoryeq{V} v'_y[1/3, 1/4] \quad\vee\quad & 2 \theoryeq{V} v'_y[1/3, 1/4]\label{eq:condition4-plotkin3}\\
1 \theoryeq{V} v'_y[2/3, 2/4] \quad\vee\quad & 2 \theoryeq{V} v'_y[2/3, 2/4]\label{eq:condition5-plotkin3}.
\end{align}
For $y \in X_1$, however:
\begin{align*}
v'_y[2/3, 2/4] & \theoryeq{V} v(1,3)[2/3, 2/4] \\
& = v(1,2).
\end{align*}
So, to satisfy equation~\eqref{eq:condition5-plotkin3}, we conclude that $v(1,2) \theoryeq{V} 1$ or $v(1,2) \theoryeq{V} 2$. By property~\ref{ax:vvar}, this means that $\{1\} \vdash v(1,2)$ or $\{2\} \vdash v(1,2)$. This contradicts property~\ref{ax:vbinary}. And so we must conclude that $X_1 = \emptyset$.
Similarly, for $y \in X_2$:
\begin{align*}
v'_y[1/3, 1/4] & \theoryeq{V} v(2,4)[1/3, 1/4] \\
& = v(2,1).
\end{align*}
In order to satisfy equation~\ref{eq:condition4-plotkin3}, we must have $v(2,1) \theoryeq{V} 1$ or $v(2,1) \theoryeq{V} 2$. By property~\ref{ax:vvar}, this means that $\{1\} \vdash v(2,1)$ or $\{2\} \vdash v(2,1)$, which contradicts property~\ref{ax:vbinary}. And so we must also conclude that $X_2 = \emptyset$.
Therefore:
\begin{align*}
  \var(p') & = X_1 \cup X_2 \\
   & = \emptyset \cup \emptyset \\
   & = \emptyset.
\end{align*}
However, from Equation~\eqref{eq:p=p'} we know: $p[g_1] \theoryeq{P} p'[g_2]$, and from property~\ref{ax:pbinary-3} we know that:
\[
\neg (\emptyset \vdash p'[g_2]).
\]
And so, $\var(p')$ cannot be empty. Contradiction! Hence no composite theory of $\mathbb{P}$ after $\mathbb{V}$ can exist.
\end{proof}

\begin{rem}
The proofs in this section require different types of substitutions, and this difference impacts their scope of application. The proofs of \refthms~\ref{thm:plotkin-1} and~\ref{thm:plotkin-2} only require variable-for-variable substitutions, and actually preclude the existence of distributive laws for pointed endofunctors, generalizing~\cite[Theorem 2.4]{KlinSalamanca2018}. The proof of \refthm~\ref{thm:plotkin-3} requires more complex substitutions, implicitly assuming the multiplication axioms. Therefore, \refthm~\ref{thm:plotkin-3} only applies to distributive laws between monads.
\end{rem}

\begin{exa}
  If we consider the algebraic theory of an idempotent binary operation, \refthm~\ref{thm:plotkin-3} shows that the induced monad cannot distribute over itself. This remains true if we add either units or associativity, showing various non-commutative variants of non-determinism cannot be distributed over themselves.

  Similarly, if we denote any of these monads by~$T$, there is no distributive law~$D \circ T \Rightarrow T \circ D$, where~$D$ is the distribution monad.
\end{exa}
\begin{nonexample}[Reader Monad]
  The presentation of the binary reader monad given in~\refex~\ref{ex:reader-monad-presentation} satisfies:
  \[
    x = (x * y) * (z * x)
  \]
   However, this theory has no term~$p$ satisfying both axioms~\ref{ax:pidem} and~\ref{ax:pprolif}. So as expected, we cannot apply~\refthm~\ref{thm:plotkin-3} to the binary reader monad.
\end{nonexample}

\subsection{Regarding Idempotence}\label{sec:plotkin-idempunit}

All theorems so far rely heavily on idempotent terms in both theories. The main advantage idempotent terms provide for our proofs is that they can be reduced to a variable in a controlled way. Idempotence is, however, not the only algebraic property with this effect. Compare idempotence:
\[
x * x = x
\]
to unitality:
\[
x * 1 = x
\]

Both idempotence and unitality have $x$ as the only variable appearing on either side of the equation, and both equations reduce a more complicated term to a single variable. We can capture this behaviour in slightly more general terms, namely: \emph{``There is a term $t$ and a substitution $f$ such that for any variable $x \in \var(t)$, $t[f(x')/x' \neq x] = x $''.} For the idempotence and unitality equations, the term $t$ would be $x * y$, and the substitution would be $f(y) = x$ for idempotence and $f(y) = 1$ for unitality.

This generalisation of the idempotence equation leads us to \refthm~\ref{thm:NoGo-Videm-Punit}. The proof technique used for \refthm~\ref{thm:plotkin-1} still works with this more general assumption, although extra care needs to be taken when substituting a $\bb{T}$-term into an $\bb{S}$-term, as this could turn a previously separated term into a term that is no longer separated. As in the previous proofs, we need extra assumptions such as commutativity to make the proof go through. Notice that we only generalise one of the two idempotent terms. For the other term, the current proof method requires the more specific properties of the idempotence equation.

\begin{rem}\label{rem:NoGo-Videm-Punit-narymary}
  In \refthm~\ref{thm:NoGo-Videm-Punit} below we state and prove the theorem for \emph{binary} terms. This is to make it easier to see where and how the more general assumption replaces the assumption of idempotency in the proof. By copying the strategy from \refthm~\ref{thm:plotkin-2}, however, it is straightforward to generalise \refthm~\ref{thm:NoGo-Videm-Punit} to the case where $p$ is an $m$-ary term and $v$ and $n$-ary term.
\end{rem}

\begin{thm}%
  \label{thm:NoGo-Videm-Punit}
  Let~$\mathbb{P}$ and~$\mathbb{V}$ be two algebraic theories, $T_{\mathbb{P}}$ a stable universal set of~$\mathbb{P}$-terms, and~$T_{\mathbb{V}}$ a stable universal set of $\mathbb{V}$-terms.
  If there are terms:
  \[
    2 \vdash_{\mathbb{P}} p \qquad\text{ and }\qquad 2 \vdash_{\mathbb{V}} v
  \]
  such that:
  \begin{enumerate}[label=(P\arabic*)]
  \item\label{ax:pcomm-dan} $p$ is commutative:
    \[
      2 \vdash p(1,2) \theoryeq{P} p(2,1)
    \]
  \item\label{ax:punit} There is a substitution $f_p: \var(p) \rightarrow T_{\mathbb{P}}$, such that:
    \[
     \Gamma \vdash p(1,f_p(2)) \theoryeq{P} 1.
     \]
  \item\label{ax:pprolif-3} For all~$n \vdash p' \in T_{\mathbb{P}}$:
   \[
     n \vdash p(1,2) \theoryeq{P} p' \quad\Rightarrow\quad 2 \vdash p'
   \]
  \end{enumerate}
  \begin{enumerate}[label=(V\arabic*)]
  \item\label{ax:videm-dan} $v$ is idempotent:
    \[
      1 \vdash v(1,1) \theoryeq{V} 1
    \]
  \item\label{ax:vvar-3} For all~$v' \in T_{\mathbb{V}}$, and each variable~$x$:
    \[
      \Gamma \vdash x \theoryeq{V} v' \quad\Rightarrow\quad \{x\} \vdash v'
    \]
  \item\label{ax:vbinary-3} For all~$v' \in T_{\mathbb{V}}$:
    \[
      \Gamma \vdash v(1,2) \theoryeq{V} v' \quad\Rightarrow\quad \neg (\{ 1 \} \vdash v'\vee \{ 2 \} \vdash v')
    \]
  \end{enumerate}
  Then there is no composite theory of~$\mathbb{P}$ after~$\mathbb{V}$.
\end{thm}

\begin{proof}
  Assume by way of contradiction that a composite theory~$\bb{U}$ of~$\bb{P}$ after~$\bb{V}$ exists. Then as~$\bb{U}$ is composite, the separation axiom tells us that there exist~$X \vdash p'$ and~$\Gamma \vdash v'_x$ for each $x \in X$ such that:
  \begin{equation}
    \Gamma \vdash v(p(1,2),p(3,4)) \theoryeq{U} p'[v'_x/x].
  \end{equation}
We make the following substitution of variables:
\begin{align*}
f_1(1) & = 1 & f_1(3) & = 1 \\
f_1(2) & = 2 & f_1(4) & = 2
\end{align*}
This yields:
\begin{align*}
   & v(p(1,2),p(1,2)) \theoryeq{U} p'[v'_x[f_1]/x] \\
 \Leftrightarrow\;\; & \reason{v \text{ is idempotent}} \\
  & p(1,2) \theoryeq{U} p'[v'_x[f_1]/x].
\end{align*}
By essential uniqueness and assumption~\ref{ax:pprolif-3}, we conclude that for all $x$:
\begin{equation}\label{eq:1or2}
v'_x[f_1] \theoryeq{V} 1 \quad \vee \quad v'_x[f_1] \theoryeq{V} 2.
\end{equation}
So, using assumption~\ref{ax:vvar-3}:
\begin{equation}\label{eq:13or24}
\{1,3\} \vdash v'_x \quad \vee \quad \{2,4\} \vdash v'_x.
\end{equation}
We make a second substitution:
\begin{align*}
f_2(1) & = 1 & f_2(3) & = 2 \\
f_2(2) & = 2 & f_2(4) & = 1
\end{align*}
This yields:
\begin{align*}
   & v(p(1,2),p(2,1)) \theoryeq{U} p'[v'_x[f_2]/x] \\
 \Leftrightarrow\;\; & \reason{p \text{ is commutative}} \\
   & v(p(1,2),p(1,2)) \theoryeq{U} p'[v'_x[f_2]/x] \\
 \Leftrightarrow\;\; & \reason{v \text{ is idempotent}} \\
  & p(1,2) \theoryeq{U} p'[v'_x[f_2]/x].
\end{align*}
By essential uniqueness and assumption~\ref{ax:pprolif-3}, we conclude that for all $x$:
\[
v'_x[f_2] \theoryeq{V} 1 \quad \vee \quad v'_x[f_2] \theoryeq{V} 2.
\]
So, using assumption~\ref{ax:vvar-3}:
\begin{equation}\label{eq:14or23}
\{1,4\} \vdash v'_x \quad \vee \quad \{2,3\} \vdash v'_x.
\end{equation}
Taking Equations~\eqref{eq:13or24} and~\eqref{eq:13or24} together, we conclude for all x:
\begin{equation}\label{eq:1or2or3or4}
  \bigvee_{i\in4} \{i\} \vdash v'_x.
\end{equation}
This implies that for each $x$, $v'_x$ is equal to a variable:
\begin{itemize}
\item If $\{1\} \vdash v'_x$, then $v'_x[f_1] = v'_x$. From Equation~\ref{eq:1or2} we know that $v'_x[f_1] \theoryeq{V} 1$ or $v'_x[f_1] \theoryeq{V} 2$. And so also $v'_x \theoryeq{V} 1$ or $v'_x \theoryeq{V} 2$.
\item If $\{2\} \vdash v'_x$, then also $v'_x[f_1] = v'_x$, and so again $v'_x \theoryeq{V} 1$ or $v'_x \theoryeq{V} 2$.
\item If $\{3\} \vdash v'_x$, then $v'_x[f_1][3/1, 4/2] = v'_x$. From Equation~\ref{eq:1or2} we know that \\ $v'_x[f_1][3/1, 4/2] \theoryeq{V}~3$ or $v'_x[f_1][3/1, 4/2] \theoryeq{V} 4$. And so also $v'_x \theoryeq{V} 3$ or $v'_x \theoryeq{V} 4$.
\item If $\{4\} \vdash v'_x$, then also $v'_x[f_1][3/1, 4/2] = v'_x$, and so again $v'_x \theoryeq{V} 3$ or $v'_x \theoryeq{V} 4$.
\end{itemize}
We make a final substitution, using the substitution $f_p$ from property~\ref{ax:punit}::
\begin{align*}
f_3(1) & = 1 & f_3(3) & = 2 \\
f_3(2) & = f_p(2) & f_3(4) & = f_p(4)
\end{align*}
This yields:
\begin{align*}
   & v(p(1,f_p(2)),p(3,f_p(4))) \theoryeq{U} p'[v'_x[f_3]/x] \\
 \Leftrightarrow\;\; & \reason{\text{property~\ref{ax:punit}}} \\
  & v(1,3) \theoryeq{U} p'[v'_x[f_3]/x] \\
  \Rightarrow\;\; & \reason{\text{ clarifying that } v(1,3) \text{ is a separated term}} \\
  & 5[v(1,3)/5] \theoryeq{U} p'[v'_x[f_3]/x].
\end{align*}
Notice that the term $p'[v'_x[f_3]/x]$ is separated: since every $v'_x$ is just a variable, the substitution $f_3$ does not break separation, even though it might insert a $\bb{P}$ term. We apply essential uniqueness: there are substitutions $g_1: \{5\} \rightarrow Z, g_2: X \rightarrow Z$ such that:
\begin{align*}
5[g_1] & \theoryeq{P} p'[g_2] \\
g_1(5) = g_2(x) & \Leftrightarrow v(1,3) \theoryeq{V} v'_x.
\end{align*}
By consistency of $\bb{P}$, we know that there is at least one $x$ such that $g_1(5) = g_2(x)$. And hence there is at least one $v'_x$ such that $v'_x \theoryeq{V} v(1,3)$. But this contradicts assumption~\ref{ax:vbinary-3}, since for all $x$, $v'_x$ is equal to a variable. We conclude that no composite theory of $\bb{P}$ after $\bb{V}$ can exist.
\end{proof}

Theorem~\ref{thm:NoGo-Videm-Punit} precludes even more distributive laws:

\begin{exa}[Multiset and Powerset Monad: Filling in the Gap]
 The theory of commutative monoids, presenting the multiset monad, does not have an idempotent term. Therefore, it has so far been unaffected by our no-go theorems. It does, however, have a unital term, bringing it in scope of \refthm~\ref{thm:NoGo-Videm-Punit}. From Manes and Mulry~\cite[Theorem 4.3.4]{ManesMulry2007} we know that there are distributive laws $M \circ M \Rightarrow M \circ M$ and $M \circ P \Rightarrow P \circ M$, where $M$ is the multiset monad and $P$ the powerset monad. We already know that the powerset monad does not distribute over itself, which leaves the combination $P \circ M \Rightarrow M \circ P$. Theorem~\ref{thm:NoGo-Videm-Punit} fills this gap, showing there is no distributive law of that type.
\end{exa}

\section{No-Go Theorems Beyond Plotkin}\label{section:beyondPlotkin} 
So far, all our impossibility results involve at least one monad with an idempotent term in its corresponding algebraic theory. But the absence of an idempotent term does not guarantee the existence of a distributive law. Consider the list monad for example. This monad is quite similar to the multiset monad, and we observed in~\refex~\ref{ex:multisetdistlaw} that the ``times over plus'' law of~\refeqn~\eqref{timesoverplus} induces a distributive law for the multiset monad over itself. If we assume this also yields a distributive law for the list monad over itself, then from one of the multiplication axioms:
\[
\lambda ([[a,b],[c,d]]) = [[a,c],[a,d],[b,c],[b,d]],
\]
whilst from the other:
\[
\lambda ([[a,b],[c,d]]) = [[a,c],[b,c],[a,d],[b,d]].
\]
These two statements are incompatible, so the list monad cannot distribute over itself in this way. However, not all distributive laws resemble the distributivity of times over plus, so from this observation alone we cannot yet rule out the possibility of a distributive law for the list monad over itself. In fact, Manes and Mulry found three other distributive laws for the \emph{non-empty} list monad over itself~\cite[Example 5.1.9]{ManesMulry2007},~\cite[Example 4.10]{ManesMulry2008}. But despite these distributive laws being good candidates, they fail to extend to distributive laws for the full list monad over itself.

The results in this section build up towards a proof that shows the search for a distributive law for the list monad over itself is futile; no such law exists. As in the previous section, we state our results in general terms, so that they do not only apply to the list monad, but to any monad satisfying the conditions of the theorems.

Firstly, we show in \refprop~\ref{propnary} that under certain circumstances, constants of one algebraic theory act as `multiplicative zeroes' for terms of the other theory. One application is that any putative distributive law of the list monad over itself must satisfy:
\[
\lambda [L_1,\ldots,L_n] = [\;] \text{ if there is an } i \text{ s.t. } L_i = [\;].
\]
When the theory has multiple constants, this observation immediately leads to a contradiction, and hence give us no-go~\refthm~\ref{thm:nogoConstants}.

Building on \refprop~\ref{propnary}, we then derive conditions under which a distributive law has to behave like the distributivity of times over plus~\eqref{timesoverplus}, resulting in \refthm~\ref{thm:times-over-plus}. As a consequence, monads satisfying the conditions of this theorem compose via a unique distributive law, if they compose at all. Again this theorem applies to the list monad, meaning there can be at most one distributive law for the list monad over itself. 

Finally, in \refthms~\ref{thm:nogoTreeList} and~\ref{thm:nogoTreeIdem} we identify properties that together with \refthm~\ref{thm:times-over-plus} provide two more no-go theorems: \refthm~\ref{thm:nogoTreeList} and \refthm~\ref{thm:nogoTreeIdem}. Whereas idempotence was the main property of interest for the no-go theorems in \refsec~\ref{section:Plotkin}, the focus now becomes unitality equations. In addition, the~\define{abides} equation ({\bf ab}ove-bes{\bf ides},~\cite{Bird1988}) will be important for \refthm~\ref{thm:nogoTreeList}:
\begin{equation}\label{eq:abides}
  (a * b) * (c * d) = (a * c) * (b * d).
\end{equation}
We will require that this equation does \emph{not} hold. This is made precise in \refpro~\ref{ax:tspecialproperty} below. From this theorem we can conclude that there is no distributive law for the list monad over itself: the lacking of the abides property is exactly what causes the problem identified at the beginning of this section.

Throughout this section, we will consider two algebraic theories~$\bb{S}$ and~$\bb{T}$. For~$\bb{S}$ we identify the following properties:
\begin{enumerate}[label=(S\arabic*)]
  \item\label{ax:svar0} For any two terms $s', s''$:
      \[
      \emptyset \vdash s' \; \wedge \; \Gamma \vdash s' \theoryeq{S} s'' \quad\Rightarrow\quad \emptyset \vdash s''.
      \]
  \item\label{ax:svar1} For any term $s'$ and variable $x$:
  \[
  \Gamma \vdash s' \theoryeq{S} x \quad\Rightarrow\quad \{x\} \vdash s'.
  \]
  \item\label{ax:newsnary} $\bb{S}$ has an $n$-ary term $\specialopS$ ($n\geq 1$), for which there is a substitution $f: \var(\specialopS) \rightarrow \bb{S}$ such that for any $x \in \var(\specialopS)$:
   \[
   \Gamma \vdash \specialopS[f(y)/y \neq x] \theoryeq{S} x.
   \]
   In addition, we require that the terms $f(y)$ do not contain the variable $x$.
   \item\label{ax:newsnary-all} For any n-ary term $s'$ ($n\geq 1$), there is a substitution $f: \var(s') \rightarrow \bb{S}$ such that for any $x \in \var(s')$:
   \[
   \Gamma \vdash s'[f(y)/y \neq x] \theoryeq{S} x,
   \]
   where the terms $f(y)$ do not contain the variable $x$.
  \item $\bb{S}$ has a binary term $\specialopS$ such that:
   \begin{enumerate}
     \item\label{ax:sbinary} $e_\specialopS$ is a unit for $\specialopS$:
      \[
      \{x\} \vdash \specialopS(x,e_\specialopS) \theoryeq{S} x \theoryeq{S} \specialopS(e_\specialopS,x).
      \]
     \item\label{ax:sidem} $\specialopS$ is idempotent:
     \[
     \{x\} \vdash \specialopS(x,x) \theoryeq{S} x.
     \]
   \end{enumerate}
\end{enumerate}
And for $\bb{T}$:
\begin{enumerate}[label=(T\arabic*)]
  \item\label{ax:newtoomanyconstantsreq} For all terms $X \vdash t'$, constant $e_\bb{T}$, and any variable substitution $f: X \rightarrow Y$:
  \[
  Y \vdash t'[f] \theoryeq{T} e_\bb{T} \quad\Rightarrow\quad X \vdash t' \theoryeq{T} e_\bb{T}.
  \]
  \item\label{ax:tvar0} For any two terms $t', t''$:
      \[
      \emptyset \vdash t' \; \wedge \; \Gamma \vdash t' \theoryeq{T} t'' \quad\Rightarrow\quad \emptyset \vdash t''.
      \]
  \item\label{ax:tvar1} For any term $t'$ and variable $x$:
  \[
  \Gamma \vdash t' \theoryeq{T} x \quad\Rightarrow\quad \{x\} \vdash t'.
  \]
  \item\label{ax:tconst} $\bb{T}$ has a constant $e_\bb{T}$.
  \item $\bb{T}$ has a binary term $\specialopT$ such that:
    \begin{enumerate}
      \item\label{ax:tunit} $e_\specialopT$ is a unit for $\specialopT$:
      \[
      \{x\} \vdash \specialopT(x,e_\specialopT) \theoryeq{T} x \theoryeq{T} \specialopT(e_\specialopT,x).
      \]
      \item\label{ax:tspecialproperty} The abides equation does not hold in $\bb{T}$:
   \begin{align*}
    & \Gamma  \vdash \specialopT(\specialopT(x,y),\specialopT(z,w)) \theoryeq{T} \specialopT(\specialopT(x,z),\specialopT(y,w)) \\
    & \Rightarrow\;  3 \vdash \specialopT(\specialopT(x,y),\specialopT(z,w)).
   \end{align*}
   \end{enumerate}
\end{enumerate}
\begin{rem}[Interpretation of Axioms]
  The properties~\ref{ax:svar0},~\ref{ax:svar1},~\ref{ax:tvar0},~\ref{ax:tvar1} are all constraints on the variables appearing in terms. \ref{ax:svar0} and~\ref{ax:tvar0} read: ``Any term provably equal to a constant cannot have any variables itself''. This is, for example, not the case for any theory involving multiplicative zeroes. In the theory of rings, $0 * x = 0$, and since the term $0 * x$ has a variable, it does not satisfy~\ref{ax:svar0}/\ref{ax:tvar0}. 

  \ref{ax:svar1} and~\ref{ax:tvar1} read: ``Any term provably equal to a variable only contains that single variable''. Idempotent terms are examples of terms that equal a variable, satisfying this condition. In theories with absorption axioms, such as the equation $x \vee (x \wedge y) = x$ from bounded lattices, properties~\ref{ax:svar1}/\ref{ax:tvar1} do not hold. 

  Properties~\ref{ax:newsnary} and~\ref{ax:newsnary-all} are generalizations of unital equations. They require that terms can be reduced to variables via a suitable substitution. However, contrary to the similar requirement~\ref{ax:punit} in \refthm~\ref{thm:NoGo-Videm-Punit}, idempotence is not an instance of properties~\ref{ax:newsnary} and~\ref{ax:newsnary-all}. Idempotence requires the substitution to change all variables to $x$: $\specialopS(x,y)[x/y] = \specialopS(x,x) = x$, which is not allowed in this case. Unitality does not use the resulting variable in its substitutions, so unitality is an instance of these properties: $\specialopS(x,y)[e_s/y] = \specialopS(x,e_s) = x$. The difference between~\ref{ax:newsnary} and~\ref{ax:newsnary-all} is the quantifier.

  Property~\ref{ax:newtoomanyconstantsreq} is a weaker version of~\ref{ax:tvar0}, focussing on the provability of an equality between a term and a constant, rather than restricting the variables appearing in that term. It reads: ``If a variable substitution of term $t'$ is provably equal to a constant, then $t'$ is already provably equal to that constant.'' The usefulness of this property compared to~\ref{ax:tvar0} is that is allows for equations such as $0 * x = 0$, which were forbidden by~\ref{ax:tvar0}.

  Notice that property~\ref{ax:tvar0} implies~\ref{ax:newtoomanyconstantsreq}: if $t'[f] = e_\bb{T}$ then~\ref{ax:tvar0} requires $t'[f]$ to contain no variables. Since $f$ is a variable substitution, this means $t'$ already had no variables, and hence $t' = t'[f] = e_\bb{T}$, so~\ref{ax:newtoomanyconstantsreq} is satisfied. This justifies our claim that one is a weaker version of the other.

  Also notice that~\ref{ax:sbinary} and~\ref{ax:sidem} both imply~\ref{ax:newsnary}.
\end{rem}

\begin{exa}
If a theory $\bb{S}$ has a presentation in which all operations are either idempotent or have a unit, then an easy induction shows that $\bb{S}$ will satisfy~\ref{ax:newsnary-all}. More precisely, if for every $s' \in \Sigma^\bb{S}$, either:
  \begin{itemize}
  \item $s'$ is idempotent, that is: $s'[x/y \neq x] \theoryeq{S} x$, or:
  \item $s'$ has a unit $e_{s'}$: $s'[e_{s'}/y \neq x] \theoryeq{S} x$,
  \end{itemize}
  then $\bb{S}$ satisfies~\ref{ax:newsnary-all}.
\end{exa}

\begin{exa}[Algebraic Properties of Key Monads]\label{ex:STproperties}
  $\;$
  \begin{itemize}
    \item The list monad, presented by the theory of monoids, satisfies~\ref{ax:svar0} and~\ref{ax:svar1}. The monoid multiplication satisfies~\ref{ax:sbinary} and hence also~\ref{ax:newsnary} and even~\ref{ax:newsnary-all},  but the theory of monoids does not satisfy~\ref{ax:sidem}.
        The equation~$(x * y) * (z * w) = (x * z) * (y * w)$ holds in the theory of monoids if and only if $y = z$, and so it satisfies all of~\ref{ax:newtoomanyconstantsreq},~\ref{ax:tvar0},~\ref{ax:tvar1},~\ref{ax:tconst},~\ref{ax:tunit}, and~\ref{ax:tspecialproperty}.
    \item The powerset monad is presented by the theory of join semilattices, which satisfies~\ref{ax:sidem} in addition to~\ref{ax:svar0},~\ref{ax:svar1},~\ref{ax:newsnary},~\ref{ax:newsnary-all}, and~\ref{ax:sbinary}. However, this theory does not have property~\ref{ax:tspecialproperty} as the join is commutative and associative and so satisfies the abides equation~\eqref{eq:abides}. Properties~\ref{ax:newtoomanyconstantsreq},~\ref{ax:tvar0},~\ref{ax:tvar1},~\ref{ax:tconst}, and~\ref{ax:tunit} still hold.
    \item The exception monad corresponds to an algebraic theory with a signature containing constants for each exception, and no axioms. It satisfies~\ref{ax:svar0},~\ref{ax:svar1}, and~\ref{ax:newsnary-all}. It does not satisfy~\ref{ax:newsnary},~\ref{ax:sbinary}, and~\ref{ax:sidem} as there are no binary terms. Similarly, it satisfies~\ref{ax:newtoomanyconstantsreq},~\ref{ax:tvar0},~\ref{ax:tvar1}, and~\ref{ax:tconst}, but not~\ref{ax:tunit} or~\ref{ax:tspecialproperty}.
  \end{itemize}
\end{exa}

\subsection{Multiplicative Zeroes}\label{section:persistentconstants}
Our first focus is on properties~\ref{ax:newsnary} and~\ref{ax:tconst}. The goal is to prove that in a composite of theories~$\bb{S}$ and~$\bb{T}$, the constant $e_\bb{T}$ behaves like a multiplicative zero, consuming any $\bb{S}$-term it appears in.

\begin{prop}\label{propnary}
Let $\bb{S}$ be an algebraic theory satisfying property~\ref{ax:newsnary} for term $\specialopS$ in $\bb{S}$, and $\bb{T}$ a theory satisfying properties~\ref{ax:newtoomanyconstantsreq} and~\ref{ax:tconst}, giving constant $e_\bb{T}$.
If $\bb{U}$ is a composite theory of $\bb{T}$ after $\bb{S}$, then we must have that, for any $x_i \in \var(\specialopS)$, $\specialopS[e_\bb{T}/x_i] \theoryeq{U} e_\bb{T}$.
\end{prop}

%

\begin{proof}
  The statement is trivial in the case that~$\specialopS$ has only one free variable, because property~\ref{ax:newsnary} implies that we must have $\specialopS(x) \theoryeq{S} x$. We may therefore assume that $\specialopS$ has at least two free variables.

  In a composite theory, every term is equal to a separated term. So there is a $t'$ and there are $s'_{x}$ such that:
  \[ \specialopS[e_\bb{T}/x_i] \theoryeq{U} t'[s'_{x}/x]. \]
  Using the substitution $f$ given by property~\ref{ax:newsnary}:
  \begin{align*}
  & \specialopS[e_\bb{T}/x_i][f(y)/y \neq x_i] \theoryeq{U} t'[s'_{x}[f(y)/y \neq x_i]/x] \\
  \Rightarrow \;\; & \reason{ \text{axiom~\ref{ax:newsnary}} } \\
  & e_\bb{T} \theoryeq{U} t'[s'_{x}[f(y)/y \neq x_i]/x].
  \end{align*}
  We now have two separated terms equal to each other, so we can use essential uniqueness to conclude that there are variable substitutions $g_1, g_2$ such that:
  \[
  e_\bb{T}[g_1] \theoryeq{T} t'[g_2].
  \]
  Since $g_2$ is a variable substitution, we can apply assumption~\ref{ax:newtoomanyconstantsreq}, and conclude that $t' \theoryeq{T} e_\bb{T}$. Going back to our original equation:
  \begin{align*}
  & \specialopS[e_\bb{T}/x_i] \theoryeq{U} t'[s'_{x}/x] \\
  \Rightarrow \;\; & \reason{ t' \theoryeq{T} e_\bb{T} } \\
  & \specialopS[e_\bb{T}/x_i] \theoryeq{U} e_\bb{T}[s'_{x}/x] \\
  \Rightarrow \;\; & \reason{ e_\bb{T} \text{ has no variables} } \\
  & \specialopS[e_\bb{T}/x_i] \theoryeq{U} e_\bb{T},
  \end{align*}
  which is what we needed to show.
\end{proof}

\begin{rem}\label{Prop1}
If $\bb{S}$ satisfies property~\ref{ax:newsnary-all}, that is, all terms in $\bb{S}$ with $\geq 1$ variables satisfy~\ref{ax:newsnary}, then the constant $e_\bb{T}$ in $\bb{T}$ annihilates any $\bb{S}$-term it appears in. We say that $e_\bb{T}$ acts as a multiplicative zero.
\end{rem}

If $\bb{T}$ has more than one constant, this could lead to inconsistencies. Our next no-go theorem makes this precise:
\begin{thm}[No-Go Theorem: Too Many Constants]\label{thm:nogoConstants}
  Let $\bb{S}$ and $\bb{T}$ be algebraic theories with properties~\ref{ax:newsnary} and~\ref{ax:newtoomanyconstantsreq} respectively. Further assume that the term $\specialopS$ satisfying~\ref{ax:newsnary} has at least two free variables. Then there exists no composite theory of $\bb{T}$ after $\bb{S}$ if $\bb{T}$ has two or more constants.
\end{thm}

\begin{proof}
  Suppose that $\bb{U}$ is a composite theory of $\bb{T}$ after $\bb{S}$ and let $e_1$ and $e_2$ be distinct constants in $\bb{T}$. Suppose that $\{x,y\} \subseteq \var(\specialopS)$. Then by \refprop~\ref{propnary} we have:
  \[
  e_1 \theoryeq{U} s[e_1/x, e_2/y] \theoryeq{U} e_2.
  \]
  By essential uniqueness, we may conclude that $e_1 \theoryeq{T} e_2$. Contradiction. So $\bb{U}$ cannot be a composite of $\bb{T}$ after $\bb{S}$.
\end{proof}

\begin{exa}[Iterated Distributive Laws]\label{ex:iterateddistlawsnogo}
Theorem~\ref{thm:nogoConstants} is remarkably useful for determining whether iterated distributive laws are possible. Cheng shows in her paper \emph{Iterated distributive laws}~\cite{Cheng2011b} that three or more monads can be composed if there are distributive laws for the pairwise compositions of the monads, and these distributive laws additionally satisfy the Yang-Baxter equation. Theorem~\ref{thm:nogoConstants} approaches the question of iterated distributive laws from the other end, by severely limiting the possibilities. If in a proposed composition of monads $C \circ B \circ A$, the monads $B$ and $C$ each have a constant (and the other mild requirements of \refthm~\ref{thm:nogoConstants} are satisfied), then there is no distributive law $A \circ (C \circ B) \Rightarrow (C \circ B) \Rightarrow A$, and hence any possible pairwise distributive laws will not satisfy Yang-Baxter.

We give a few concrete examples involving the list, multiset and powerset monads $L, M, P$, whose algebraic theories are monoids, commutative monoids, and join semilattices respectively. We know from Manes and Mulry~\cite[Theorem 4.3.4]{ManesMulry2007} that we can form the monads $M \circ L$, $M \circ M$, and $M \circ P$ via distributive laws. All of these monads have two constants, which satisfy the condition for $\bb{T}$ in \refthm~\ref{thm:nogoConstants}. By picking the term $x*y$ in each of the theories for $L$, $M$, and $P$, we see that these monads satisfy the condition for $\bb{S}$. We can therefore exclude all of the following compositions via distributive laws:
\begin{table}[h]
  \centering
  \caption{Overview of possible distributive laws of type:}\label{tab:iterateddistlawsnogo}
  row $\circ$ column $\Rightarrow$ column $\circ$ row, involving the monads \\ list $(L)$, multiset $(M)$, and powerset $(P)$.

  \belowrulesep=0ex
  \aboverulesep=0ex
  \renewcommand{\arraystretch}{1.1}
  \begin{tabular}{l|ccc} 
    \toprule
      & $M \circ L$ & $M \circ M$ & $M \circ P$ \\ \midrule
    $L$ & \no & \no & \no \\
    $M$ & \no & \no & \no \\
    $P$ & \no & \no & \no \\\bottomrule
  \end{tabular}
\end{table}
\end{exa}

\begin{exa}[An Error in the Literature]%
  \label{ex:error-in-the-literature}
  We saw in the previous example that the term $x*y$ from monoids (the list monad) satisfies the conditions for $\bb{S}$ in \refthm~\ref{thm:nogoConstants}. The exception monad satisfies~\ref{ax:newtoomanyconstantsreq}, so when the exception monad has more than one exception, \refthm~\ref{thm:nogoConstants} states that there is no distributive law $L\circ (-+E) \Rightarrow (-+E) \circ L$.

However, Manes and Mulry claim to have a distributive law of this type for the case where $E = \{a,b\}$~\cite[Example 4.12]{ManesMulry2008}, given by:

\begin{align*}
\lambda [] &= [] & \lambda L & = L \text{ if no element of $L$ is in } E \\
\lambda [e] & = e \text{ for any exception } e \in E & \lambda L & = a \text{ otherwise.}
\end{align*}
We check more concretely that this cannot be a distributive law by showing that it fails the first multiplication axiom from \refdef~\ref{def:distlaw}:
\begin{center}
  \begin{tikzcd}[column sep=small]
  {[[b],[]]} \arrow[d, "\mu^L_{EX}"] \arrow[rr, "L(\lambda_X)"] & & {[b, []]} \arrow[rr, "\lambda_{LX}"] & & a \arrow[d, "E(\mu^L_X)"] \\
  {[b]} \arrow[rrr, "\lambda_X"]& & & b \arrow[r, phantom, "\neq"] & a
  \end{tikzcd}
\end{center}
The given distributive law follows directly from Manes and Mulry's Theorem 4.6~\cite{ManesMulry2008}. We suspect that the problem originates in Lemma 4.5 of this paper. Louis Parlant found that the proof of this lemma might use the isomorphism $(A \otimes I) \cong A$ implicitly if the signature of the theory has constants, while the lemma explicitly does not assume any monoidal properties of its functors. So the lemma, and hence also Theorem 4.6, may not be valid in the case that the theory has constants. In addition, the induction in the proof of Lemma 4.5 starts at $n = 1$, where $n$ is the number of variables appearing in a term. This induction therefore excludes constants, which should be considered separately but are absent from the proof.


It is important to notice that \refthm~\ref{thm:nogoConstants} does not contradict the well-known result that the exception monad distributes over every set monad $T$; that result is for the other direction $(-+E) \circ T \Rightarrow T \circ (-+E)$.
\end{exa}
\begin{nonexample}[Exception Monad]
It is well known that the exception monad distributes over itself. Even though the corresponding theory satisfies properties~\ref{ax:newsnary} and~\ref{ax:newtoomanyconstantsreq}, there are no terms with more than one free variable, and hence \refthm~\ref{thm:nogoConstants} does not apply.
\end{nonexample}

\subsection{The One Distributive Law, If It Exists}
Needing just the properties~\ref{ax:newsnary},~\ref{ax:newtoomanyconstantsreq}, and~\ref{ax:tconst}, \refprop~\ref{propnary} already greatly restricts the possibilities for a distributive law between monads $S$ and $T$. We will now see that if both $\bb{S}$ and $\bb{T}$ have binary terms with units, then in a composite theory, the binary of $\bb{S}$ distributes over the binary of $\bb{T}$ like times over plus in \refeqn~\eqref{timesoverplus}. For the monads corresponding to these theories, this means that there is only one candidate distributive law to consider.

\begin{thm}[Times over Plus Theorem]%
\label{thm:times-over-plus}
Let $\bb{S}$ and $\bb{T}$ be two algebraic theories, satisfying~\ref{ax:svar0},~\ref{ax:svar1},~\ref{ax:newsnary-all} and~\ref{ax:tvar0},~\ref{ax:tvar1},~\ref{ax:tconst} respectively. Assume furthermore that there are terms:
 \[
 2 \vdash_{\mathbb{S}} \specialopS \qquad\text{ and }\qquad 2 \vdash_{\mathbb{T}} \specialopT,
 \]
 satisfying~\ref{ax:sbinary} and~\ref{ax:tunit} respectively. Finally, let $\bb{U}$ be a composite theory of $\bb{T}$ after $\bb{S}$. Then $\specialopS$ distributes over $\specialopT$:
\begin{align}
    \specialopS(\specialopT(y_1,y_2), x_0) & \theoryeq{U} \specialopT(\specialopS(y_1,x_0),\specialopS(y_2,x_0))\label{theoremeq1} \\
    \specialopS(x_0, \specialopT(y_1,y_2)) & \theoryeq{U} \specialopT(\specialopS(x_0,y_1),\specialopS(x_0,y_2)).\label{theoremeq2}
\end{align}
\end{thm}

We derive the distributional behaviour in three stages, relying as always on separation and essential uniqueness in a composite theory. 
Suppose that $t'[s'_x/x]$ is a separated term such that $\specialopS(\specialopT(y_1,y_2), x_0) \theoryeq{U} t'[s'_x/x]$, then we derive the following about $t'[s'_x/x]$:
\begin{enumerate}
  \item First we prove which variables appear in the terms $s'_x$ of the separated term: $\var(s'_x) = \{ y_1, x_0\}$ or $\var(s'_x) = \{y_2,x_0\}$.
  \item Then, we prove that each of the $s'_x$ is either equal to $\specialopS(y_1,x_0)$ or to $\specialopS(y_2,x_0)$.
  \item Finally, we derive that the separated term $t'[s'_x/x]$ has to be equal to $\specialopT(\specialopS(y_1,x_0),\specialopS(y_2,x_0))$.
\end{enumerate}
The proofs of these three stages are quite long. They are separated into different lemmas to make it easier to keep track of the main line of reasoning.
\begin{lem}\label{t_0(s_i)}
  Let $\bb{S}$ and $\bb{T}$ be two algebraic theories satisfying~\ref{ax:svar0},~\ref{ax:svar1},~\ref{ax:newsnary-all}, and~\ref{ax:tvar0},~\ref{ax:tconst} respectively. Assume furthermore that there are terms:
 \[
 2 \vdash_{\mathbb{S}} \specialopS \qquad\text{ and }\qquad 2 \vdash_{\mathbb{T}} \specialopT,
 \]
 satisfying~\ref{ax:sbinary} and~\ref{ax:tunit} respectively. Finally, let $\bb{U}$ be a composite theory of $\bb{T}$ after $\bb{S}$. Then there is a $\bb{T}$-term $X \vdash t'$ and there is a family of $\bb{S}$-terms $s'_x, x \in X$ such that:
  \[
    \specialopS(\specialopT(y_1,y_2), x_0) \theoryeq{U} t'[s'_x/x],
  \]
  and for each $x \in \var(t')$:
  \[
  \var(s'_x) = \{y_1,x_0\} \qquad \text{or} \qquad \var(s'_x) = \{y_2,x_0\}.
  \]
  Moreover, there is an $x$ such that $\var(s'_x) = \{y_1,x_0\}$ and an $x$ such that $\var(s'_x) = \{y_2,x_0\}$.

  Similarly, there is a $\bb{T}$-term $X' \vdash t''$ and there is a family of $\bb{S}$-terms $s''_{x'}, x' \in X'$ such that:
  \[
    \specialopS(x_0,\specialopT(y_1,y_2)) \theoryeq{U} t''[s''_{x'}/x'],
  \]
  and for each $x' \in \var(t'')$:
  \[
  \var(s''_{x'}) = \{y_1,x_0\} \qquad \text{or} \qquad \var(s''_{x'}) = \{y_2,x_0\}.
  \]
  Moreover, there is an $x'$ such that $\var(s''_{x'}) = \{y_1,x_0\}$ and an $x'$ such that $\var(s''_{x'}) = \{y_2,x_0\}$.
\end{lem}

\begin{proof}
We only explicitly prove the statements for $\specialopS(\specialopT(y_1,y_2), x_0)$. The proof for $\specialopS(x_0, \specialopT(y_1,y_2))$ is similar.

From the fact that $\bb{U}$ is a composite of the theories $\bb{S}$ and $\bb{T}$, we know that every term in $\bb{U}$ is equal to a separated term. And so, there is a $X \vdash t'$ and there is a family $s'_x, x \in X$ such that:
\begin{align}\label{lambda_2}
\specialopS(\specialopT(y_1,y_2), x_0) \theoryeq{U} t'[s'_x/x].
\end{align}

We substitute $x_0 \mapsto e_\specialopS$ in \refeqn~\eqref{lambda_2}. This yields:
\begin{align}
& \specialopS(\specialopT(y_1,y_2), e_\specialopS) \theoryeq{U} t'[s'_x[e_\specialopS/x_0]/x] \nonumber \\
\Rightarrow \;\; & \reason{e_\specialopS \text{ is the unit of } \specialopS} \nonumber \\
& \specialopT(y_1,y_2) \theoryeq{U} t'[s'_x[e_\specialopS/x_0]/x].\label{t-knowledge2}
\end{align}
By the essential uniqueness property, we conclude that there are functions:
\[
  \essuniqfunctionA: \{y_1, y_2\} \rightarrow Z, \quad \essuniqfunctionB: X \rightarrow Z,
\]
such that:
\begin{align}\label{t-knowledge1}
& \specialopT(y_1, y_2)[\essuniqfunctionA] \theoryeq{T} t'[\essuniqfunctionB].
\end{align}
Furthermore, whenever $\essuniqfunctionA(y_1) = \essuniqfunctionB(x)$ or $\essuniqfunctionA(y_2) = \essuniqfunctionB(x)$, we have respectively:
\begin{align}
  y_1 & \theoryeq{S} s'_x[e_\specialopS/x_0]\label{s-knowledge1} \\
  y_2 & \theoryeq{S} s'_x[e_\specialopS/x_0].\label{s-knowledge2}
\end{align}
Since we assume variables $y_1$ and $y_2$ to be distinct, essential uniqueness also gives us $\essuniqfunctionA(y_1) \neq \essuniqfunctionA(y_2)$.

We analyse \refeqn~\eqref{t-knowledge1} more closely, comparing the variables appearing in both $\specialopT(y_1, y_2)[\essuniqfunctionA]$ and $t'[\essuniqfunctionB]$. First, we show that $\{\essuniqfunctionB(x) \given x \in \var(t')\} \subseteq \{\essuniqfunctionA(y_1), \essuniqfunctionA(y_2)\}$. This follows from the following equalities:
\begin{align*}
  & e_\specialopT \\
\theoryeq{T}\; & \reason{e_\specialopT \text{ is the unit for } \specialopT} \\
  & \specialopT(e_\specialopT, e_\specialopT) \\
\theoryeq{T}\; & \reason{\text{Substitution}} \\
  & \specialopT(\essuniqfunctionA(y_1), \essuniqfunctionA(y_2))[e_\specialopT/\essuniqfunctionA(y_1), e_\specialopT/\essuniqfunctionA(y_2)] \\
\theoryeq{T}\; & \reason{\text{Equation }\eqref{t-knowledge1}} \\
  & t'[\essuniqfunctionB][e_\specialopT/\essuniqfunctionA(y_1), e_\specialopT/\essuniqfunctionA(y_2)].
  \end{align*}
So: $e_\specialopT \theoryeq{T} t'[\essuniqfunctionB][e_\specialopT/\essuniqfunctionA(y_1), e_\specialopT/\essuniqfunctionA(y_2)]$. Then by assumption~\ref{ax:tvar0}:
\[
  \var(t'[\essuniqfunctionB][e_\specialopT/\essuniqfunctionA(y_1), e_\specialopT/\essuniqfunctionA(y_2)]) = \emptyset.
\]
Therefore, $t'[\essuniqfunctionB]$ can contain no other variables than $\essuniqfunctionA(y_1)$ and $\essuniqfunctionA(y_2)$. That is:
\begin{equation}\label{eq:variable-subset}
  \{\essuniqfunctionB(x) \given x \in \var(t')\} \subseteq \{\essuniqfunctionA(y_1), \essuniqfunctionA(y_2)\}.
\end{equation}
Next, we show that both $\essuniqfunctionA(y_1)$ and $\essuniqfunctionA(y_2)$ need to appear in $t'[\essuniqfunctionB]$. Suppose that $\essuniqfunctionA(y_1)$ does not appear in $t'[\essuniqfunctionB]$. Then from \refeqn~\eqref{eq:variable-subset} we know that for all $x \in \var(t')$, $\essuniqfunctionB(x) = \essuniqfunctionA(y_2)$. Then:
\begin{align*}
 & \essuniqfunctionA(y_1) \\
\theoryeq{T}\; & \reason{e_\specialopT \text{ is the unit for } \specialopT} \\
 & \specialopT(\essuniqfunctionA(y_1), e_\specialopT) \\
\theoryeq{T}\; & \reason{\text{Substitution}} \\
 & \specialopT(\essuniqfunctionA(y_1), \essuniqfunctionA(y_2))[e_\specialopT/\essuniqfunctionA(y_2)] \\
\theoryeq{T}\; & \reason{\text{Equation~\eqref{t-knowledge1}}} \\
 & t'[\essuniqfunctionB][e_\specialopT/\essuniqfunctionA(y_2)] \\
\theoryeq{T}\; & \reason{\text{For all } x, \essuniqfunctionB(x) = \essuniqfunctionA(y_2)} \\
 & t'[e_\specialopT/x].
\end{align*}
So: $\essuniqfunctionA(y_1) \theoryeq{T} t'[e_\specialopT/x]$, but this contradicts assumption~\ref{ax:tvar0}, because $\var(t'[e_\specialopT/x]) = \emptyset$, since every free variable in $t'$ has been substituted with the constant $e_\specialopT$, and $\var(\essuniqfunctionA(y_1)) = \{\essuniqfunctionA(y_1)\} \neq \emptyset$.

So $\essuniqfunctionA(y_1)$ has to appear in $t'[\essuniqfunctionB]$. A similar line of reasoning yields the same conclusion for $\essuniqfunctionA(y_2)$. Therefore, there is an $x$ such that $\essuniqfunctionB(x) = \essuniqfunctionA(y_1)$ and there is an $x$ such that $\essuniqfunctionB(x) = \essuniqfunctionA(y_2)$.
In summary, if we define:
\begin{align*}
X_1 &= \{x \in \var(t') \given \essuniqfunctionB(x) = \essuniqfunctionA(y_1)\} \\
X_2 &= \{x \in \var(t') \given \essuniqfunctionB(x) = \essuniqfunctionA(y_2)\},
\end{align*}
then we know that neither $X_1$ nor $X_2$ is empty and that $X_1 \cup X_2 = \var(t')$.

We finally consider \refeqns~\eqref{s-knowledge1} and~\eqref{s-knowledge2} to reach a conclusion about the variables appearing in the terms $s'_x$.
Since for all $x \in X_1: \essuniqfunctionB(x) = \essuniqfunctionA(y_1)$, we have by \refeqn~\eqref{s-knowledge1} that $s'_x[e_\specialopS/x_0] \theoryeq{S} y_1$. Similarly, for all $x \in X_2$, $s'_x[e_\specialopS/x_0] \theoryeq{S} y_2$.
By assumption~\ref{ax:svar1}, we conclude that:
\begin{align}
\forall x \in X_1:\; & \{x_0, y_1\} \vdash s'_x\label{variable_subset1} \\
\forall x \in X_2:\; & \{x_0, y_2\} \vdash s'_x\label{variable_subset2}.
\end{align}
In addition, since for any $x \in X_1$, $y_1 \theoryeq{S} s'_x[e_\specialopS/x_0]$, we would have $y_1$ equal to a constant if $y_1$ would not appear in $s'_x$, contradicting assumption~\ref{ax:svar0}. Similarly for $y_2$ and $s'_x, x \in X_2$. And so:
\begin{align}
\forall x \in X_1:\; & y_1 \in \var(s'_x)\label{y_1_in_var}\\
\forall x \in X_2:\; & y_2 \in \var(s'_x)\label{y_2_in_var}.
\end{align}

To prove that $x_0 \in \var(s'_x)$ for all $x \in \var(t')$, we substitute $x_0 \mapsto e_\specialopT$ in \refeqn~\eqref{lambda_2}:
\[
\specialopS(\specialopT(y_1,y_2), e_\specialopT) \theoryeq{U} t'[s'_x[e_\specialopT/x_0]/x].
\]
By \refprop~\ref{propnary}, $\specialopS(\specialopT(y_1,y_2), e_\specialopT) \theoryeq{U} e_\specialopT$. Therefore we must have that also:
\begin{equation}\label{eq:t'=et}
t'[s'_x[e_\specialopT/x_0]/x] \theoryeq{U} e_\specialopT.
\end{equation}
The left hand side of this equation might not be separated, since we substitute a $\bb{T}$-term inside $\bb{S}$-terms. We analyse the terms $s'_x[e_\specialopT/x_0]$ further to separate them into $\bb{T}$-terms of $\bb{S}$-terms. There are two cases: either $x_0 \in \var(s'_x)$ or not.
\begin{itemize}
\item If $x_0 \in \var(s'_x)$, then by property~\ref{ax:newsnary-all}: $s'_x[e_\specialopT/x_0] = e_\specialopT$, which is a separated term.
\item If $x_0 \notin \var(s'_x)$, then $s'_x[e_\specialopT/x_0] = s'_x$, which is also a separated term.
\end{itemize}
We conclude that $t'[s'_x[e_\specialopT/x_0]/x]$ is separated. It is our goal to show that we must have $x_0$ appearing in each $s'_x$. To this end, define:
\[
X_3 = \{x \in \var(t') \;|\; x_0 \in s'_x \}.
\]
We will show that we must have $X_3 = \var(t')$. We define:
\[
t'' = t'[e_\specialopT / x \in X_3].
\]
Then:
\begin{align*}
& t''[s'_x / x] \\
\theoryeq{U} \;\; & \reason{\text{definition of } t''} \\
& t'[e_\specialopT / x \in X_3, s'_x / x \notin X_3] \\
\theoryeq{U} \;\; & \reason{\text{for } x \in X_3: s'_x[e_\specialopT/x_0] = e_\specialopT \right. \\
                  & \left. \;\;\;\;\,\text{and for } x \notin X_3: s'_x[e_\specialopT/x_0] = s'_x} \\
& t'[s'_x[e_\specialopT/x_0]/x] \\
\theoryeq{U} \;\; & \reason{\text{\refeqn~\eqref{eq:t'=et}}} \\
& e_\specialopT.
\end{align*}
So we can apply essential uniqueness to the equation $t''[s'_x / x] \theoryeq{U} e_\specialopT$. We conclude that there must be a variable substitution $g$ such that $t''[g] \theoryeq{T} e_\specialopT$. By property~\ref{ax:tvar0} we conclude that $\var(t''[g]) = \emptyset$ and hence also $\var(t'') = \emptyset$. Since $t'' = t'[e_\specialopT / x \in X_3]$, we conclude that $t'$ has no variables that are not in $X_3$. In other words: $X_3 = \var(t')$, and hence for all $x \in \var(t')$, $x_0 \in \var(s'_x)$.


This, together with \refeqns~\eqref{variable_subset1},~\eqref{variable_subset2},~\eqref{y_1_in_var} and~\eqref{y_2_in_var}, proves that for all $x \in \var(t')$:
\[
\var(s'_x) = \{y_1,x_0\} \qquad \text{or} \qquad \var(s'_x) = \{y_2,x_0\}.
\]
The fact that neither $X_1$ nor $X_2$ are empty means that this proves the lemma.
\end{proof}

With the first step done, we move on to the second step:

\begin{lem}\label{t_0-diamond}
  Let $\bb{S}$ and $\bb{T}$ be two algebraic theories satisfying~\ref{ax:svar0},~\ref{ax:svar1},~\ref{ax:newsnary-all} and~\ref{ax:tvar0},~\ref{ax:tvar1}, \ref{ax:tconst} respectively. Assume furthermore that there are terms: 
 \[
 2 \vdash_{\mathbb{S}} \specialopS \qquad\text{ and }\qquad 2 \vdash_{\mathbb{T}} \specialopT,
 \]
 satisfying~\ref{ax:sbinary} and~\ref{ax:tunit} respectively. Finally, let $\bb{U}$ be a composite theory of $\bb{T}$ after $\bb{S}$.  Then there is a $\bb{T}$-term $X \vdash t'$ and there is a family of $\bb{S}$-terms $s'_x, x \in X$ such that:
  \[
    \specialopS(\specialopT(y_1,y_2), x_0) \theoryeq{U} t'[s'_x/x],
  \]
  and for each $x \in \var(t')$:
  \[
  s'_x \theoryeq{S} \specialopS(y_1,x_0) \qquad \text{or} \qquad s'_x \theoryeq{S} \specialopS(y_2,x_0).
  \]
  Moreover, there is an $x$ such that $s'_x \theoryeq{S} \specialopS(y_1,x_0)$ and an $x$ such that $s'_x \theoryeq{S} \specialopS(y_2,x_0)$.

  Similarly, there is a $\bb{T}$-term $X' \vdash t''$ and there is a family of $\bb{S}$-terms $s''_{x'}, x' \in X'$ such that:
  \[
    \specialopS(x_0,\specialopT(y_1,y_2)) \theoryeq{U} t''[s''_{x'}/x'],
  \]
  and for each $x' \in \var(t'')$:
  \[
  s''_{x'} \theoryeq{S} \specialopS(x_0,y_1) \qquad \text{or} \qquad s''_{x'} \theoryeq{S} \specialopS(x_0,y_2).
  \]
  Moreover, there is an $x'$ such that $s''_{x'} \theoryeq{S} \specialopS(y_1,x_0)$ and an $x'$ such that $s''_{x'} \theoryeq{S} \specialopS(y_2,x_0)$.
\end{lem}

\begin{proof}
Again, we only explicitly prove the statements for $\specialopS(\specialopT(y_1,y_2), x_0)$. The second half of the claim follows similarly.

As $\bb{U}$ is a composite theory, we know from the separation axiom that there is a $\bb{T}$-term $X \vdash t'$ and a family of $\bb{S}$-terms $s'_x,\, x \in X$ such that:
  \[
    \specialopS(\specialopT(y_1,y_2), x_0) \theoryeq{U} t'[s'_x/x].
  \]
We substitute $y_1 \mapsto e_\specialopT$:
\begin{align}
  & \specialopS(\specialopT(e_\specialopT,y_2), x_0) \theoryeq{U} t'[s'_x[e_\specialopT/y_1]/x] \nonumber\\
  \Rightarrow \;\; & \reason{e_\specialopT \text{ is the unit of } \specialopT} \nonumber\\
  & \specialopS(y_2,x_0) \theoryeq{U} t'[s'_x[e_\specialopT/y_1]/x] \nonumber\\
  \Rightarrow \;\; & \reason{\text{showing that $\specialopS(y_2,x_0)$ is a separated term}} \nonumber\\
  & z[\specialopS(y_2,x_0)/z] \theoryeq{U} t'[s'_x[e_\specialopT/y_1]/x].\label{eq:z-t0}
\end{align}

To use essential uniqueness, we need two separated terms. However, $t'[s'_x[e_\specialopT/y_1]/x]$ is a $\bb{T}$-term built out of $\bb{S}$-terms with possibly a $\bb{T}$-constant in them. So we need to separate this term. We use \refprop~\ref{propnary} in combination with our knowledge from \reflem~\ref{t_0(s_i)} about the variables appearing in each $s'_x$ to do this. Define:
\begin{align*}
X_1 & = \{x \in \var(t') \given \var(s'_x) = \{y_1, x_0\}\} \\
X_2 & = \{x \in \var(t') \given \var(s'_x) = \{y_2, x_0\}\}.
\end{align*}
From \reflem~\ref{t_0(s_i)} we know that neither $X_1$ nor $X_2$ is empty and their union contains all variables in $t'$. For each $x \in X_1$, we use property~\ref{ax:newsnary-all} to apply \refprop~\ref{propnary}, which tells us $s'_x[e_\specialopT/y_1] \theoryeq{U} e_\specialopT$. For each $x \in X_2$, we know that since $y_1$ does not appear in $s'_x$, $s'_x[e_\specialopT/y_1] \theoryeq{U} s'_x$. Therefore:
\begin{equation}\label{eq:t'}
t'[s'_x[e_\specialopT/y_1]/x] \theoryeq{T} t'[e_\specialopT/x \in X_1, s'_x/x \in X_2].
\end{equation}
Next, define:
\[
t'' = t'[e_\specialopT/x \in X_1].
\]
Then $X_2 \vdash_\bb{T} t''$. We have:
\begin{align*}
   & t''[s'_x/x] \\
  \theoryeq{U} \;  & \reason{\text{Definition of } t''} \\
   & t'[e_\specialopT/x \in X_1, s'_x/x \in X_2]  \\
  \theoryeq{U} \; & \reason{\text{Equation~\eqref{eq:t'}} }\\
   \; &  t'[s'_x[e_\specialopT/y_1]/x] \\
  \theoryeq{U} \; & \reason{\text{Equation~\eqref{eq:z-t0}}} \\
   &  z[\specialopS(y_2,x_0)/z].
\end{align*}
Now we can use the essential uniqueness property, and conclude that there are functions $\essuniqfunctionA: \{z\} \rightarrow Z, \essuniqfunctionB: X_2 \rightarrow Z$ such that:
\begin{align}\label{t-knowledge3}
& z[\essuniqfunctionA(z)/z] \theoryeq{T} t''[\essuniqfunctionB(x)/x].
\end{align}
Furthermore, we have:
\begin{align}
   \essuniqfunctionB(x) = \essuniqfunctionA(z) \Leftrightarrow s'_x & \theoryeq{S} \specialopS(y_2,x_0).\label{s-knowledge3}
\end{align}

From \refeqn~\eqref{t-knowledge3} and assumption~\ref{ax:tvar1} we conclude that $\{\essuniqfunctionA(z)\} \vdash t''[\essuniqfunctionB(x)/x]$. And hence for all $x \in X_2$, $\essuniqfunctionB(x) = \essuniqfunctionA(z)$. So by \refeqn~\eqref{s-knowledge3}, for each $x \in X_2$: $s'_x \theoryeq{S} \specialopS(y_2,x_0)$, which gives us half of the desired conclusion about each $s'_x, x \in \var(t')$.

 A similar argument using the substitution $y_2 \mapsto e_\specialopT$ instead of $y_1 \mapsto e_\specialopT$ leads to the conclusion that for each $x \in X_1$, $s'_x \theoryeq{S} \specialopS(y_1, x_0)$.
\end{proof}

And finally, we prove the last step, in which we show that $\specialopS$ distributes over $\specialopT$. This yields a proof of \refthm~\ref{thm:times-over-plus}. We restate the theorem for convenience:

\begin{thm}[Times over Plus Theorem]
Let $\bb{S}$ and $\bb{T}$ be two algebraic theories, satisfying~\ref{ax:svar0},~\ref{ax:svar1},~\ref{ax:newsnary-all} and~\ref{ax:tvar0},~\ref{ax:tvar1},~\ref{ax:tconst} respectively. Assume furthermore that there are terms:
 \[
 2 \vdash_{\mathbb{S}} \specialopS \qquad\text{ and }\qquad 2 \vdash_{\mathbb{T}} \specialopT,
 \]
 satisfying~\ref{ax:sbinary} and~\ref{ax:tunit} respectively. Finally, let $\bb{U}$ be a composite theory of $\bb{T}$ after $\bb{S}$. Then $\specialopS$ distributes over $\specialopT$:
\begin{align}
    \specialopS(\specialopT(y_1,y_2), x_0) & \theoryeq{U} \specialopT(\specialopS(y_1,x_0),\specialopS(y_2,x_0)) \\
    \specialopS(x_0, \specialopT(y_1,y_2)) & \theoryeq{U} \specialopT(\specialopS(x_0,y_1),\specialopS(x_0,y_2)).
\end{align}
\end{thm}

\begin{proof}
Again, we only explicitly prove the first statement, as the proof of the second statement is similar, using the appropriate parts of \reflems~\ref{t_0(s_i)} and~\ref{t_0-diamond}.

From the separation axiom of composite theories and \reflem~\ref{t_0-diamond} we know that there is a $\bb{T}$-term $X \vdash t'$ and a family of $\bb{S}$-terms $s'_x,\, x \in X$ such that either $s'_x = \specialopS(y_1,x_0)$ or $s'_x = \specialopS(y_2,x_0)$ and:
 \[
   \specialopS(\specialopT(y_1,y_2), x_0) \theoryeq{U} t'[s'_x/x].
 \]
Define:
\begin{align*}
X_1 & = \{x \in \var(t') \given s'_x = \specialopS(y_1,x_0)\} \\
X_2 & = \{x \in \var(t') \given s'_x = \specialopS(y_2,x_0)\}.
\end{align*}
Then, using the substitution $x_0 \mapsto e_\specialopS$, we get:
\begin{align*}
& \specialopT(y_1,y_2) \\
\theoryeq{U} \; & \reason{e_\specialopS \text{ is the unit of } \specialopS} \\
& \specialopS(\specialopT(y_1,y_2), e_\specialopS) \\
\theoryeq{U} \; & \reason{\text{Substitution}} \\
& \specialopS(\specialopT(y_1,y_2), x_0)[e_\specialopS/x_0] \\
\theoryeq{U} \; & \reason{\text{\reflem~\ref{t_0-diamond}}} \\
& t'[s'_x[e_\specialopS/x_0]/x] \\
\theoryeq{U} \; & \reason{s'_x = \specialopS(y_1,x_0) \text{ or } s'_x = \specialopS(y_2,x_0)\right.  \\
& \left. \;\;\;\text{ and } e_\specialopS \text{ is the unit of } \specialopS} \\
& t' [y_1/x \in X_1, y_2/x \in X_2].
\end{align*}
So:
\begin{equation}\label{eq:theorem1}
  \specialopT(y_1,y_2) \theoryeq{U} t' [y_1/x \in X_1, y_2/x \in X_2].
\end{equation}
We also have:
\begin{align*}
& t'[s'_x/x]\\
\theoryeq{U} \; & \reason{\text{\reflem~\ref{t_0-diamond}}} \\
& t'[\specialopS(y_1,x_0)/x \in X_1, \specialopS(y_2,x_0)/x \in X_2] \\
\theoryeq{U}\; & \reason{\text{Substitution: } \text{term} = \text{variable}[\text{term}/\text{variable}]} \\
& t'[y_1[\specialopS(y_1,x_0)/y_1]/x \in X_1, y_2[\specialopS(y_2,x_0)/y_2]/x \in X_2].
\end{align*}
So we conclude:
\begin{align*}
& t'[s'_x/x] \\
\theoryeq{U} \; & \reason{\text{by the reasoning above}} \\
 &  t'[y_1[\specialopS(y_1,x_0)/y_1]/x \in X_1, y_2[\specialopS(y_2,x_0)/y_2]/x \in X_2] \\
\theoryeq{U} \; & \reason{\text{Equation~\eqref{eq:theorem1}}} \\
&  \specialopT(y_1,y_2)[\specialopS(y_1,x_0)/y_1, \specialopS(y_2,x_0)/y_2] \\
\theoryeq{U} \; &  \specialopT(\specialopS(y_1,x_0),\specialopS(y_2,x_0)),
\end{align*}
which proves the theorem.
\end{proof}

In suitable cases, \refthm~\ref{thm:times-over-plus} reduces the search space for distributive laws to a single possibility. From \refprop~\ref{thm:distlaw-from-compth} we know that the action of distributive laws is determined by the separated terms in the composite theory. And so:
\begin{cor}\label{cor:upperlimit}
   Let $S$ and $T$ be two monads presented by algebraic theories $\bb{S}$ and $\bb{T}$, having signatures with at least one constant and one binary operation. If for both theories the constant acts as a unit for the binary operation and the theories further satisfy~\ref{ax:svar0},~\ref{ax:svar1},~\ref{ax:newsnary-all}, and~\ref{ax:tvar0} and~\ref{ax:tvar1} respectively, then any distributive law $S \circ T \Rightarrow T \circ S$ distributes the binary from $\bb{S}$ over the binary from $\bb{T}$ as in \refeqn~\eqref{timesoverplus}.
\end{cor}
\begin{exa}[Unique Distributive Laws]\label{ex:uniquedistlaws}
 Let $S$ be any of the monads tree\footnote{respresented by the theory of a single binary operation and a constant, satisfying only left and right unitality.}, list, or multiset. Then the corresponding algebraic theory $\bb{S}$ contains only linear equations. Let $T$ be either the multiset or powerset monad. Since the multiset and powerset monads are commutative monads, we know that there is a distributive law $S \circ T \Rightarrow T \circ S$~\cite[Theorem 4.3.4]{ManesMulry2007}. Corollary~\ref{cor:upperlimit} states that this distributive law is unique. In particular, the distributive law for the multiset monad over itself mentioned in \refex~\ref{ex:multisetdistlaw} is unique.
\end{exa}

\subsection{Lacking the Abides Property: a No-Go Theorem}

With \refthm~\ref{thm:times-over-plus} narrowing down the possible distributive laws for two monads, it is easier to find cases in which no distributive law can exist at all. We identify two properties that clash with \refthm~\ref{thm:times-over-plus}, one for $\bb{T}$ and one for $\bb{S}$. In this section we show that not satisfying the abides equation, \refpro~\ref{ax:tspecialproperty}, in combination with \refthm~\ref{thm:times-over-plus} prevents the existence of a distributive law. In the next section we do the same for idempotence, \refpro~\ref{ax:sidem}. Both properties are sufficiently common to cover a broad class of monads.

\begin{thm}[No-Go Theorem: Lacking Abides]%
\label{thm:nogoTreeList}
   Let $\bb{S}$ and $\bb{T}$ be algebraic theories satisfying the conditions of \refthm~\ref{thm:times-over-plus}, and assume that the binary $\specialopT$ in $\bb{T}$ promised by~\ref{ax:tunit} additionally satisfies~\ref{ax:tspecialproperty}, then there does not exist a composite theory of $\bb{T}$ after $\bb{S}$.
\end{thm}
\begin{proof}
Suppose there exists a composite theory $\bb{U}$. Given \refthm~\ref{thm:times-over-plus}, we compute a separated term equal in $\bb{U}$ to $\specialopS(\specialopT(y_1,y_2), \specialopT(y_3,y_4))$:

\begin{align}
   & \specialopS(\specialopT(y_1,y_2), \specialopT(y_3,y_4)) \nonumber\\
\theoryeq{U}\; & \reason{\text{Substitution}} \nonumber\\
   & \specialopS(\specialopT(y_1,y_2),x_0)[\specialopT(y_3,y_4)/x_0] \nonumber\\
\theoryeq{U}\; & \reason{\text{Equation~\eqref{theoremeq1} from \refthm~\ref{thm:times-over-plus}}} \nonumber\\
   & \specialopT(\specialopS(y_1,x_0), \specialopS(y_2,x_0))[\specialopT(y_3,y_4)/x_0] \nonumber\\
\theoryeq{U}\; & \reason{\text{Substitution}} \nonumber\\
   & \specialopT(\specialopS(y_1,\specialopT(y_3,y_4)), \specialopS(y_2,\specialopT(y_3,y_4))) \nonumber\\
\theoryeq{U}\; & \reason{\text{Equation~\eqref{theoremeq2} from \refthm~\ref{thm:times-over-plus}}} \nonumber\\
   & \specialopT(\specialopT(\specialopS(y_1,y_3),\specialopS(y_1,y_4)), \specialopT(\specialopS(y_2,y_3),\specialopS(y_2,y_4))) .\label{eq:nogoTreeList1-appendix}
\end{align}
Notice that we made a choice, taking out the right $\specialopT$ term in
\[
  \specialopS(\specialopT(y_1,y_2), \specialopT(y_3,y_4)) \theoryeq{U} \specialopS(\specialopT(y_1,y_2),x_0)[\specialopT(y_3,y_4)/x_0],
\]
rather than the left:
\[
  \specialopS(\specialopT(y_1,y_2), \specialopT(y_3,y_4)) \theoryeq{U} \specialopS(x_0, \specialopT(y_3,y_4))[\specialopT(y_1,y_2)/x_0].
\]
The latter option yields:
\begin{align}
    & \specialopS(\specialopT(y_1,y_2), \specialopT(y_3,y_4)) \nonumber \\
\theoryeq{U}\; & \reason{\text{Substitution}} \nonumber\\
    & \specialopS(x_0, \specialopT(y_3,y_4))[\specialopT(y_1,y_2)/x_0] \nonumber\\
\theoryeq{U}\; & \reason{\text{Equation~\eqref{theoremeq2} from \refthm~\ref{thm:times-over-plus}}} \nonumber\\
    & \specialopT(\specialopS(x_0,y_3), \specialopS(x_0,y_4))[\specialopT(y_1,y_2)/x_0] \nonumber\\
\theoryeq{U}\; & \reason{\text{Substitution}} \nonumber \\
    & \specialopT(\specialopS(\specialopT(y_1,y_2),y_3), \specialopS(\specialopT(y_1,y_2), y_4)) \nonumber\\
\theoryeq{U}\; & \reason{\text{Equation~\eqref{theoremeq1} from \refthm~\ref{thm:times-over-plus}}} \nonumber\\
 & \specialopT(\specialopT(\specialopS(y_1,y_3),\specialopS(y_2,y_3)),\specialopT(\specialopS(y_1,y_4),\specialopS(y_2,y_4))) .\label{eq:nogoTreeList2-appendix}
\end{align}
Of course, both computations are equally valid, so the terms in \refeqns~\eqref{eq:nogoTreeList1-appendix} and~\eqref{eq:nogoTreeList2-appendix} must be equal:
\begin{align*}
& \specialopT(\specialopT(\specialopS(y_1,y_3),\specialopS(y_1,y_4)),\specialopT(\specialopS(y_2,y_3),\specialopS(y_2,y_4))) \\ \theoryeq{U}\; & \specialopT(\specialopT(\specialopS(y_1,y_3),\specialopS(y_2,y_3)),\specialopT(\specialopS(y_1,y_4),\specialopS(y_2,y_4))).
\end{align*}
Since these are two separated terms that are equal, we can apply the essential uniqueness property, stating that there exist functions:
\begin{align*}
  \essuniqfunctionA:&\: \{x_1,x_2,x_3,x_4\} \rightarrow Z\\
  \essuniqfunctionB:&\: \{x_5,x_6,x_7,x_8\} \rightarrow Z,
\end{align*}
such that:
\begin{itemize}
  \item Equality in $\bb{T}$:
         \begin{align*}
         & \specialopT(\specialopT(\essuniqfunctionA(x_1),\essuniqfunctionA(x_2)), \specialopT(\essuniqfunctionA(x_3),\essuniqfunctionA(x_4))) \\ \theoryeq{T}\; &  \specialopT(\specialopT(\essuniqfunctionB(x_5),\essuniqfunctionB(x_6)), \specialopT(\essuniqfunctionB(x_7),\essuniqfunctionB(x_8))).
         \end{align*}
  \item $\essuniqfunctionA(x_i) = \essuniqfunctionB(x_j)$ iff the $\bb{S}$-terms substituted for $x_i$ and $x_j$ in \refeqns~\eqref{eq:nogoTreeList1-appendix} and~\eqref{eq:nogoTreeList2-appendix} are equal in $\bb{S}$.
\end{itemize}
From the second part of essential uniqueness we get that:
\begin{align*}
\essuniqfunctionA(x_1) & = \essuniqfunctionB(x_5) & \essuniqfunctionA(x_3) & = \essuniqfunctionB(x_6)\\
\essuniqfunctionA(x_2) & = \essuniqfunctionB(x_7) & \essuniqfunctionA(x_4) & = \essuniqfunctionB(x_8).
\end{align*}
For readability of the next argument, set:
\begin{align*}
z_1 = \essuniqfunctionA(x_1) & = \essuniqfunctionB(x_5) & z_3 = \essuniqfunctionA(x_3) & = \essuniqfunctionB(x_6)\\
z_2 = \essuniqfunctionA(x_2) & = \essuniqfunctionB(x_7) & z_4 = \essuniqfunctionA(x_4) & = \essuniqfunctionB(x_8).
\end{align*}
Putting this in the equality in $\bb{T}$ we found under the first bullet point yields:
\begin{align*}
 \specialopT(\specialopT(z_1,z_2), \specialopT(z_3,z_4)) & \theoryeq{T} \specialopT(\specialopT(z_1,z_3), \specialopT(z_2,z_4)).
\end{align*}
And so by \refpro~\ref{ax:tspecialproperty}:
\[
\#\var(\specialopT(\specialopT(z_1,z_2), \specialopT(z_3,z_4))) \leq 3.
\]
So there must be $i,j$ such that $i \neq j$ and $z_i = z_j$. Suppose without loss of generality that $z_1 = z_2$. Then by essential uniqueness we must have that $\specialopS(y_1,y_3) \theoryeq{S} \specialopS(y_1,y_4)$. But then we can reason:
\begin{align*}
   & y_1 \\
   \theoryeq{S} \; & \reason{ e_\specialopS \text{ is a unit for } \specialopS } \\
   & \specialopS(y_1, e_\specialopS) \\
   \theoryeq{S} \; & \reason{\text{Substitution}} \\
   &  \specialopS(y_1, y_3)[e_\specialopS/y_3] \\
   \theoryeq{S} \; & \reason{\specialopS(y_1, y_3) = \specialopS(y_1, y_4)} \\
   & \specialopS(y_1, y_4)[e_\specialopS/y_3] \\
   \theoryeq{S} \; & \reason{\text{Substitution: no } y_3 \text{ in } \specialopS(y_1, y_4) } \\
   & \specialopS(y_1, y_4).
\end{align*}

We conclude that, by \refpro~\ref{ax:svar1}: $\{y_1\} \vdash \specialopS(y_1, y_4)$, and so we must have $y_4 = y_1$. Since these variables are assumed to be distinct, we have a contradiction.
The same argument holds for any other $i,j$ pair. Therefore, the existence of a composite theory leads to a contradiction. In other words, no such composite theory exists.
\end{proof}
\begin{cor}
 If monads~$S$ and~$T$ are presented by algebraic theories~$\bb{S}$ and~$\bb{T}$, satisfying the axioms of \refthm~\ref{thm:nogoTreeList}, then there does not exist a distributive law $S \circ T \Rightarrow T \circ S$.
\end{cor}
\begin{exa}[Resolving an Open Question]
  This finally settles the question of whether the list monad distributes over itself, posed repeatedly by Manes and Mulry~\cite{ManesMulry2007,ManesMulry2008}. The theory of monoids satisfies all the conditions required of both theories in \refthm~\ref{thm:nogoTreeList}, and hence there is no distributive law for the list monad over itself.

  Note that a distributive law for lists was claimed by King and Wadler~\cite{KingWadler1993}, although it was subsequently shown to be incorrect by Jones and Duponcheel~\cite{JonesDuponcheel1993}.
\end{exa}
\begin{rem}
  Although there is no distributive law for the list monad over itself, the functor $LL$ does still carry a monad structure. We are very grateful to Bartek Klin for pointing this out to us. The monad structure on $LL$ can be described as follows:
  \begin{itemize}
    \item There is a distributive law for the list monad over the non-empty list monad $L \circ L^+ \Rightarrow L^+ \circ L$~\cite{ManesMulry2007}.
    \item There is a distributive law for the resulting monad over the maybe monad $(L^+L) \circ {(-)}_\bot \Rightarrow {(-)}_\bot \circ (L^+L)$, derived from general principles~\cite{ManesMulry2007}.
    \item The resulting functor ${(-)}_\bot \circ (L^+ \circ L)$ is isomorphic to $L \circ L$, and carries a monad structure. Hence $L \circ L$ carries a monad structure, but not one that can be derived from a distributive law $L \circ L \Rightarrow L \circ L$.
  \end{itemize}
\end{rem}
\begin{nonexample}[Multiset Monad]\label{counter:multisetNolist}
The multiset monad is closely related to the list monad, with an algebraic theory having just one extra equation compared to the list monad: commutativity. Because of this equation, the theory does not have \refpro~\ref{ax:tspecialproperty}. As we have seen in \refex~\ref{ex:uniquedistlaws}, there is a unique distributive law for the multiset monad over itself.
\end{nonexample}

\subsection{Yet Another No-Go Theorem Caused by Idempotence}
In \refsec~\ref{section:Plotkin} we saw a no-go theorem that required an idempotent term on one side, and a unital one on the other: \refthm~\ref{thm:NoGo-Videm-Punit}. We will now see a second theorem of this type. Adding idempotence of the binary to the assumptions for $\bb{S}$ yields yet another no-go theorem, which partly overlaps with \refthm~\ref{thm:NoGo-Videm-Punit}, but neither theorem is a consequence of the other.

\begin{thm}[No-Go Theorem: Idempotence and Units]%
 \label{thm:nogoTreeIdem}
  Let $\bb{S}$ and $\bb{T}$ be algebraic theories satisfying~\ref{ax:svar0},~\ref{ax:svar1},~\ref{ax:newsnary-all} and~\ref{ax:tvar0},~\ref{ax:tvar1},~\ref{ax:tconst} respectively. Assume furthermore that there are terms
\[
 2 \vdash_{\mathbb{S}} \specialopS \qquad\text{ and }\qquad 2 \vdash_{\mathbb{T}} \specialopT,
 \]
 satisfying~\ref{ax:sbinary},~\ref{ax:sidem} and~\ref{ax:tunit} respectively. Then there exists no composite theory of $\bb{T}$ after $\bb{S}$.
\end{thm}
\begin{proof}
  Suppose such a composite theory $\bb{U}$ exists. Then we have:
  \begin{align*}
     & \specialopT(y_1,y_2) \\
  \theoryeq{U}\; & \reason{\text{\ref{ax:sidem}: } \specialopS \text{ is idempotent} } \\
     & \specialopS(\specialopT(y_1,y_2), \specialopT(y_1,y_2)) \\
  \theoryeq{U}\; & \reason{\text{Substitution}} \\
     & \specialopS(\specialopT(y_1,y_2),x_0)[\specialopT(y_1,y_2)/x_0] \\
  \theoryeq{U}\; & \reason{\text{Equation~\eqref{theoremeq1} from \refthm~\ref{thm:times-over-plus}}} \\
     & \specialopT(\specialopS(y_1,x_0), \specialopS(y_2,x_0))[\specialopT(y_1,y_2)/x_0] \\
  \theoryeq{U}\; & \reason{\text{Substitution}} \\
     & \specialopT(\specialopS(y_1,\specialopT(y_1,y_2)), \specialopS(y_2,\specialopT(y_1,y_2))) \\
  \theoryeq{U}\; & \reason{\text{Equation~\eqref{theoremeq2} from \refthm~\ref{thm:times-over-plus}}} \\
     & \specialopT(\specialopT(\specialopS(y_1,y_1),\specialopS(y_1,y_2)), \specialopT(\specialopS(y_2,y_1),\specialopS(y_2,y_2))) \\
  \theoryeq{U}\; & \reason{\text{\ref{ax:sidem}: } \specialopS \text{ is idempotent} } \\
     & \specialopT(\specialopT(y_1,\specialopS(y_1,y_2)), \specialopT(\specialopS(y_2,y_1),y_2)).
  \end{align*}

  From the essential uniqueness property, we may conclude that there are functions $\essuniqfunctionA: \{y_1,y_2\} \mapsto Z, \essuniqfunctionB: \{y'_1,y'_2,y'_3,y'_4\} \mapsto Z$ such that:
  \begin{align*}
    \specialopT(\essuniqfunctionA(y_1), \essuniqfunctionA(y_2)) & \theoryeq{T} \specialopT(\specialopT(\essuniqfunctionB(y'_1),\essuniqfunctionB(y'_2)), \specialopT(\essuniqfunctionB(y'_3), \essuniqfunctionB(y'_4))),
  \end{align*}
  and $\essuniqfunctionA(y_i) = \essuniqfunctionB(y'_j)$ if and only if the $\bb{S}$-terms substituted for $y_i$ and $y'_j$ in $\specialopT(y_1,y_2)$ and $\specialopT(\specialopT(y_1,\specialopS(y_1,y_2)), \specialopT(\specialopS(y_2,y_1),y_2))$ are equal. From this we immediately get:
  \begin{align*}
   \essuniqfunctionA(y_1) & = \essuniqfunctionB(y'_1) \\
   \essuniqfunctionA(y_2) & = \essuniqfunctionB(y'_4).
  \end{align*}
   We know from essential uniqueness that $\essuniqfunctionA(y_1) \neq \essuniqfunctionA(y_2)$.
   We show that $\{\essuniqfunctionB(y'_2), \essuniqfunctionB(y'_3)\} \subseteq \{\essuniqfunctionA(y_1),\essuniqfunctionA(y_2)\}$. Since we have:
   \[
   \specialopT(\essuniqfunctionA(y_1), \essuniqfunctionA(y_2)) \theoryeq{T} \specialopT(\specialopT(\essuniqfunctionB(y'_1), \essuniqfunctionB(y'_2)), \specialopT(\essuniqfunctionB(y'_3),\essuniqfunctionB(y'_4))),
   \]
   we also have:
   \begin{align*}
   & \specialopT(\essuniqfunctionA(y_1), \essuniqfunctionA(y_2))[e_\specialopT/\essuniqfunctionA(y_1),e_\specialopT/\essuniqfunctionA(y_2)] \\
   & \theoryeq{T} \specialopT(\specialopT(\essuniqfunctionB(y'_1), \essuniqfunctionB(y'_2)), \specialopT(\essuniqfunctionB(y'_3),\essuniqfunctionB(y'_4)))[e_\specialopT/\essuniqfunctionA(y_1),e_\specialopT/\essuniqfunctionA(y_2)] \\
   \Rightarrow \; & \reason{\text{Substitution, and } \essuniqfunctionA(y_1) = \essuniqfunctionB(y'_1), \essuniqfunctionA(y_2) = \essuniqfunctionB(y'_4)} \\
   & \specialopT(e_\specialopT, e_\specialopT) \theoryeq{T} \specialopT(\specialopT(e_\specialopT, \essuniqfunctionB(y'_2)), \specialopT(\essuniqfunctionB(y'_3),e_\specialopT))[e_\specialopT/\essuniqfunctionA(y_1),e_\specialopT/\essuniqfunctionA(y_2)] \\
   \Rightarrow \; & \reason{\text{\ref{ax:tunit}: }e_\specialopT \text{ is the unit of } \specialopT} \\
   & e_\specialopT \theoryeq{T} \specialopT(\essuniqfunctionB(y'_2),\essuniqfunctionB(y'_3))[e_\specialopT/\essuniqfunctionA(y_1),e_\specialopT/\essuniqfunctionA(y_2)].
   \end{align*}
   So by~\ref{ax:tvar0}, $\var(\specialopT(\essuniqfunctionB(y'_2), \essuniqfunctionB(y'_3))[e_\specialopT/\essuniqfunctionA(y_1),e_\specialopT/\essuniqfunctionA(y_2)]) = \emptyset$. So we must have:
   \[
     \{\essuniqfunctionB(y'_2), \essuniqfunctionB(y'_3)\} \subseteq \{\essuniqfunctionA(y_1),\essuniqfunctionA(y_2)\}.
   \]
   But then, by the second part of the essential uniqueness property:
   \[
     \specialopS(y_1,y_2) \theoryeq{S} y_1 \quad \text{or} \quad \specialopS(y_1,y_2) \theoryeq{S} y_2.
   \]
   Both contradict~\ref{ax:svar1}. Therefore, the composite theory $\bb{U}$ cannot exist.
\end{proof}
\begin{cor}
  If monads~$S$ and~$T$ are presented by algebraic theories~$\bb{S}$ and $\bb{T}$, satisfying the axioms of \refthm~\ref{thm:nogoTreeIdem}, then there does not exist a distributive law $S \circ T \Rightarrow T \circ S$.
\end{cor}

\begin{exa}[Powerset Monad Again]
The theory of join semilattices satisfies all the axioms required of both theories in Theorem~\ref{thm:nogoTreeIdem}. Therefore, there is no distributive law for the powerset monad over itself. This was already shown by Klin and Salamanca~\cite{KlinSalamanca2018} using similar methods as in \refsec~\ref{section:Plotkin}. Theorem~\ref{thm:nogoTreeIdem} gives a second, independent proof of this fact.
\end{exa}

\begin{rem}
Theorems~\ref{thm:nogoTreeIdem} and~\ref{thm:NoGo-Videm-Punit} both require an idempotent term in theory $\bb{S}$, and a unital term in theory $\bb{T}$, to preclude a composite theory of $\bb{T}$ after $\bb{S}$. However, these theorems are neither equivalent, nor does one imply the other. The most obvious difference is that Theorem~\ref{thm:NoGo-Videm-Punit} is stated slightly more general, where instead of a unital term a generalisation of both unitality and idempotence is required. When restricting this requirement to just a unital term, there are still differences between the two theorems: Theorem~\ref{thm:NoGo-Videm-Punit} applies only if the unital term in $\bb{T}$ is commutative, whereas \refthm~\ref{thm:nogoTreeIdem} allows this term to be non-commutative. Conversely, \refthm~\ref{thm:nogoTreeIdem} requires the idempotent term in $\bb{S}$ to be unital, where \refthm~\ref{thm:NoGo-Videm-Punit} does not have this restriction.
\end{rem}

\begin{nonexample}[Multiset Monad: The Sweet Spot]
 We come back to the multiset monad. In \refnon~\ref{counter:multisetNolist} we saw that the algebraic theory presenting the multiset monad had one extra equation compared to the theory for the list monad: commutativity. Because of this equation, \refpro~\ref{ax:tspecialproperty} did not hold, and therefore \refthm~\ref{thm:nogoTreeList} did not apply.

 There is a similar relation between the multiset monad and the powerset monad. Compared to the powerset monad, the theory presenting the multiset monad \emph{lacks} just one equation: idempotence, which is exactly what \refpro~\ref{ax:sidem} requires. The lack of this equation in the theory for the multiset monad therefore means that \refthm~\ref{thm:nogoTreeIdem} does not apply to multiset either. So multiset holds a sort of `sweet spot' in between the two no-go theorems, where a distributive law~$M \circ M \Rightarrow M \circ M$ still can and does exist.
\end{nonexample}

\section{A Conjecture of Beck Yielding More No-Go Theorems}%
\label{sec:plus-over-times}

\subsection{Motivation}
The classical example of a distributive law constructs the ring monad from the list monad~$L$, and the Abelian group monad~$A$, via a distributive law~$L \circ A \Rightarrow A \circ L$. We encountered this distributive law in~\refex~\ref{ex:ring-monad}. It dates back to Beck's original paper~\cite{Beck1969}, and exploits the arithmetic distribution of multiplication over addition. A natural question, briefly considered by Beck, is
whether there can be a distributive law~$A \circ L \Rightarrow L \circ A$ with these monads reversed.
Beck gives the intuition that such a distributive law ``would have the air of a universal solution to the problem of factoring polynomials into linear factors''. As such, he suggests $L \circ A$ ``has little chance of being a triple''. Addressing the question of whether ``plus distributes over times'' is the motivating example for the work in this section.
Unsurprisingly, Beck's intuition is correct, and no such distributive law exists. It turns out that \refprop~\ref{propnary} is key to the proof.

\begin{counter}\label{counter:noplusovertimes}
There is no distributive law~$A \circ L \Rightarrow L \circ A$ for the Abelian group monad over the list monad.
\end{counter}
\begin{proof}
The theory of Abelian groups $\bb{A}$, yielding the Abelian group monad~$A$, has presentation:
\begin{itemize}
\item Signature: $\Sigma^{\bb{A}} = \{0^{(0)}, -{(.)}^{(1)}, +^{(2)}\}$.
\item Equations: $E^\bb{A}$ contains the equations stating that $0$ is the unit of $+$, $+$ is associative and commutative, and~$-x$ is the additive inverse of $x$: $x + (-x) = 0$.
\end{itemize}
The theory of monoids $\bb{M}$, yielding the list monad~$L$, has presentation:
\begin{itemize}
\item Signature: $\Sigma^{\bb{M}}= \{1^{(0)}, *^{(2)}\}$.
\item Equations: $E^\bb{M}$ contains the equations stating that $1$ is unit of $*$, and~$*$ is associative.
\end{itemize}
The term $x + y$ in $\bb{A}$ satisfies the conditions for $\bb{S}$ in \refprop~\ref{propnary}, and $\bb{M}$ satisfies the conditions for $\bb{T}$.
We conclude that in any composite theory $\bb{U}$ of $\bb{M}$ after $\bb{A}$, the following equation must hold:
\[
x + 1 \theoryeq{U} 1.
\]
We prove that this implies that $x \theoryeq{U} 0$:
\begin{align*}
& x \\
\theoryeq{U}\;\; & \reason{\text{unit}}\\
&  x + 0 \\
\theoryeq{U}\;\; & \reason{\text{inverse}} \\
&  x + (1 + (-1)) \\
\theoryeq{U}\;\; & \reason{\text{associativity}} \\
&  (x + 1) + (-1) \\
\theoryeq{U}\;\; & \reason{x + 1 \theoryeq{U} 1} \\
&  1 + (-1) \\
\theoryeq{U}\;\; & \reason{\text{inverse}} \\
&  0.
\end{align*}
Hence for any two variables: $x \theoryeq{U} 0 \theoryeq{U} y$, which means any composite theory $U$ is inconsistent. Since the component theories $\mathbb{M}$ and $\mathbb{A}$ are consistent, \refprop~\ref{prop:consistency} tells us that there is no such composite theory.
\end{proof}
As for the previous proof ideas, this specific negative result can be abstracted to yield general theorems. We shall address this in the following section.

\subsection{Generalization}
The well-known \emph{positive} result of Manes and Mulry~\cite{ManesMulry2007} about distributive laws requires monads presented by theories with only linear equations. That is, with variables appearing exactly once on each side of the equation. These monads distribute over commutative monads: $linear \circ commutative \Rightarrow commutative \circ linear$. The careful management of variables has also been important to the proofs of the \emph{negative} theorems in~\refsec~\ref{section:Plotkin} and \refsec~\ref{section:beyondPlotkin}. Generally these theorems do not apply to theories with inverses such as groups, as~$x - x = 0$ means we can make variables appear and disappear almost at will, leading to behaviour too wild to analyze.
The concrete \refcounter~\ref{counter:noplusovertimes} \emph{did} apply to an algebraic theory involving inverses. In fact, it made essential use of the ability to make variables appear and disappear, turning this behaviour from a liability into an asset.

We now abstract from \refcounter~\ref{counter:noplusovertimes}, deducing two general theorems that require `variables to go missing' in different ways. The first theorem requires an equation of form $s' = e_\bb{S}$, where $e_\bb{S}$ is a constant and $s'$ has at least one variable. The second theorem asks for an equation of form $s' = s$, where the term $s$ has a variable that does not appear in $s'$, and neither $s$ nor $s'$ are constants.

We begin with the case of a theory with an equation~$s' = e_\bb{S}$. Here, the example to keep in mind is an inverse axiom such as~$x - x = 0$, typical of group-like structures, such as groups and rings. This is the most direct abstraction of the motivating counterexample.

Notice that the proof of \refcounter~\ref{counter:noplusovertimes} used associativity. We could have avoided this by using the following alternative, but less intuitive, reasoning that applies \refprop~\ref{propnary} twice:
\begin{align*}
& x \\
\theoryeq{U}\;\; & \reason{\text{unit}}\\
&  x + 0 \\
\theoryeq{U}\;\; & \reason{\text{inverse}} \\
&  x + (1 + (-1)) \\
\theoryeq{U}\;\; & \reason{\text{substitution}} \\
&  x + (1 + y)[-1/y] \\
\theoryeq{U}\;\; & \reason{\text{\refprop~\ref{propnary}: } 1 + y \theoryeq{U} 1} \\
&  x + 1 [-1/y] \\
\theoryeq{U}\;\; & \reason{\text{no variable } y} \\
&  x + 1 \\
\theoryeq{U}\;\; & \reason{\text{\refprop~\ref{propnary}: } x + 1 \theoryeq{U} 1} \\
&  1.
\end{align*}
This argumentation uses fewer assumptions, so this is the proof we generalize below.

\begin{thm}[No-Go Theorem: Inverse Trouble]\label{theoremwithconst}
  Let $\bb{S}$ be an algebraic theory such that:
  \begin{enumerate}[label=(S\arabic*)]
  \item $\bb{S}$ has a constant $e_\bb{S}$.
  \item\label{ax:sunitforthmwithconst} $\bb{S}$ has a term $s$ of arity $\geq 2$ such that $e_\bb{S}$ is a unit of $s$, that is, for any variable $x \in \var(s)$:
  \[
  s[e_\bb{S}/y\neq x] \theoryeq{S} x.
  \]
  \item\label{ax:complicatedcond} $\bb{S}$ satisfies an equation of form
  \[
  X \vdash s' \theoryeq{S} e_\bb{S},
  \]
  with $\var(s') \neq \emptyset$, and $s'$ can be written as $s''[s_i/x_i]$, such that $\var(s'') \cap \var(s') \neq \emptyset$, and there is a substitution $f: \var(s'') \rightarrow \bb{S}$ such that for any $x \in \var(s'')$:
   \[
   \Gamma \vdash s''[f(y)/y \neq x] \theoryeq{S} x.
   \]
  \end{enumerate}
  Let $\bb{T}$ be an algebraic theory such that:
  \begin{enumerate}[label=(T\arabic*)]
  \item $\bb{T}$ has a constant $e_\bb{T}$.
  \item\label{ax:inverse-weakt0} For all terms $t'$ and any substitution $f: X \rightarrow Y$:
  \[
  Y \vdash t'[f] \theoryeq{T} e_\bb{T} \quad\Rightarrow\quad X \vdash t' \theoryeq{T} e_\bb{T}.
  \]
  \end{enumerate}
Then there does not exist a composite theory of $\bb{T}$ after $\bb{S}$.
\end{thm}
\begin{rem}[Interpretation of Axioms]~\ref{ax:complicatedcond} is designed to generalize the proof technique illustrated above. We used the equation $x + (-x) = 0$, and the fact that $x + (-x)$ could be written as $x + y[(-x) / y]$. We then used \refprop~\ref{propnary} on $x + y$. In order for this argument to work in general, we hence require:
  \begin{itemize}
  \item A term $s'$ which is equal to a constant, $x + (-x)$ in the example above.
  \item A term $s''$ such that $s'$ is equal to $s''$ under a certain substitution, $x + y$ in the example.
  \item Since we want to apply \refprop~\ref{propnary}, $s''$ needs to satisfy the conditions for this proposition.
  \item And lastly we require that $s'$ and $s''$ share at least one variable, which is a technicality needed to make the proof go through.
  \end{itemize}

  \noindent
  We need axiom~\ref{ax:inverse-weakt0} to be able to apply \refprop{propnary}. As a reminder, it reads: ``If a variable substitution of term $t$ is provably equal to a constant, then $t$ is already provably equal to that constant.''
\end{rem}
\begin{proof}
Let $\bb{U}$ be any (candidate) composite theory of $\bb{T}$ after $\bb{S}$.
Consider the equation $s' = e_\bb{S}$.
We know that $\bb{S}$ has a term $s$ such that $e_\bb{S}$ is a unit of $s$. Choose variable $x$ such that $x \notin \var(s')$. Then we use both the unit equation for $s$ and the fact that $s' = e_\bb{S}$:
\begin{align*}
 & x \theoryeq{U} s[e_\bb{S}/y \neq x] \\
\Rightarrow \;\; & \reason{ e_\bb{S} \theoryeq{S} s' } \\
 & x \theoryeq{U} s[s'/y \neq x].
\end{align*}
Next, we substitute $e_\bb{T}$ into all variables in $s'$. Since we chose $x$ such that $x \notin \var(s')$, this substitution has no effect on the left hand side of our equation.
\begin{align*}
 & x \theoryeq{U} s[s'/y \neq x] \\
\Rightarrow \;\; & \reason{ \text{Substitution} } \\
 & x[e_\bb{T}/z\in\var(s')] \theoryeq{U} s[s'[e_\bb{T}/z\in\var(s')]/y \neq x] \\
\Rightarrow \;\; & \reason{ x \notin \var(s')} \\
 & x \theoryeq{U} s[s'[e_\bb{T}/z\in\var(s')]/y \neq x].
\end{align*}
We will now work on the term $s[s'[e_\bb{T}/z\in\var(s')]/y \neq x]$. Recall that $s'$ can be written as $s''[s_i/x_i]$ and $s''$ satisfies the conditions for \refprop~\ref{propnary}. Also, since $\var(s'')\cap\var(s')\neq\emptyset$, we know that the substitution $s'[e_\bb{T}/z\in\var(s')] \theoryeq{U} s''[s_i/x_i][e_\bb{T}/z\in\var(s')]$ yields a term where at least one of the variables of $s''$ gets substituted with $e_\bb{T}$. Hence by \refprop~\ref{propnary}, this resulting term is equal to $e_\bb{T}$. Therefore:
\begin{align*}
& x \theoryeq{U} s[s'[e_\bb{T}/z\in\var(s')]/y \neq x] \\
\Rightarrow \;\; & \reason{\text{writing } s' \text{ as } s''[s_i/x_i] } \\
 & x \theoryeq{U} s[s''[s_i/x_i][e_\bb{T}/z\in\var(s')]/y \neq x] \\
\Rightarrow \;\; & \reason{\text{\refprop~\ref{propnary}: } s''[s_i/x_i][e_\bb{T}/z\in\var(s')] \theoryeq{U} e_\bb{T}} \\
& x \theoryeq{U} s[e_\bb{T}/y \neq x] \\
\Rightarrow \;\; & \reason{ \text{\refprop~\ref{propnary} again}} \\
 & x \theoryeq{U} e_\bb{T}.
\end{align*}
Notice that in the last step, we applied \refprop~\ref{propnary} to $s$ instead of $s''$. We are allowed to do this because $s$ is unital by assumption~\ref{ax:sunitforthmwithconst}, and hence also satisfies \refprop~\ref{propnary}.

From the equation $x \theoryeq{U} e_\bb{T}$, we get by simple variable substitution that $y \theoryeq{U} e_\bb{T}$ for any variable $y$, and so specifically: $x \theoryeq{U} e_\bb{T} \theoryeq{U} y$. We conclude that $\bb{U}$ is inconsistent. As the original theories are assumed to be consistent, there is no such composite theory.
\end{proof}
\begin{exa}
In case of Abelian groups and monoids, the equation $s' = e_\bb{S}$ required from Abelian groups is the inverse equation $x + (-x) = 0$. This equation can be written as $(x + y)[(-x)/y]$. Since the terms $x + y$ and $x + (-x)$ share the variable $x$, condition~\ref{ax:complicatedcond} is satisfied.
\end{exa}

\begin{exa}
There are countless monads satisfying the criteria for $\bb{T}$. A few natural examples are the list, multiset, powerset, and the exception monads.

We have already seen that Abelian groups satisfy the criteria for $\bb{S}$. In addition, any theory with a multiplicative zero $x * 0 = 0$ can be a good candidate, if the binary operation $*$ is either idempotent or unital. Rings are an obvious example, but \refprop~\ref{propnary} gives us many more. The multiset monad satisfies all the criteria for $\bb{S}$ and $\bb{T}$ in \refprop~\ref{propnary}. We know that there is a distributive law for the multiset monad over itself, see \refex~\ref{ex:multisetdistlaw}, so the unit of one of the binary operations in the composite theory corresponding to the monad $M \circ M$ must act as a multiplicative zero for the other.

We can hence make the following table of example compositions \refthm~\ref{theoremwithconst} proves impossible via a distributive law:


\begin{table}[h]
  \centering
  \caption{Overview of some distributive laws of type: }\label{tab:theoremwithconstexample}
     row $\circ$ column $\Rightarrow$ column $\circ$ row, which are excluded by \refthm~\ref{theoremwithconst}.\\
     \aboverulesep=0ex
     \belowrulesep=0ex
     \renewcommand{\arraystretch}{1.1}
    \begin{tabular}{l|cccc} 
    \toprule
      & List \;\; & Multiset \;\; & Powerset \;\; & Exception \\ \midrule
    Abelian groups \;\;     & \no & \no & \no & \no \\
    Rings \;\;     & \no & \no & \no & \no \\
    (Multiset)$^2$ \;\; & \no & \no & \no & \no \\\bottomrule
    \end{tabular}
\end{table}
More examples illustrating the scope of \refthm~\ref{theoremwithconst} will be given in \refsec~\ref{ch:boom}.
\end{exa}

The second theorem we find as a generalization of \refcounter~\ref{counter:noplusovertimes} focusses on the equation $s = s'$, where $s$ has a variable that does not appear in $s'$. Here
the motivating example axioms are absorption laws~$a \wedge (a \vee b) = a$, seen in lattices and similar structures.
This is a slightly less direct abstraction of~\refcounter~\ref{counter:noplusovertimes}, exploiting the observation
that the key requirement is controlled introduction and elimination of variables.
\begin{thm}[No-Go Theorem: Absorption Trouble]\label{theoremwithvar}
  Let $\bb{S}$ be an algebraic theory such that:
  \begin{enumerate}[label=(S\arabic*)]
  \item\label{ax:absorptionax1} $\bb{S}$ satisfies an equation of form $X \vdash s \theoryeq{S} s'$, where $\var(s)\setminus \var(s') \neq \emptyset$ (that is, $s$ has a variable that does not appear in $s'$).
   \item\label{ax:absorptionax2} There is a substitution $f: \var(s) \rightarrow \bb{S}$ such that for any $x \in \var(s)$:
   \[
   \Gamma \vdash s[f(y)/y \neq x] \theoryeq{S} x.
   \]
  \item\label{ax:absorptionax3} There is a substitution $f': \var(s') \rightarrow \bb{S}$ such that for any $x \in \var(s')$:
   \[
   \Gamma \vdash s'[f'(y)/y \neq x] \theoryeq{S} x.
   \]
  \end{enumerate}
  Let $\bb{T}$ be an algebraic theory such that:
  \begin{enumerate}[label=(T\arabic*)]
  \item $\bb{T}$ has a constant $e_\bb{T}$.
  \item\label{ax:absorptionax4} For all terms $t'$ and any substitution $f: X \rightarrow Y$:
  \[
  Y \vdash t'[f] \theoryeq{T} e_\bb{T} \quad\Rightarrow\quad X \vdash t' \theoryeq{T} e_\bb{T}.
  \]
  \end{enumerate}

\noindent
Then there does not exist a composite theory of $\bb{T}$ after $\bb{S}$.
\end{thm}

\begin{rem}[Interpretation of Axioms]
Apart from~\ref{ax:absorptionax1}, all the required properties are familiar from \refsec~\ref{section:beyondPlotkin}.~\ref{ax:absorptionax2} and~\ref{ax:absorptionax3} are generalizations of idempotence/unitality equations, while~\ref{ax:absorptionax4} states that ``If a variable substitution of term $t$ is provably equal to a constant, then $t$ is already provably equal to that constant.''
\end{rem}

\begin{proof}
Let $\bb{U}$ be any (candidate) composite theory of $\bb{T}$ after $\bb{S}$. Consider the equation $s = s'$ from the assumptions. Let $f'$ be a substitution of the variables in $s'$ such that $s'[f'] \theoryeq{S} x$, where $x \in \var(s')$. Then:
\begin{align*}
& s \theoryeq{S} s' \\
\Rightarrow \;\; & \reason{\text{Substitution}}\\
 & s[f'] \theoryeq{S} s'[f'] \\
\Rightarrow \;\; & \reason{ s'[f'] \theoryeq{S} x \text{ by construction of } f' } \\
 & s[f'] \theoryeq{S} x.
\end{align*}

Now let $g$ be the constant substitution $g: \var(s)\setminus\var(s') \rightarrow \{e_\bb{T}\}$, mapping all the variables that appear in $s$ but not in $s'$ to the constant $e_\bb{T}$. Since $\var(s)\setminus\var(s') \neq \emptyset$, at least one instance of $e_\bb{T}$ will be present in $s[f'][g]$. Then:
\begin{align*}
& s[f'] \theoryeq{U} x \\
\Rightarrow \;\; & \reason{\text{Substitution}}\\
 & s[f'][g] \theoryeq{U} x[g] \\
\Rightarrow \;\; & \reason{ x \in \var(s') \text{ and therefore untouched by } g } \\
& s[f'][g] \theoryeq{U} x \\
\Rightarrow \;\; & \reason{ \text{\refprop~\ref{propnary}, since } s[f'][g] \text{ contains at least one } e_\bb{T} } \\
& e_\bb{T} \theoryeq{U} x.
\end{align*}
Using substitution, we can hence show that $x \theoryeq{U} e_\bb{T} \theoryeq{U} y$, from which it follows that $\bb{U}$ is inconsistent. As the original theories are assumed to be consistent, there is no such composite theory.
\end{proof}

\refthm~\ref{theoremwithvar} solves a question Julian Salamanca posed in 2018~\cite{Salamanca2018}, asking whether there is a distributive law $bL \circ P \Rightarrow P \circ bL$, distributing the bounded lattice monad $bL$ over the powerset monad $P$. The answer is no.

\begin{exa}\label{ex:latticeandpower}
There is no distributive law for the (bounded) lattice monad over the powerset monad.
The presentation for the bounded lattice monad is given by:
\begin{itemize}
\item Signature: $\{\top^{(0)}, \bot^{(0)}, \vee^{(2)}, \wedge^{(2)}\}$.
\item Equations: $\top$ is the unit of $\wedge$, $\bot$ is the unit of $\vee$, associativity of $\vee$ and $\wedge$, commutativity of $\vee$ and $\wedge$, idempotence of $\vee$ and $\wedge$, absorption both ways: $x \vee (x \wedge y) = x$ and  $x \wedge (x \vee y) = x$.
\end{itemize}
We see that the bounded lattice monad satisfies the criteria for $\bb{S}$ in \refthm~\ref{theoremwithvar}:
\begin{itemize}
  \item The equation $s = s'$ is $x \vee (x \wedge y) = x$.
  \item The substitution $f$ such that $s[f] = x$ uses the unit of $\wedge$ and idempotence of $\vee$:
  \[
  x \vee (x \wedge y) [\top / y] = x \vee (x \wedge \top) = x \vee x = x.
  \]
  \item The substitution $f'$ such that $s'[f'] = x$ is the identity, since the term $s'$ is just the variable $x$.
\end{itemize}
The presentation of the powerset monad is given by:
\begin{itemize}
\item Signature: $\{0^{(0)}, +^{(2)}\}$.
\item Equations: $0$ is unit of $+$, $+$ is associative, commutative, and idempotent.
\end{itemize}
The powerset monad satisfies the criteria for $\bb{T}$ in \refthm~\ref{theoremwithvar}, and so we conclude that there is no distributive law $bL \circ P \Rightarrow P \circ bL$.
\end{exa}

\section{The Boom Hierarchy: a case study for distributive laws}\label{ch:boom}
We now pursue a detailed investigation of when distributive laws exist between some natural families of monads. To do so, we shall combine the techniques developed in earlier sections with results from the wider literature. Our objectives are to illustrate that the absence of distributive laws is not at all unusual, to document many useful examples, and to develop some intuitions via concrete applications.

We shall begin with the so-called Boom hierarchy, a small family of monads considered in the functional programming literature~\cite{Meertens1986}. Later, to increase our available data points, we will expand the original Boom hierarchy to include more exotic, if slightly less natural, data structures. Similar expansions of the Boom hierarchy have been studied by Uustalu~\cite{UUSTALU2016}.

The Boom hierarchy is a simple family of four monads, providing a pleasing connection between commonly used data structures and natural algebraic axioms.
The hierarchy consists of the tree, list, multiset, and powerset monads. Each of these monads has the same signature, consisting of a constant and a binary operation.
If the only axiom is the unitality axiom, the resulting monad is the binary tree monad. Adding associativity yields the list monad. Further adding commutativity yields the multiset monad, and finally adding idempotence results in the finite powerset monad, as shown in~\reftab~\ref{Boom1}.


\begin{rem}
  The Boom hierarchy is named after the Dutch Computer Scientist Hendrik Boom. The fact that `Boom' also means `tree' in Dutch is not entirely coincidental. Allegedly, the name was coined by Peter Grogono in a meeting with Stephen Spackman and Hendrik Boom. Spackman was a MSc student co-supervised by Grogono and Boom, working on this hierarchy of data structures. When Grogono suggested the name for the hierarchy, Boom's response was ``What, because it is about trees?''. The name has stuck ever since~\cite{Spackman2019}. Lambert Meertens is the first to mention this hierarchy in the literature~\cite{Meertens1986}, citing an unpublished working paper by Boom~\cite{Boom1981}.
\end{rem}

\begin{center}
\begin{table}[h]
  \caption{The Boom hierarchy}\label{Boom1}
  \begin{minipage}{\textwidth}
    \begin{center}
    \begin{tabular}{lcccc}
    \toprule
    theory \;\; & unit\;\; & associative \;\; & commutative \;\; & idempotent \\ \midrule
    tree \;\; & \yes  & \no & \no & \no \\
    list \;\; & \yes & \yes & \no & \no \\
    multiset \;\; & \yes & \yes & \yes & \no \\
    powerset \;\; & \yes & \yes & \yes & \yes \\\bottomrule
    \end{tabular}
    \end{center}
  \end{minipage}
\end{table}
\end{center}
Studying the patterns of distributive laws in the Boom hierarchy provided some of the original inspiration for the abstract no-go theorems presented in earlier sections.
We now use the same hierarchy, and generalizations of it, to demonstrate both their scope and limitations, and their relationship to the existing positive results we are aware of in the literature.

\subsection{The Original Boom Hierarchy}%
\label{sec:original-boom}
For the original Boom hierarchy we have complete knowledge of possible compositions via distributive laws. An overview is presented in \reftab~\ref{tab:distlawsBoom1} below.

The negative result for $P \circ P \Rightarrow P \circ P$ was already shown by Klin and Salamanca~\cite{KlinSalamanca2018}, and can also be recovered from both~\refthms~\ref{thm:plotkin-1} and~\ref{thm:nogoTreeIdem}. The other negative results follow from either \refthm~\ref{thm:nogoTreeList} or \refthm~\ref{thm:nogoTreeIdem}. Sometimes both theorems can be applied, for example to preclude a distributive law~$P \circ L \Rightarrow L \circ P$.

The positive results are due to Manes and Mulry. They show that any monad with only linear equations in its presentation distributes over any commutative monad\footnote{A commutative monad is a strong monad for which the two possible double strengths coincide~\cite{Kock1970}. Algebraically, this means that all operations in the signature commute with one another. This is quite different from the algebraic property of commutativity that we consider in the Boom hierarchy.} via the times over plus distributivity~\cite[Theorem 4.3.4]{ManesMulry2007}. The multiset and powerset monad are both commutative, and the theories of the tree, list and multiset monads all have solely linear equations.  This yields the six distributive laws indicated in the table. \refthm~\ref{thm:times-over-plus} proves that these distributive laws are in fact the only possible distributive laws for these monads.
\begin{center}
\begin{table}[h]
  \centering
  \caption{Possible compositions in the Boom hierarchy,}\label{tab:distlawsBoom1}
  with distributive laws of type: row $\circ$ column $\Rightarrow$ column $\circ$ row.
  \begin{minipage}{\textwidth}
    \begin{center}
    \belowrulesep=0ex
    \aboverulesep=0ex
    \renewcommand{\arraystretch}{1.1}
    \begin{tabular}{l|cccc} 
    \toprule
     & tree\;\; & list \;\; & multiset \;\; & powerset \\ \midrule
    tree \;\;     & \no & \no & \yes & \yes \\
    list \;\;     & \no & \no & \yes & \yes \\
    multiset \;\; & \no & \no & \yes & \yes \\
    powerset \;\; & \no & \no & \no & \no \\\bottomrule
    \end{tabular}
    \end{center}
  \end{minipage}
\end{table}
\end{center}

In this small sample, Manes and Mulry's theorem is powerful enough to yield all the possible positive results. However, \reftab~\ref{tab:distlawsBoom1} is too small to draw general conclusions. To extract more information, we expand our hierarchy of monads in the next section.

\subsection{The Extended Boom Hierarchy}%
\label{sec:extended-boom}
The original Boom hierarchy discussed in~\refsec~\ref{sec:original-boom} consisted of a small number of well motivated data structures. Unfortunately, this small size means that it provides limited scope for identifying patterns in distributive law phenomena. To address this, we now consider an extended hierarchy in which all possible combinations of the original algebraic axioms appear, rather than the axioms being gradually added in a fixed order.
This yields a total of eight different monads: tree, idempotent tree, commutative tree (mobile), associative tree (list), idempotent and commutative tree, idempotent and associative tree (square-free list), associative and commutative tree (multiset), idempotent commutative and associative tree (powerset). If we additionally consider the non-empty versions of these monads, corresponding algebraically to removing the constant from the signature, the number of monads doubles to sixteen. Of course, some of these monads are less natural from a functional programming perspective, but they provide a convenient range of candidates for investigation.

An overview of these monads is given in \reftab~\ref{Boom2} below. We have named the monads according to the axioms their theories satisfy: $U$(unitality), $A$(associativity), $C$(commutativity), $I$(idempotence). For example, multisets are associative and commutative trees with units, so they are denoted $UAC$ in the table, whereas their non-empty version, which has no unit, is called $AC$. Note that this convention is unambiguous as we always impose the unitality axiom when the unit constant is present.

\begin{center}
\begin{table}[h]
  \caption{The extended Boom hierarchy}\label{Boom2}
  \begin{minipage}{\textwidth}
    \begin{center}
    \begin{tabular}{lcccc}
    \toprule
    theory \;\; & unital\;\; & associative \;\; & commutative \;\; & idempotent \\ \midrule
    non-empty tree \;\; & \no & \no & \no & \no \\
    I \;\;           & \no & \no & \no & \yes \\
    C \;\;           & \no & \no & \yes & \no \\
    CI \;\;           & \no & \no & \yes & \yes \\
    A \;\;           & \no & \yes & \no & \no \\
    AI \;\;           & \no & \yes & \no & \yes \\
    AC \;\;           & \no & \yes & \yes & \no \\
    ACI \;\;           & \no & \yes & \yes & \yes \\
    U (tree) \;\; & \yes  & \no & \no & \no \\
    UI \;\;           & \yes & \no & \no & \yes \\
    UC (mobile) \;\;           & \yes & \no & \yes & \no \\
    UCI \;\;           & \yes & \no & \yes & \yes \\
    UA (list) \;\; & \yes & \yes & \no & \no \\
    UAI (square-free list) \;\;           & \yes & \yes & \no & \yes \\
    UAC (multiset) \;\; & \yes & \yes & \yes & \no \\
    UACI (powerset)\;\; & \yes & \yes & \yes & \yes \\\bottomrule
    \end{tabular}
    \end{center}
  \end{minipage}
\end{table}
\end{center}
Within the extended Boom hierarchy, there are a total of 256 monad compositions to consider. Some distributive laws arise via Manes and Mulry's positive general theorems~\cite[Theorem 4.3.4]{ManesMulry2007} and~\cite[Example 4.9]{ManesMulry2008}. Other combinations are known to have a distributive law because an ad-hoc one has been found, for example for the non-empty list monad over itself~\cite{ManesMulry2007, ManesMulry2008}. A large number of the combinations are proven impossible by theorems from this paper. Our current knowledge about the existence of distributive laws for this extended Boom hierarchy is given in~\reftab~\ref{tab:DistBoomExtended}.
\begin{center}
\begin{table}[h]
  \caption{The extended Boom hierarchy, with laws of type: row $\circ$ column $\Rightarrow$ column $\circ$ row.}\label{tab:DistBoomExtended}
  \begin{minipage}{\textwidth}
    \begin{center}
    {\small
    \belowrulesep=0ex
    \aboverulesep=0ex
    \renewcommand{\arraystretch}{1.1}
    \begin{tabular}{l|cccccccc|cccccccc} 
    \toprule
     & $\emptyset$ & I & C & CI & A & AI & AC & ACI & U & UI & UC & UCI & UA & UAI & UAC & UACI \\ \midrule
    $\emptyset$ & \yes & \maybe & \maybe  & \maybe & \maybe & \maybe & \yes & \yes & \maybe & \maybe & \maybe & \maybe & \maybe & \maybe & \yes & \yes \\
    I          & \maybe & \no & \maybe & \no & \maybe & \no & \maybe & \no & \maybe & \no & \no & \no & \maybe & \no & \no & \no \\
    C          & \yes & \maybe & \maybe & \maybe & \maybe & \maybe & \yes & \yes & \maybe & \maybe & \maybe & \maybe & \maybe & \maybe & \yes & \yes \\
    CI         & \maybe & \no & \maybe & \no & \maybe & \no & \maybe & \no & \maybe & \no & \no & \no & \maybe & \no & \no & \no \\
    A          & \yes & \maybe & \maybe & \maybe & \yes & \maybe & \yes & \yes & \maybe & \maybe & \maybe & \maybe & \maybe & \maybe & \yes & \yes \\
    AI         & \maybe & \no & \maybe & \no & \maybe & \no & \maybe & \no & \maybe & \no & \no & \no & \maybe & \no & \no & \no\\
    AC         & \yes & \maybe & \maybe & \maybe & \maybe & \maybe & \yes & \yes & \maybe & \maybe & \maybe & \maybe & \maybe & \maybe & \yes & \yes \\
    ACI        & \maybe & \no & \maybe & \no & \maybe & \no & \maybe & \no & \maybe & \no & \no & \no & \maybe & \no & \no & \no\\ \midrule
    U          & \yes  & \maybe & \maybe & \maybe & \maybe & \maybe & \yes & \yes & \no & \no & \no & \no & \no & \no & \yes & \yes \\
    UI         & \maybe & \no & \maybe & \no & \maybe & \no & \maybe & \no & \no & \no & \no & \no & \no & \no & \no & \no \\
    UC         & \yes  & \maybe & \maybe & \maybe & \maybe & \maybe & \yes & \yes & \no & \no & \no & \no & \no & \no & \yes & \yes \\
    UCI        & \maybe & \no & \maybe & \no & \maybe & \no & \maybe & \no & \no & \no & \no & \no & \no & \no & \no & \no \\
    UA         & \yes  & \maybe & \maybe & \maybe & \maybe & \maybe & \yes & \yes & \no & \no & \no & \no & \no & \no & \yes & \yes \\
    UAI        & \maybe & \no & \maybe & \no & \maybe & \no & \maybe & \no & \no & \no & \no & \no & \no & \no & \no & \no \\
    UAC        & \yes  & \maybe & \maybe & \maybe & \maybe & \maybe & \yes & \yes & \no & \no & \no & \no & \no & \no & \yes & \yes \\
    UACI       & \maybe & \no & \maybe & \no & \maybe & \no & \maybe & \no & \no & \no & \no & \no & \no & \no & \no & \no \\\bottomrule
    \end{tabular}
    } \end{center}
  \end{minipage}
\end{table}
\end{center}
Out of the 256 pairs of monads in \reftab~\ref{tab:DistBoomExtended}:
\begin{itemize}
\item The 41 labelled with~\yes have a distributive law between them.
\item The 122 labelled with~\no do not have a distributive law between them.
\item The remaining 93 pairs labelled with~\maybe remain to be understood.
\end{itemize}
That is, almost half of the combinations of monads from the extended Boom hierarchy do not have a distributive law between them. The bottom right corner of \reftab~\ref{tab:DistBoomExtended} is especially striking. Every possible combination is understood, and 56 out of 64 combinations do not have a distributive law. This provides further evidence that we should not assume ``most'' monads will compose via suitable distributive laws.

An unsurprising pattern that emerges from \reftab~\ref{tab:DistBoomExtended} is that the axioms of idempotence and units are `bad' properties for monad compositions. Since all of our no-go theorems require at least one of these properties to hold, this observation does not lead to any new insights.

In the positive results, the most apparent patterns are the columns $\emptyset$, $AC$, $ACI$, $UAC$, and~$UACI$. These are precisely the cases captured by Manes and Mulry.

For the remaining open cases, we cannot make any meaningful predictions. Our current techniques in no-go theorems require some way of bringing a term down to a variable, either via idempotence or via units. Whether this is the key property that prevents the existence of distributive laws remains an open question. On the other hand, all general positive results rely on one of the monads being commutative. The ad-hoc distributive law of the non-empty list monad over itself clearly indicates that commutativity is not a necessary condition for distributive laws to exist.

One thing is certain: to systematically fill in the gaps in \reftab~\ref{tab:DistBoomExtended} we will need additional ideas, supporting either further positive, or negative, theorems. As such, analyzing these hierarchies highlight directions which will deepen our understanding of distributive laws.

\subsection{Repeated Compositions}
Another way to extend the original Boom hierarchy is to add the monads resulting from the six distributive laws appearing in~\reftab~\ref{tab:distlawsBoom1}.
In order to study these additional monads, we first give concrete presentations for each of them:
\begin{lem}
 The presentations of the composite monads $MT$, $ML$, $MM$, $PT$, $PL$, $PM$ are as follows:
  \begin{itemize}
  \item The composite monad $MT$ is presented by the following theory:
    \begin{itemize}
    \item Signature: $\Sigma^{MT} = \{0^{(0)},1^{(0)},+^{(2)},*^{(2)}\}$.
    \item Equations: $0$ is the unit of $+$, $+$ is associative and commutative, $1$ is the unit of $*$, $0$ is a multiplicative zero:
        \begin{align}
        0 * x & = 0\label{eq:zero1} \\
        x * 0 & = 0\label{eq:zero2},
        \end{align}
        and $*$ distributes over $+$ from both left and right:
        \begin{align}
        x * (y + z) & = (x * y) + (x * z)\label{eq:timesoverplus1} \\
        (x + y) * z & = (x * z) + (y * z)\label{eq:timesoverplus2}.
        \end{align}
    \end{itemize}
  \item $ML$ has the same signature and equations as $MT$, with the additional equation that $*$ is associative.
  \item $MM$ also has the same signature and equations as $MT$, with $*$ additionally being associative and commutative.
  \item $PT$, $PL$, and $PM$ have the same signatures and equations as $MT$, $ML$, and $MM$ respectively, with one additional equation: $+$ is idempotent.
  \end{itemize}
\end{lem}

\begin{proof}
We prove only that the given presentation for the monad $MT$ is correct, the others follow similarly.
From \refcor~\ref{cor:presentationforcomposite}, we know that the composite monad $MT$ has presentation:

    \begin{itemize}
    \item Signature: $\Sigma^{MT} = \Sigma^M \uplus \Sigma^T = \{0^{(0)},1^{(0)},+^{(2)},*^{(2)}\}$.
    \item Equations: $E^{MT} = E^M \cup E^T \cup E^{\lambda}$ = $\{ 0 $ is the unit of $+$, $+$ is associative and commutative$\}$ $\cup$ $\{ 1$ is the unit of $* \}$ $\cup$ $\{ s[t_x/x] \theoryeq{MT} t[s_y/y] \;|\; s, s_y$ terms in $\bb{M}$, $t, t_x$ terms in $\bb{T} \}$.
   \end{itemize}
So all we need to show is that \refeqns~\eqref{eq:zero1},~\eqref{eq:zero2},~\eqref{eq:timesoverplus1}, and~\eqref{eq:timesoverplus2} are enough to prove all equations in $E^{\lambda}$, that is, of form $s[t_x/x] \theoryeq{MT} t[s_y/y]$.


To prove that all equations in $E^{\lambda}$ are provable from the four given axioms, it is enough to show that each term of form $s[t_x/x]$ is separable via the given axioms. Essential uniqueness then does the rest: Suppose that both $s[t_x/x] \theoryeq{MT} t[s_y/y]$ and $s[t_x/x] \theoryeq{MT} t'[s'_{y'}/{y'}]$ are in $E^{\lambda}$, and the first of these equations is provable using \refeqns~\eqref{eq:zero1},~\eqref{eq:zero2},~\eqref{eq:timesoverplus1}, and~\eqref{eq:timesoverplus2}. Then by transitivity of $\theoryeq{MT}$ we know $t'[s'_{y'}/{y'}] \theoryeq{MT} t[s_y/y]$. Since this is an equality between two separated terms, essential uniqueness gives us that it is provable using just the axioms in $E^M$ and $E^T$. So once we have derived one equation in $E^{\lambda}$ involving a particular term $s[t_x/x]$, we automatically gain all equations in $E^{\lambda}$ involving this term.

To prove that every term of form $s[t_x/x]$ is separable using just the axioms \refeqns~\eqref{eq:zero1},~\eqref{eq:zero2},~\eqref{eq:timesoverplus1}, and~\eqref{eq:timesoverplus2}, we use induction on the term complexity of $s$:

Base cases: $s$ is a constant or a variable. In these cases $s[t_x/x]$ is already separated.
Induction step: Suppose that $s = s_1 * s_2$, and assume that $s_1[t_x/x]$ and $s_2[t_x/x]$ are both separable using the four equations described above. Then we need to show that $s[t_x/x]$ is separable. Let $t'_1[s'_{y'}/y']$ and $t'_2[s'_{y'}/y']$ be the terms resulting from separating $s_1$ and $s_2$ respectively. We need induction on both $t'_1$ and $t'_2$:
\begin{itemize}
  \item If $t'_1$ is a constant, $t'_1 = 0$, then by \refeqn~\eqref{eq:zero1}: $s[t_x/x] = 0 * s_2 = 0$, and hence $s[t_x/x]$ is separable.
  \item If $t'_1$ is a variable, $t'_1 = x_1$, we use induction $t'_2$:
   \begin{itemize}
   \item If $t'_2$ is a constant, then by \refeqn~\eqref{eq:zero2}: $s[t_x/x] = s_1 * 0 = 0$, and hence $s[t_x/x]$ is separable.
   \item If $t'_2$ is a variable, $t'_2 = x_2$, then $s[t_x/x] = x_1 * x_2$, which is a separated term.
   \item If $t'_2$ is term of form $t'_3 + t'_4$, and we assume (induction hypothesis) that $x_1 * t'_3$ is separable and $x_1 * t'_4$ is separable, then $s[t_x/x] = x_1 * (t'_3 + t'_4)$. By \refeqn~\eqref{eq:timesoverplus1}, we can write:
       \begin{align*}
       & s[t_x/x] \\
       = \; &  x_1 * (t'_3 + t'_4) \\
       = \; & (x_1 * t'_3) + (x_1 * t'_4).
       \end{align*}
       Since both $(x_1 * t'_3)$ and $(x_1 * t'_4)$ are separable, this shows that $s[t_x/x]$ is separable.
    \end{itemize}
    \item If $t'_1$ is a term of form $t'_5 + t'_6$, we may assume (induction hypothesis) that $t'_5 * t'_2$ and $t'_6 * t'_2$ are separable. By \refeqn~\eqref{eq:timesoverplus2}, we know:
       \begin{align*}
       & s[t_x/x] \\
       = \; &  (t'_5 + t'_6) * t'_2 \\
       = \; & (t'_5 * t'_2) + (t'_6 * t'_2).
       \end{align*}
       Since both $(t'_5 * t'_2)$ and $(t'_6 * t'_2)$ are separable, this shows that $s[t_x/x]$ is separable.
\end{itemize}
We conclude that the given presentation is indeed a presentation for the monad $MT$.
\end{proof}
Checking these presentations against the various no-go theorems presented in this paper yields a new table of possible compositions, displayed in \reftab~\ref{tab:distlawsBoom2}.
Other than the six combinations we already discovered in the original Boom hierarchy, we find no new combinations of monads that compose via a distributive law.  The results shown in the columns of the composite monads $MT, ML, MM, PT, PL, PM$ are covered by \refthm~\ref{thm:nogoConstants}, while the rows with any of these composite monads are covered by \refthm~\ref{theoremwithconst}.
\begin{table}[h]
  \centering
  \caption{Possible compositions in the Boom hierarchy,}\label{tab:distlawsBoom2}
  with distributive laws of type: row $\circ$ column $\Rightarrow$ column $\circ$ row.
  \begin{minipage}{\textwidth}
    \begin{center}
    \aboverulesep=0ex
    \belowrulesep=0ex
    \renewcommand{\arraystretch}{1.1}
    \begin{tabular}{l|cccc|cccccc} 
    \toprule
     & tree\;\; & list \;\; & multiset \;\; & powerset & MT & ML & MM & PT & PL & PM \\ \midrule
    tree \;\;     & \no & \no & \yes & \yes & \no & \no & \no & \no & \no & \no\\
    list \;\;     & \no & \no & \yes & \yes & \no & \no & \no & \no & \no & \no\\
    multiset \;\; & \no & \no & \yes & \yes & \no & \no & \no & \no & \no & \no\\
    powerset \;\; & \no & \no & \no & \no & \no & \no & \no & \no & \no & \no\\ \midrule
    MT & \no & \no & \no & \no & \no & \no & \no & \no & \no & \no \\
    ML & \no & \no & \no & \no & \no & \no & \no & \no & \no & \no \\
    MM & \no & \no & \no & \no & \no & \no & \no & \no & \no & \no \\
    PT & \no & \no & \no & \no & \no & \no & \no & \no & \no & \no \\
    PL & \no & \no & \no & \no & \no & \no & \no & \no & \no & \no \\
    PM & \no & \no & \no & \no & \no & \no & \no & \no & \no & \no \\\bottomrule
    \end{tabular}
    \end{center}
  \end{minipage}
\end{table}

\section{Conclusion}\label{sec:conclusion}
We have shown there can be no distributive law between large classes of monads:
\begin{itemize}
\item Section~\ref{section:Plotkin} developed general theorems for demonstrating when distributive laws cannot exist, derived from a classical counterexample of Plotkin, making essential use of idempotence axioms.
\item Section~\ref{section:beyondPlotkin} introduced novel ideas, exploiting unitality axioms, yielding counterexamples beyond those possible in \refsec~\ref{section:Plotkin}.
\item Section~\ref{sec:plus-over-times} developed a third family of no-go theorems, motivated by a question of Beck. Here the ability to ``make variables disappear'' in terms is central, such as in inverse and absorption axioms.
\end{itemize}
Our results cover many naturally occurring combinations of monads, including all previously known negative results. They also identify issues in the existing literature, resolve the open question of whether the list monad distributes over itself,
and confirm a negative conjecture of Beck.

We strongly advocated the use of algebraic methods. These techniques were used for all of our proofs. Taking this approach, rather than direct calculations involving Beck's axioms, enabled us to single out the essentials of each proof, so that the resulting theorems could be stated in full generality.

Lastly, we would like to emphasize that the methods described in this paper are of broader application, beyond our specific theorems. For example, Julian Salamanca found that there is no distributive law of the group monad over the powerset monad $G \circ P \Rightarrow P \circ G$~\cite{Salamanca2018}, before we formulated \refthm~\ref{theoremwithconst}. Although elements of his proof are very similar to those used in Section~\refsec~\ref{section:Plotkin}, the proof itself uses further ideas outside the scope of our theorems.

\subsection{Summary of Axioms and Theorems}
Our theorems use a range of different properties of algebraic theories. To aid navigation, we conclude with a summary of the properties that we exploit, and which theorems exploit them.

Our theorems typically require the existence of a special term with multiple properties, such as a binary term that is both idempotent and commutative. Below we list all these properties separately. Table~\ref{overview} then specifies which theorems require which properties to hold simultaneously for a single term. Note that some properties, such as~\ref{endax:general-idemp-exists1}, seem trivial on their own, but in conjunction with other properties we require of the special term in our theorems they are no longer trivial and must be stated as a requirement.

In addition, our theorems require certain properties to hold for the theories in question, such as the existence of a constant in the theory, or conditions that must hold for all terms in the theory, such as \emph{``if a term is equal to a constant, then it cannot have any variables''}. In the list below, we split the properties into those that must hold for one special term, and those that must hold for the theory.

\begin{itemize}[label=$ $]
\item {\bf Properties of terms:}
\begin{itemize}[label=$\bullet$]
\item There is a binary term $\anyterm$ such that:
\begin{enumerate}[label=Ax\arabic*]
  \item\label{endax:binary-idemp} $\anyterm$ is idempotent:
        \[
          \{x\} \vdash \anyterm(x,x) = x
        \]
  \item\label{endax:binary-unit} $\anyterm$ has a unit:
  \[
    \{x\} \vdash \anyterm(x,e) = x = \anyterm(e,x)
  \]
  \item\label{endax:general-idemp-exists-binary} There is a substitution $f$ of terms for variables such that $\anyterm[f]$ is equal to a variable:
  \[
    \exists f~s.t.~\Gamma \vdash \anyterm[f] = x
  \]
  \item\label{endax:binary-comm} $\anyterm$ is commutative:
  \[
    \{x,y\} \vdash \anyterm(x,y) = \anyterm(y,x)
  \]
  \item\label{endax:no-abides} $\anyterm$ does not satisfy the abides equation:
  \[
    \left(\Gamma  \vdash \anyterm(\anyterm(x,y),\anyterm(z,w)) = \anyterm(\anyterm(x,z),\anyterm(y,w))\right) \;\;\Rightarrow\;\; \#\{x,y,z,w\} \leq 3
  \]
  \item\label{endax:bounded-var-binary} Any term provably equal to $\anyterm$ has at most two variables:
  \[
    \forall \anytermtwo : (\Gamma \vdash \anyterm(x,y) = \anytermtwo) \;\;\Rightarrow\;\; \{x,y\} \vdash \anytermtwo
  \]
  \item\label{endax:at-least-one-var-binary} Any term provably equal to $\anyterm$ has at least one variable:
  \[
    \forall \anytermtwo : (\Gamma \vdash \anyterm(x,y) = \anytermtwo) \;\;\Rightarrow\;\; \neg (\emptyset \vdash \anytermtwo)
  \]
  \item\label{endax:at-least-two-var-binary} Any term provably equal to $\anyterm$ has at least two variables:
  \[
    \forall \anytermtwo : (\Gamma \vdash \anyterm(x,y) = \anytermtwo) \;\;\Rightarrow\;\; \neg (\{x\} \vdash \anytermtwo \;\vee\; \{y\} \vdash \anytermtwo)
  \]
\end{enumerate}
\item There is an $n$-ary term $\anyterm$ ($n \geq 2$) such that:
\begin{enumerate}[label=Ax\arabic*] \setcounter{enumi}{8}
  \item\label{endax:nary-idemp} $\anyterm$ is idempotent:
         \[
           \{x\} \vdash \anyterm[x/x_i] = x
          \]
  \item\label{endax:nary-unit} $\anyterm$ has a unit:
  \[
    \{x_i\} \vdash \anyterm[e/x_{j\neq i}] = x_i
  \]
  \item\label{endax:general-idemp-exists} There is a substitution $f$ of terms for variables such that $\anyterm[f]$ is equal to a variable, and none of the terms in the range of $f$ contain that variable:
  \[
    \exists f~s.t.~\Gamma \vdash \anyterm[f] = x \; \wedge \; \forall \anytermtwo \in \ran(f) : \neg (x \in \var(\anytermtwo))
  \]
  \item\label{endax:nary-comm} $\anyterm$ is stable under a fixed-point free permutation $\sigma$ of its variables:
  \[
    \Gamma \vdash \anyterm = \anyterm[\sigma]
  \]
  \item\label{endax:bounded-var-nary} Any term provably equal to $\anyterm$ has at most $n$ variables:
  \[
    \forall \anytermtwo : (\Gamma \vdash \anyterm(x_1,\ldots,x_n) = \anytermtwo) \;\;\Rightarrow\;\; \{x_1,\ldots,x_n\} \vdash \anytermtwo
  \]
  \item\label{endax:at-least-two-var-nary} Any term provably equal to $\anyterm$ has at least two variables:
  \[
    \forall \anytermtwo : (\Gamma \vdash \anyterm = \anytermtwo) \;\;\Rightarrow\;\; \neg (\bigvee_{x \in \Gamma} \{x\} \vdash \anytermtwo)
  \]
\end{enumerate}
\item There is an $n$-ary term $\anyterm$ ($n \geq 1$) such that:
\begin{enumerate}[label=Ax\arabic*] \setcounter{enumi}{14}
  \item\label{endax:general-idemp-exists1} There is a substitution $f$ of terms for variables such that $\anyterm[f]$ is equal to a variable, and none of the terms in the range of $f$ contain that variable:
  \[
    \exists f~s.t.~\Gamma \vdash \anyterm[f] = x \; \wedge \; \forall \anytermtwo \in \ran(f) : \neg (x \in \var(\anytermtwo))
  \]
  \item\label{endax:thedifficultone} $\anyterm$ can be written as $\anytermtwo[\anytermthree_i/x_i]$, such that $\anyterm$ and $\anytermtwo$ share at least one variable, and~\ref{endax:general-idemp-exists1} holds for $\anytermtwo$.
  \item\label{endax:missing-var-const} $\anyterm$ is provably equal to a constant, but contains at least one variable:
  \[
    \exists e~s.t.~(\emptyset \vdash e) \;\wedge\; (e = \anyterm) \;\wedge\; \neg (\emptyset \vdash \anyterm)
  \]
  \item\label{endax:missing-var-terms} There is a term $\anytermtwo$ that is provably equal to $\anyterm$, and $\anyterm$ contains a variable that $\anytermtwo$ does not have. Moreover, $\anytermtwo$ satisfies~\ref{endax:general-idemp-exists1}.
  \[
    \exists \anytermtwo~s.t.~(\Gamma \vdash \anyterm = \anytermtwo) \;\wedge\; \var(\anyterm)\setminus\var(\anytermtwo) \neq \emptyset
  \]
\end{enumerate}
\end{itemize}
\item {\bf Properties of theories:}
\begin{itemize}[label=$\bullet$]
\item The algebraic theory in question must satisfy:
\begin{enumerate}[label=Ax\arabic*] \setcounter{enumi}{18}
\item\label{endax:general-idemp-forall} For all terms $\anyterm$ of arity $\geq 1$ there is a substitution $f$ of terms for variables such that $\anyterm[f]$ is equal to a variable, and none of the terms in the range of $f$ contain that variable:
    \[
      \forall \anyterm~\exists f~s.t.~\Gamma \vdash \anyterm[f] = x  \; \wedge \; \forall \anytermtwo \in \ran(f) : \neg (x \in \var(\anytermtwo))
    \]
\item\label{endax:constant} There exists a constant term in the theory. That is, there is a term that has no variables.
\[
\exists e~s.t.~\emptyset \vdash e
\]
\item\label{endax:twoconstants} There are at least two constants.
\item\label{endax:weakno-var} If any variable substitution of a term is equal to a constant, then that term itself is equal to a constant:
\[
\forall \anyterm, e : (\exists f~s.t.~\anyterm[f] = e) \;\;\Rightarrow\;\; \anyterm = e
\]
\item\label{endax:no-var} If a term is provably equal to a constant, then it does not have any variables:
\[
\forall \anyterm, e : (\emptyset \vdash e \wedge e = \anyterm) \;\;\Rightarrow\;\; \emptyset \vdash \anyterm
\]
\item\label{endax:single-var} If a term is provably equal to a variable, then that variable is the only variable appearing in the term:
\[
\forall \anyterm, x : (\Gamma \vdash \anyterm = x)  \;\;\Rightarrow\;\; \{x\} \vdash \anyterm
\]
\end{enumerate}
\end{itemize}
\end{itemize}

\noindent
An overview of our theorems and which axioms they use is presented in~\reftab~\ref{overview}. 
We advise to use this table only as a quick reference, and to always check the precise statements of the theorems before drawing any conclusions about specific monads.
\begin{table}[h!]
\caption{Theorems and axioms appearing in this paper.}%
  \label{overview}

\bigskip
\begin{tabularx}{\textwidth}{Xll}
  Theorem & Axioms for $S$ & Axioms for $T$ \\ \toprule
  \refthm~\ref{thm:plotkin-1} \newline (first Plotkin generalisation) & \ref{endax:binary-idemp},~\ref{endax:at-least-two-var-binary},~\ref{endax:single-var} & \ref{endax:binary-idemp},~\ref{endax:binary-comm},~\ref{endax:bounded-var-binary}  \\\midrule 
  \refthm~\ref{thm:plotkin-2} \newline (n-ary version) & \ref{endax:binary-idemp},~\ref{endax:at-least-two-var-nary},~\ref{endax:single-var} & \ref{endax:nary-idemp},~\ref{endax:nary-comm},~\ref{endax:bounded-var-nary}  \\\midrule 
  \refthm~\ref{thm:plotkin-3} \newline (no commutativity) & \ref{endax:binary-idemp},~\ref{endax:at-least-two-var-binary},~\ref{endax:single-var} & \ref{endax:binary-idemp},~\ref{endax:bounded-var-binary},~\ref{endax:at-least-one-var-binary},~\ref{endax:single-var} \\\midrule 
  \refthm~\ref{thm:NoGo-Videm-Punit} \newline (generalising idempotence) & \ref{endax:binary-idemp},~\ref{endax:at-least-two-var-binary},~\ref{endax:single-var}& \ref{endax:general-idemp-exists-binary},~\ref{endax:binary-comm},~\ref{endax:bounded-var-binary} \\\midrule 
  \refthm~\ref{thm:nogoConstants} \newline (too many constants) & \ref{endax:general-idemp-exists} & \ref{endax:twoconstants},~\ref{endax:weakno-var}\\\midrule 
  \refthm~\ref{thm:times-over-plus} \newline (times over plus) & \ref{endax:binary-unit},~\ref{endax:general-idemp-forall},~\ref{endax:no-var},~\ref{endax:single-var} & \ref{endax:binary-unit},~\ref{endax:no-var},~\ref{endax:single-var}  \\\midrule 
  \refthm~\ref{thm:nogoTreeList} \newline (lacking abides) & \ref{endax:binary-unit},~\ref{endax:general-idemp-forall},~\ref{endax:no-var},~\ref{endax:single-var} & \ref{endax:binary-unit},~\ref{endax:no-abides},~\ref{endax:no-var},~\ref{endax:single-var} \\\midrule 
  \refthm~\ref{thm:nogoTreeIdem} \newline (idempotence and units) & \ref{endax:binary-idemp},~\ref{endax:binary-unit},~\ref{endax:general-idemp-forall},~\ref{endax:no-var},~\ref{endax:single-var} & \ref{endax:binary-unit},~\ref{endax:no-var},~\ref{endax:single-var} \\\midrule 
  \refthm~\ref{theoremwithconst} \newline (inverse trouble) & \ref{endax:nary-unit},~\ref{endax:thedifficultone},~\ref{endax:missing-var-const}   & \ref{endax:constant},~\ref{endax:weakno-var} \\\midrule 
  \refthm~\ref{theoremwithvar} \newline (absorption trouble) & \ref{endax:general-idemp-exists1},~\ref{endax:missing-var-terms}  & \ref{endax:constant},~\ref{endax:weakno-var} \\ 
\end{tabularx}
\end{table}

\section*{Acknowledgements}

We are very grateful to Jeremy Gibbons, Bartek Klin, Hector Miller-Bakewell, Julian Salamanca and our reviewers for their insightful feedback on earlier versions of this paper. We also would like to thank Prakash Panangaden, Fredrik Dahlqvist, Louis Parlant, Sam Staton, Ohad Kammar, Jules Hedges, and Ralf Hinze for inspiring discussions, which all had a positive influence on this paper.

This work was partially supported by Institute for Information \& communications Technology Promotion(IITP) grant funded by the Korea government(MSIT) (No.2015--0--00565, Development of Vulnerability Discovery Technologies for IoT Software Security).


\bibliographystyle{alphaurl}
\bibliography{paper}

\end{document}